%% file: main.tex
\documentclass{article}

\PassOptionsToPackage{numbers}{natbib}
\usepackage[preprint]{neurips_2026}
\makeatletter
\renewcommand{\@noticestring}{Preprint. September 27, 2026.}
\makeatother

\usepackage[utf8]{inputenc} 
\usepackage[T1]{fontenc}    
\usepackage{hyperref}       
\usepackage{url}            
\usepackage{booktabs}       
\usepackage{amsfonts}       
\usepackage{microtype}     
\usepackage{xcolor}        
\usepackage{amsthm}
\usepackage{algorithm}
\usepackage{algorithmic}
\usepackage[most]{tcolorbox}
\usepackage{amssymb}
\usepackage{wrapfig}
\usepackage{graphicx}
\usepackage{subcaption}
\usepackage{etoc}
\usepackage{amsmath}

\definecolor{algblue}{RGB}{35,95,155}
\newcommand{\algcomment}[1]{\hfill{\color{gray}\(\triangleright\) #1}}

\newtcolorbox{remarkbox}{
  colback=blue!6,
  colframe=gray!40,
  boxrule=0.6pt,
  arc=2pt,
  left=6pt,right=6pt,top=4pt,bottom=4pt
}

\title{Geometry-Adaptive Mechanisms \\ for Private Synthetic Data}

\author{
  \parbox{\textwidth}{\centering
    \mbox{Raoof Zare Moayedi$^{1}$}\enspace
    \mbox{Amir R.~Asadi$^{2}$}\thanks{Work carried out while at the University of Cambridge.}\enspace
    \mbox{Mohammad Hossein Yassaee$^{1}$}\enspace
    \mbox{Gholamali Aminian$^{3}$}\\[0.4em]
    $^{1}$Sharif University of Technology \quad
    $^{2}$University of Birmingham \quad
    $^{3}$The Alan Turing Institute\\
    \texttt{raoofmoayedi2000@gmail.com, a.r.asadi@bham.ac.uk,  yassaee@sharif.edu, gaminian@turing.ac.uk
            }
  }
}
\newtheorem{theorem}{Theorem}
\newtheorem{lemma}{Lemma}
\newtheorem{proposition}{Proposition}
\newtheorem{corollary}{Corollary}
\newtheorem{assumption}{Assumption}

\newtheorem{remark}{Remark}

\usepackage[colorinlistoftodos]{todonotes}

\begin{document}

\maketitle

\begin{abstract}
    Generating differentially private synthetic data with meaningful Wasserstein utility guarantees is challenging in high dimensions. For datasets of size \(n\) on $[0,1]^d$ with $d\ge2$, existing pure \(\varepsilon\)-differentially private mechanisms achieve expected $1$-Wasserstein error of order $(\varepsilon n)^{-1/d}$, reflecting the curse of dimensionality. While this rate is optimal in the worst case, it can be overly pessimistic when the data are supported on a lower-dimensional set. We formalize this through a multiscale packing-growth dimension $k$, which captures the geometric complexity of the support via the growth of packing numbers across scales. We propose \emph{Adaptive Pruned-PMM}, a pure $\varepsilon$-differentially private mechanism that combines private depth selection with our pruned variant of the Private Measure Mechanism (PMM) of He et al.\ (2023). The mechanism supports deeper, geometry-adapted hierarchies with expected running time $O\!\left(d(n+d)\log(\varepsilon n)\right)$, which is near-linear in $n$ for fixed dimension and privacy budget. Under an external multiscale packing-growth condition with dimension $k$, we show that, for fixed positive privacy budgets and fixed geometry, the expected $1$-Wasserstein error is of order $(\varepsilon n)^{-1/k}$ for $k>1$ as $n$ grows. We also prove a lower bound under a corresponding internal packing-growth condition, showing that the exponent $1/k$ is sharp within this framework.
\end{abstract}

\section{Introduction}
Differential privacy provides a rigorous framework for releasing information
about a dataset while limiting the influence of any single data point
\citep{DworkMcSherryNissimSmith2006Calibrating,DworkRoth2014Algorithmic}.

A central goal within this framework is the private release of a synthetic dataset: an
artificial sample that can be shared, reused, and analyzed in place of the
original data. This has motivated a broad literature on
private synthetic data generation
\citep{
BoedihardjoStrohmerVershynin2022PrivateSampling,
BoedihardjoStrohmerVershynin2024PrivateMeasures,
He2023Algorithmically,Donhauser2024Certified,
HeStrohmerVershyninZhu2025LowDimensional}.

Much of the differential privacy literature focuses on specific learning,
statistical, or query-release tasks, including empirical risk minimization,
clustering, parameter estimation, stochastic gradient descent, deep learning,
counting queries, range queries, and marginals
\citep{BassilySmithThakurta2014PrivateERM,ChaudhuriMonteleoniSarwate2011DPERM,
SuCaoLiBertinoJin2016KMeans,DuchiJordanWainwright2018Local,
SongChaudhuriSarwate2013SGD,AbadiEtAl2016DeepLearning,
BlumDworkMcSherryNissim2005SULQ,HardtTalwar2010Geometry,
HardtLigettMcSherry2012DataRelease,BlumLigettRoth2013Learning,
UllmanVadhan2011Hardness,DworkNikolovTalwar2015Marginals,
Vietri2022PrivateSynthetic}. These approaches are powerful, but their
guarantees are typically tied to a prescribed task or query. 

Synthetic data offers a complementary route: once a differentially private
synthetic dataset is released, downstream analyses can be performed without
additional privacy loss
\citep{WassermanZhou2010Framework,BellovinDuttaReitinger2019Synthetic}.
Practical private synthetic-data methods include Bayesian-network approaches,
marginal-based systems, and generative modeling approaches
\citep{Zhang2017PrivBayes,McKennaMiklauSheldon2021Winning,
JordonYoonVanderSchaar2019PATEGAN}. However, many such methods measure utility
through histograms, marginals, or task-specific discrepancies, which may not
fully exploit the geometry of the data domain.

In this paper, we study private synthetic data under Wasserstein utility. Given
original data $X=(x_1,\dots,x_n)$ and synthetic data $Y=(y_1,\dots,y_m)$, let
$
\mu_X=\frac1n\sum_{i=1}^n\delta_{x_i},
\,\,
\mu_Y=\frac1m\sum_{i=1}^m\delta_{y_i}
$
be their empirical measures. We measure utility by the
\(1\)-Wasserstein distance \(W_1(\mu_X,\mu_Y)\). By Kantorovich--Rubinstein duality, controlling \(W_1(\mu_X,\mu_Y)\) gives
uniform control over all \(1\)-Lipschitz statistics, rather than a fixed query class. This is useful because many statistical and machine
learning procedures have Lipschitz or stability properties in metric spaces
\citep{VonLuxburgBousquet2004Lipschitz,BubeckSellke2021Robustness,
MeunierDelattreAraujoAllauzen2022LipschitzNN}.

A natural way to obtain Wasserstein-accurate synthetic data is to describe the
dataset at multiple geometric scales. Hierarchical and multiscale partitions
also appear in non-private approximation of distributions under the Wasserstein distance
\citep{BaNguyenNguyenRubinfeld2011EMD,WeedBach2019Wasserstein}. 

In the private
setting, the Private Measure Mechanism (PMM) of \citet{He2023Algorithmically}
builds a binary hierarchy of cells, perturbs the hierarchical counts, enforces
consistency across the tree, and then samples synthetic points from the leaves.

On the ambient cube $[0,1]^d$, the PMM of \citet{He2023Algorithmically} is $\varepsilon$-differentially private and achieves expected $1$-Wasserstein error of order $(\varepsilon n)^{-1/d}$ for $d\ge2$, reflecting the curse of dimensionality.

However, many data sets do not fill the ambient cube. They may concentrate near
a curve, a surface, a union of low-dimensional pieces, or another structured
support. In such cases, most cells in a fine partition are empty. The
standard PMM analysis is governed by the full hierarchy, and the usual
implementation traverses that hierarchy, even when the input dataset occupies only a small portion of it. This suggests that for such data, the relevant complexity should
be the multiscale growth of occupied cells, not the ambient dimension alone.

\subsection{Our Approach} 
We formalize this idea through a \textit{multiscale packing-growth} condition.
Informally, at scale \(2^{-j/d}\), the number of partition cells near the data grows at most like \(2^{jk/d}\). We call \(k\) the packing-growth dimension. The dimension \(k\in\{1,\dots,d\}\) summarizes the multiscale geometric
complexity of the support: when \(k\ll d\), the support occupies
far fewer cells across scales than the ambient cube.

Our starting point is a pruned variant of PMM, which we call
\emph{Pruned-PMM}. Like PMM, it perturbs hierarchical counts, enforces
consistency, and generates synthetic data from leaf cells. The tree is explored recursively: each noisy sibling pair is
corrected immediately, and only children with positive
corrected counts are expanded. Thus the mechanism 
follows the active part of
the hierarchy rather than materializing the full 
tree.

\begin{wrapfigure}{r}{0.48\textwidth}
\vspace{-1.2em}
\centering
\includegraphics[width=0.47\textwidth]{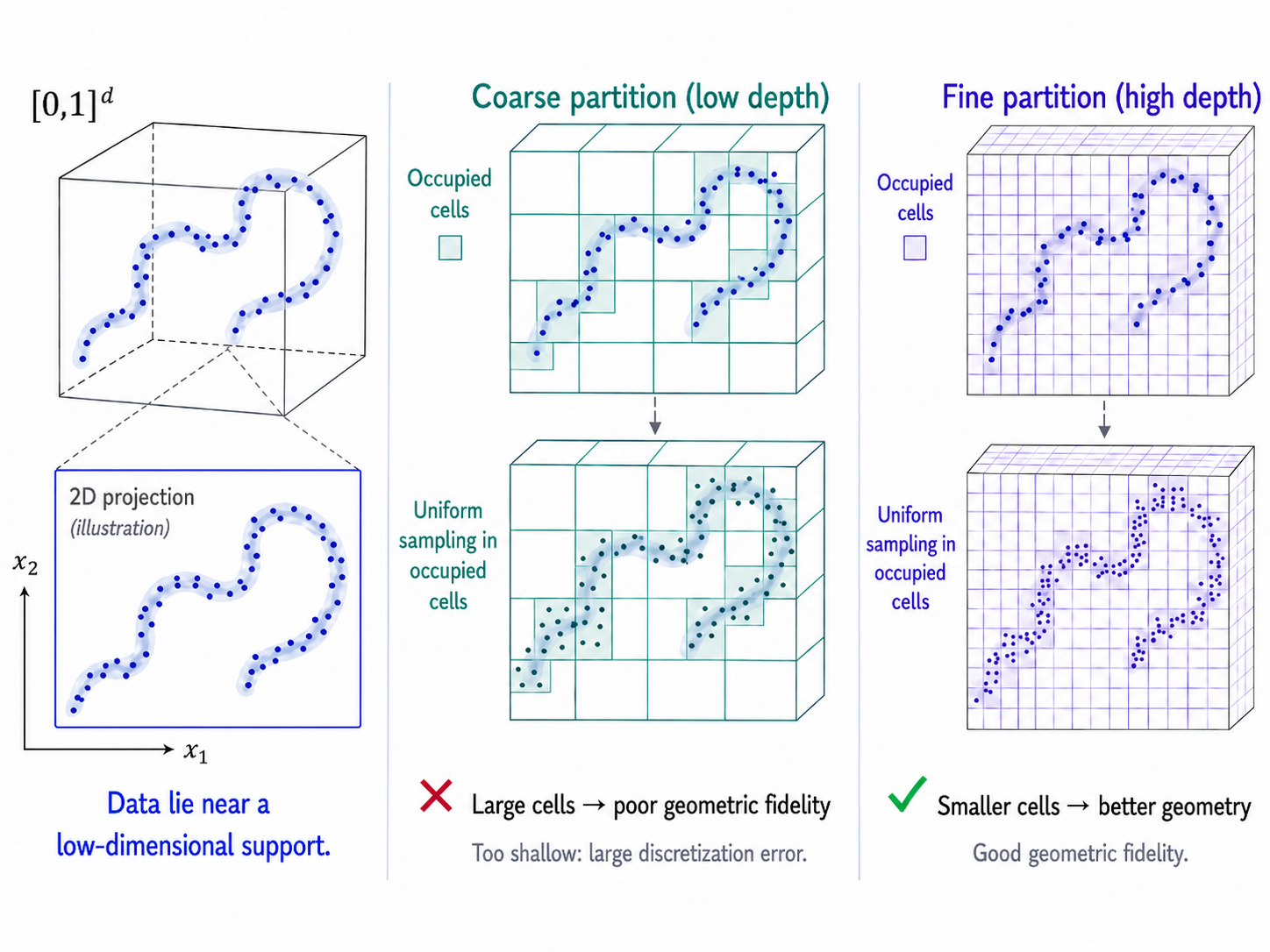}
\vspace{-0.8em}
\caption{
Coarse partitions lead to poor
geometric fidelity, finer partitions better track the support but introduce more noisy counts.
}
\label{fig:intro-depth-tradeoff}
\vspace{-1.2em}
\end{wrapfigure}
The depth of the hierarchy remains crucial, but the optimal depth depends on
the unknown multiscale geometry of the data. This depth sets the
resolution--noise tradeoff. Choosing it directly from the
occupied-cell profile would itself reveal information about the data. We therefore introduce \emph{Adaptive Pruned-PMM}, which combines Pruned-PMM with
a private depth-selection step. The mechanism fixes a finite public set \(\mathcal A\) of depth--noise schedules,
uses the exponential mechanism to select one element of \(\mathcal A\) according
to an occupancy-sensitive score, and then runs Pruned-PMM once with the selected
schedule. With this method, the algorithm adapts to the
multiscale geometry of the empirical support without requiring the packing-growth
dimension to be known in advance.

\paragraph{Main results (informal).}
Suppose the occupied hierarchy has packing-growth dimension $k$. Then, with a total privacy budget $\varepsilon$, we show that
Adaptive Pruned-PMM is $\varepsilon$-differentially private. For fixed positive privacy budgets and fixed geometry, it achieves, without knowing $k$,

$
\mathbb E W_1(\mu_X,\mu_Y)
\lesssim
(\varepsilon n)^{-1/k}
$ for $k>1$ as $n$ grows.
When $k=d$, this recovers the
ambient-dimensional PMM exponent, while when $k\ll d$, the rate improves drastically by
replacing the ambient dimension with the packing-growth dimension. The upper bound assumes that, at each tree level \(j\), the number of occupied
cells grows no faster than \(2^{jk/d}\). Conversely, our lower bound assumes that the support contains packings of size at least \(t^{-k}\) at small scales. Under this matching lower growth condition, we show that every $\varepsilon$-differentially private synthetic-data mechanism must incur Wasserstein error of order at least \((\varepsilon n)^{-1/k}\), up to the natural diameter cutoff. Thus, the exponent \(1/k\) is not an artifact of Pruned-PMM; it is forced by the geometry of the support.

\paragraph{Our contributions.}
\begin{enumerate}
    \item We introduce \emph{Pruned-PMM}, a zero-respecting sparse-tree variant of PMM that explores only active branches of the hierarchy, reducing runtime and memory relative to the full-tree mechanism. 

    \item We prove geometry-adaptive upper bounds for Pruned-PMM under an external multiscale packing-growth condition, showing that the ambient-dimensional dependence can be replaced by a bound driven by the packing-growth dimension \(k\).

    \item We introduce \emph{Adaptive Pruned-PMM}, which privately selects a depth--noise schedule from a public candidate set via the exponential mechanism, and we prove that its error is bounded by the best fixed-schedule guarantee in the candidate set plus an explicit private-selection penalty. Including private selection and synthesis, its expected running time
is $O\!\left(d(n+d)\log(\varepsilon n)\right)$.

    \item We prove a lower bound under a corresponding internal packing-growth condition, showing that the exponent \(1/k\) is sharp within this framework. For finite unions of \(k\)-dimensional coordinate cubes with \(k\ge2\), we establish that Pruned-PMM attains the minimax-optimal rate \((\varepsilon n)^{-1/k}\).
\end{enumerate}

\subsection{Related Work}
Several recent works study differentially private synthetic data with
geometric or Wasserstein-type utility, including private sampling, private
measures, and metric-geometry formulations of the
privacy--utility tradeoff
\citep{BoedihardjoStrohmerVershynin2023PrivacyStatistical,
BoedihardjoStrohmerVershynin2022PrivateSampling,
BoedihardjoStrohmerVershynin2024PrivateMeasures,
BoedihardjoStrohmerVershynin2024CovarianceLoss,
BoedihardjoStrohmerVershynin2024MetricGeometry}. Related works study private
release for smooth queries, Gibbs or exponential mechanism approaches, private
density estimation, and certified release for sparse Lipschitz functions
\citep{WangJinFanZhangHuangZhongWang2016SmoothQueries,
AsadiLoh2023Gibbs,BojkovicLoh2024PrivateDensity,Donhauser2024Certified}.

Recent works also exploit structure beyond the ambient-dimensional worst
case. One approach is to first reduce dimension privately, for example by
private PCA, and then apply a lower-dimensional synthetic-data mechanism
\citep{HeStrohmerVershyninZhu2025LowDimensional}.
\citet[Appendix~D, Theorem~3]{Donhauser2024Certified} also adapt
to unknown support structure under full \(W_1\) loss, obtaining
a leading sample-size exponent \(1/(k+1)\) under a discretized-support
growth condition of exponent \(k\).
Our multiscale analysis controls transport error through the
occupied-cell counts at each level and yields the leading exponent
\(1/k\) for Adaptive Pruned-PMM when \(k>1\), under our packing-growth
condition and for fixed \(d\), geometric constants, and positive
privacy budgets.

\textbf{Runtime and refinement depth.}
For PMM, \citet{He2023Algorithmically}
report $O(\varepsilon dn)$ runtime at the ambient-depth choice
$r\approx\log_2(\varepsilon n)$, while their full-tree implementation
materializes $\Theta(2^r)$ cells.
The support-adaptive construction of
\citet[Appendix~D]{Donhauser2024Certified} uses discretized
histograms and optimization, without establishing a comparable
near-linear runtime guarantee.
Our implementation has expected synthesis time $O(nr+dn)$ at
depth $r$, and Adaptive Pruned-PMM has expected total running time
$O\!\left(d(n+d)\log(\varepsilon n)\right)$, including private
selection.
These bounds allow deeper, geometry-adapted hierarchies while
maintaining near-linear expected runtime for fixed dimension
and positive privacy budgets.
\subsection{Organization}
The rest of the paper is organized as follows. Section~\ref{sec:preliminary}
introduces the preliminaries and notation. Section~\ref{sec:pmm} reviews the
PMM. Section~\ref{sec:adaptive} presents Pruned-PMM and
Adaptive Pruned-PMM. Section~\ref{sec:guarantees} states the upper bound guarantees,
and Section~\ref{sec:lower} gives the lower bound. Section~\ref{sec:experiments}
reports empirical results, and Section~\ref{sec:conclusion} concludes.

\section{Preliminaries}
\label{sec:preliminary}
\paragraph{Notation and setup.}
For an integer $n$, write $[n]=\{1,\dots,n\}$. For any real number $x$, let $(x)_+:=\max\{x,0\}$. For a finite multiset $A$, $|A|$
denotes its cardinality with multiplicity. For a binary string $\theta$,
$|\theta|$ denotes its length; the empty string is denoted by $\varnothing$.
For a set \(S\subseteq\Omega\), write
\(\operatorname{dist}_\infty(z,S):=\inf_{s\in S}\|z-s\|_\infty\).
We work on $\Omega=[0,1]^d$ equipped with
$\rho(x,y)=\|x-y\|_\infty$, so $\operatorname{diam}(\Omega)=1$. For $n\ge1$, a
dataset is an ordered tuple $X=(x_1,\dots,x_n)\in\Omega^n$. We write \(S_X=\{x_1,\dots,x_n\}\) for the empirical support, ignoring
multiplicity when \(S_X\) is used inside distances or packing numbers. We write
$\mu_X:=\frac1n\sum_{i=1}^n\delta_{x_i}$ for its empirical measure. A
synthetic-data mechanism may output $Y=(y_1,\dots,y_m)\in\Omega^m$ of random size \(m\ge1\), with empirical measure
$\mu_Y:=\frac1m\sum_{i=1}^m\delta_{y_i}$. Utility is measured by
$W_1(\mu_X,\mu_Y)$.

\paragraph{Wasserstein utility.}
For probability measures $\mu,\nu$ on $\Omega$, we use the dual form
\[
W_1(\mu,\nu)
:=
\sup_{f \rm{~is~} 1-\rm{Lipschitz}}
\left|
\int f\,d\mu-\int f\,d\nu
\right|.
\]
Therefore, controlling $W_1(\mu_X,\mu_Y)$ uniformly controls the error of every
1-Lipschitz statistic on the original and synthetic empirical distributions.

\paragraph{Differential privacy.}
Two datasets $X,X'\in\Omega^n$ are adjacent,
written $X\sim X'$, if they differ in exactly one coordinate. A randomized
mechanism $\mathcal M$ is $\varepsilon$-differentially private if, for every
adjacent pair $X\sim X'$ and every measurable event $E$ in the output space,
\[
\Pr(\mathcal M(X)\in E)\le e^\varepsilon \Pr(\mathcal M(X')\in E)
.\]
We use the standard facts that post-processing preserves differential privacy
and privacy losses add under sequential composition.

\paragraph{Packing numbers.}
For a set $A\subseteq\Omega$ and a scale $t>0$, let
$N_{\mathrm{pack}}(A,t)$ denote the largest cardinality of a subset
$T\subseteq A$ such that $\|u-v\|_\infty\ge t$ for all distinct $u,v\in T$.

We also use the following definition of external packing number.
For a finite multiset $S\subseteq\Omega$ and parameters $\eta,\xi>0$, define
\[
N_{\mathrm{pack}}^{\mathrm{ext}}(S;\eta,\xi)
:=
\max\Bigl\{
|T|:\ T\subseteq\Omega,\ 
\operatorname{dist}_\infty(z,S)\le \eta\ ,\,\, \forall z,z'\in T, z\neq z',\ 
\|z-z'\|_\infty\ge \xi\ 
\Bigr\}.
\]

Equivalently, this is the packing number of the $\eta$-neighborhood
$S^{(+\eta)}:=\{z\in\Omega:\operatorname{dist}_\infty(z,S)\le \eta\}$ at
separation $\xi$. The upper-bound results impose this external packing condition on the realized
empirical support \(S_X\), not on any underlying population distribution.

\paragraph{Packing-growth dimension.}
The empirical support \(S_X\) has packing-growth dimension \(k\in[d]\) if
$
N_{\rm pack}^{\rm ext}(S_X;\eta,\eta/2)\lesssim \eta^{-k},
\,\, 0<\eta\le1.
$
The implicit constant is a geometric regularity constant, independent of
\(n\), \(\varepsilon\), and the algorithmic depth. In the main text, we suppress
this constant and refer to \(k\) as the packing-growth dimension. The appendix
states the corresponding bounds with the explicit constant \(C_{\rm pack}\).
\paragraph{Discrete Laplace noise.}
For $\sigma>0$, let $\mathrm{LapZ}(\sigma)$ denote the integer-valued discrete
Laplace distribution with probability mass function $\Pr(\lambda=z)
=
\frac{1-e^{-1/\sigma}}{1+e^{-1/\sigma}}e^{-|z|/\sigma},
 z\in\mathbb Z.
$
We use independent discrete Laplace variables to perturb integer counts.

\section{Background on the Private Measure Mechanism}
\label{sec:pmm}
The Private Measure Mechanism (PMM) of \citet{He2023Algorithmically} is our
full-tree baseline. PMM generates synthetic data by privatizing hierarchical
counts on a complete depth-\(r\) binary tree: it builds a binary partition,
adds independent noise to all tree counts, enforces consistency across the
tree, and samples synthetic points from the final leaves. Pruned-PMM keeps this
count-perturbation, but changes consistency structure and how the tree is
traversed.
\paragraph{Binary hierarchy.}
Fix a depth $r\ge 1$. PMM builds cells
$\{\Omega_\theta\}_{\theta\in\{0,1\}^{\le r}}$ with root
$\Omega_\varnothing=\Omega$. Each internal cell is split into two children by
bisecting one coordinate, cycling through coordinates with depth. Given
$X=(x_1,\dots,x_n)$, define
$
n_\theta:=\#\{i\in[n]:x_i\in\Omega_\theta\}
$
for $|\theta|\le r$. These counts are consistent:
$n_\theta=n_{\theta0}+n_{\theta1}$. We write
$
\tau_j:=\{\theta\in\{0,1\}^j:n_\theta>0\}
$
and \(O_j:=|\tau_j|\) for the occupied cells and their count at level \(j\).
\paragraph{Noisy counts.}
PMM perturbs every tree count independently. If $|\theta|=j$, it samples
$\lambda_\theta\sim\mathrm{LapZ}(\sigma_j)$ and sets 
$
n'_\theta=(n_\theta+\lambda_\theta)_+.
$
The truncation makes counts nonnegative, but the noisy family is generally not
consistent: \(n'_\theta\neq n'_{\theta0}+n'_{\theta1}\).
\paragraph{Consistency correction.}
PMM restores consistency by a top-down pass. Starting from
\(m_\varnothing=n'_\varnothing\), each internal node \(\theta\) assigns
nonnegative child counts \(m_{\theta0},m_{\theta1}\) such that
$
m_{\theta0}+m_{\theta1}=m_\theta.
$
The chosen child vector is required to be comparable to the noisy child vector
\((n'_{\theta0},n'_{\theta1})\): either
\(m_{\theta b}\le n'_{\theta b}\), or
\(m_{\theta b}\ge n'_{\theta b}\) for both \(b\in\{0,1\}\). This produces a consistent nonnegative count family
\((m_\theta)\).
\paragraph{Synthetic data.}
Finally, for each leaf $\Omega_\theta$, $|\theta|=r$, PMM places $m_\theta$
synthetic points inside $\Omega_\theta$ using a fixed rule independent of the
original data, such as the cell center. The union of these points is the synthetic
dataset $Y$.

\section{Algorithms}
\label{sec:adaptive}
In this section, we first present the 
fixed-depth Pruned-PMM algorithm, and then describe a private
selection rule for choosing the depth and the levelwise noise schedule. Pruned-PMM
keeps the same multiscale accuracy mechanism as PMM, but avoids materializing
branches that are pruned, improving runtime and memory usage. The depth controls
 the usual resolution--noise tradeoff: increasing \(r\) makes the leaf cells
smaller and reduces discretization error, but it also introduces noisy counts at
more levels. Under packing-growth dimension \(k\), the number of occupied cells at depth \(j\) is controlled by
\(2^{jk/d}\), rather than by the full tree size \(2^j\). This suggests the
geometry-adapted depth \(r\asymp d\log_2(\varepsilon n)\) when \(k=1\), and
$
r\asymp
\frac{d}{k}
\log_2\!\left(
1+\varepsilon n\left(\frac{k-1}{d}\right)^2
\right)
$
when \(k>1\).

The packing-growth dimension \(k\), and hence the best depth, is generally
unknown. Directly choosing the depth from the true occupied-cell profile would be
data-dependent and cannot be done for free under differential privacy. Our adaptive mechanism handles this by fixing a finite public set \(\mathcal A\)
of optimized depth--schedule pairs and using the exponential mechanism with a
data-dependent, sensitivity-controlled occupancy score to select one pair
privately. It then
runs fixed-depth Pruned-PMM once with the selected schedule.

\subsection{Fixed-Depth Pruned-PMM}

Assume first that the packing-growth dimension is known, so that the
depth and noise scales are chosen in advance. Pruned-PMM generates the
two noisy child counts at each active node and immediately applies
the zero-respecting correction in Algorithm~\ref{alg:pruned_consistency}.
Only children with positive corrected counts remain active.
Thus no subtree with zero corrected mass is constructed.
The multiscale accuracy guarantee is preserved.
Algorithm~\ref{alg:pruned-pmm} gives the fixed-depth subroutine.

\begin{algorithm}
\caption{Fixed-Depth Pruned-PMM}
\label{alg:pruned-pmm}
\begin{algorithmic}[1]
\REQUIRE Data $X\in\Omega^n$, depth $r$, noise scales $\sigma_0,\dots,\sigma_r>0$
\ENSURE Synthetic data $Y$

\STATE $A_0\gets\{\varnothing\}$ and
$n'_\varnothing\gets(n+\lambda_\varnothing)_+$, where
$\lambda_\varnothing\sim\mathrm{LapZ}(\sigma_0)$.
\algcomment{\(\varnothing\) is the root node}
\STATE $m_\varnothing\gets\max\{1,n'_\varnothing\}$.
\FOR{$j=0,\dots,r-1$}
    \STATE $A_{j+1}\gets\emptyset$.
    \algcomment{initialize the next active level}
    \FOR{each $\theta\in A_j$}
        \STATE For $b\in\{0,1\}$, set
        $n'_{\theta b}\gets(n_{\theta b}+\lambda_{\theta b})_+$, where
        $n_{\theta b}:=\#\{i:x_i\in\Omega_{\theta b}\}$ and
        $\lambda_{\theta b}\sim\mathrm{LapZ}(\sigma_{j+1})$ independently.
        \algcomment{only expand active nodes}
        \STATE $(m_{\theta0},m_{\theta1})\gets
\textsc{Pruned Consistency}(m_\theta,n'_{\theta0},n'_{\theta1})$.
\STATE $\mathrm{Child}(\theta)\gets
\{\theta b:b\in\{0,1\},\ m_{\theta b}>0\}$.
        \algcomment{prune zero corrected counts}
        \STATE $A_{j+1}\gets A_{j+1}\cup\mathrm{Child}(\theta)$.
    \ENDFOR
\ENDFOR

\STATE For each active leaf $\theta\in A_r$, place $m_\theta$ synthetic points in
$\Omega_\theta$; output their union as $Y$.
\end{algorithmic}
\end{algorithm}

Algorithm~\ref{alg:pruned_consistency} is a local correction applied
during tree construction. Its fixed, constant-time rule uses only
the parent mass, the two noisy child counts, and auxiliary randomness
independent of the data. Every inactive node and its descendants
have corrected count zero.

\begin{algorithm}
\caption{Pruned Consistency}
\label{alg:pruned_consistency}
\begin{algorithmic}[1]
\REQUIRE Parent mass $M\in\mathbb Z_{>0}$ and noisy child counts
$(a_0,a_1)\in\mathbb Z_+^2$
\ENSURE Corrected counts $(u_0,u_1)\in\mathbb Z_+^2$
with $u_0+u_1=M$
\IF{exactly one of $a_0,a_1$ is positive}
    \STATE Let $b$ satisfy $a_b>0$; set $u_b\gets M$ and $u_{1-b}\gets0$.
\ELSIF{$a_0=a_1=0$}
    \STATE Choose $b$ uniformly from $\{0,1\}$; set
    $u_b\gets M$ and $u_{1-b}\gets0$.
\ELSE
    \STATE Use the fixed local rule to choose $(u_0,u_1)$
    comparable to $(a_0,a_1)$ with $u_0+u_1=M$.
\ENDIF
\STATE Output $(u_0,u_1)$.
\end{algorithmic}
\end{algorithm}

\subsection{Adaptive Pruned-PMM}
In general, the packing-growth dimension, and hence the appropriate depth, is
not known. The best depth depends on the multiscale occupied-cell counts, which
are data-dependent. Adaptive Pruned-PMM therefore splits the privacy budget as
\(\varepsilon=\varepsilon_{\rm sel}+\varepsilon_{\rm main}\). The first part is
used to select a depth--schedule pair by the exponential mechanism;
the second part is used to run Pruned-PMM once with the selected schedule. Algorithm~\ref{alg:depth-select} gives the private selection step; the full
adaptive mechanism runs Algorithm~\ref{alg:depth-select} and then calls
Algorithm~\ref{alg:pruned-pmm} with the selected depth and schedule, root noise scale
\(\sigma_0=\varepsilon_{\rm main}^{-1}\), and fresh independent randomness.

For the \textit{standard candidate set}, assume \(\varepsilon_{\rm main}n\ge2\);
\(s\in[d]\) represents the candidate packing-growth dimension. Let
\(L:=\lfloor d\log_2(\varepsilon_{\rm main}n)\rfloor\), set \(r_1:=L\), and
for \(s\ge2\) set
$
r_s:=\max\!\left\{1,\left\lfloor
\frac{d}{s}\log_2\!\left(1+\varepsilon_{\rm main}n\left(\frac{s-1}{d}\right)^2\right)
\right\rfloor\right\}.
$
With \(P_j^{(s)}:=\min\{n,2^{js/d}\}\) and
\(A_{s,r_s}:=2\sum_{j<r_s}\sqrt{P_j^{(s)}}2^{-j/(2d)}\), define
$
\sigma_{j+1}^{(s,r_s)}
=
\frac{A_{s,r_s}}
{\varepsilon_{\rm main}\sqrt{P_j^{(s)}}2^{-j/(2d)}} .
$
Then \(2\sum_{i=1}^{r_s}1/\sigma_i^{(s,r_s)}=\varepsilon_{\rm main}\), so each
candidate \((s,r_s,\sigma^{(s,r_s)})\) is a fixed public depth--schedule pair
using the main privacy budget. The standard candidate set is defined as
\[
\mathcal A:=\{(s,r_s,\sigma^{(s,r_s)}):s\in[d]\}.
\]
For \(a\in\mathcal A\), define the score
\begin{equation}
\label{main:B_a_def}
B_a(X)
:=
\frac{C_{\rm abs}}{\varepsilon_{\rm main}n}
+
\frac{4C_{\rm abs}}{n}
\sum_{j=0}^{r_a-1}\sigma^a_{j+1}O_j(X)2^{-j/d}
+
2\cdot2^{-r_a/d}.
\end{equation}
This is the fixed-schedule upper bound evaluated at the observed occupied-cell
profile. Each \(O_j\) with
\(j\ge1\) changes by at most one, the sensitivity of this score is bounded by
$
\Delta_a
:=
\frac{4C_{\rm abs}}{n}
\sum_{j=1}^{r_a-1}\sigma^a_{j+1}2^{-j/d},
\,\,
\Delta_{\mathcal A}:=\max_{a\in\mathcal A}\Delta_a .
$
If \(\Delta_{\mathcal A}>0\), the exponential mechanism selects \(\widehat a\in\mathcal A\) with probability
proportional to
$
\exp\left(
-\frac{\varepsilon_{\rm sel}B_a(X)}{2\Delta_{\mathcal A}}
\right)
$
. If \(\Delta_{\mathcal A}=0\), select a minimizer of \(B_a(X)\) using a fixed public tie-breaking rule. Conditional on
\(\widehat a\), the final synthesis stage is an ordinary fixed-depth
Pruned-PMM run with budget \(\varepsilon_{\rm main}\). Hence, the two stages are
\(\varepsilon\)-differentially private by sequential composition.

\begin{algorithm}
\caption{Adaptive Depth Selection}
\label{alg:depth-select}
\begin{algorithmic}[1]
\REQUIRE Data \(X\in\Omega^n\), budgets
\(\varepsilon_{\rm sel},\varepsilon_{\rm main}>0\), candidate set \(\mathcal A\)
\ENSURE Candidate \(\widehat a=(\widehat r,\widehat\sigma_1,\dots,\widehat\sigma_{\widehat r})\)

\STATE Compute \(O_j\) for every \(0\le j<\max_{a\in\mathcal A}r_a\).
\algcomment{occupied cells}
\STATE For each \(a\in\mathcal A\), compute \(B_a(X)\) in \eqref{main:B_a_def}.
\algcomment{data-dependent score}
\STATE Let \(\Delta_{\mathcal A}:=\max_{a\in\mathcal A}\Delta_a\).
\algcomment{global sensitivity bound}

\IF{$\Delta_{\mathcal A}=0$}
    \STATE Choose \(\widehat a\in\arg\min_{a\in\mathcal A}B_a(X)\) with fixed public tie-breaking.
\ELSE
\STATE Sample \(\widehat a\) with probability proportional to
$
\exp\left(
-\frac{\varepsilon_{\rm sel}B_a(X)}{2\Delta_{\mathcal A}}
\right).
$
\algcomment{exponential mechanism}
\ENDIF

\STATE Output \(\widehat a\).
\end{algorithmic}
\end{algorithm}

\section{Theoretical Guarantees}
\label{sec:guarantees}

We state the guarantees in the same order as the mechanisms. We first give the
fixed-depth rate under external packing-growth, which is the structural
assumption used in the upper bound.

\subsection{Fixed-Depth Pruned-PMM}

We first consider the case where the packing-growth dimension is known, so the
depth and noise schedule are fixed before running the mechanism.

\begin{theorem}[Fixed-depth rate under external packing-growth]
\label{thm:fixed-depth-rate-main}
Assume \(\varepsilon_{\rm main}n\ge2\). Suppose the empirical support \(S_X\)
has packing-growth dimension \(k\in[d]\).
Then there is a choice of depth and noise scales such that Pruned-PMM is
\(\varepsilon_{\rm main}\)-differentially private and
\[
\mathbb E W_1(\mu_X,\mu_Y)
\le
\begin{cases}
C\dfrac{d^2\log^2(\varepsilon_{\rm main}n)}{\varepsilon_{\rm main}n},
& k=1,\\[2mm]
C\left(\dfrac{d}{k-1}\right)^{2/k}(\varepsilon_{\rm main}n)^{-1/k},
& 2 \le k\le d, 
\end{cases}
\]
where \(C\) depends only on the packing-growth constants and universal constants.
\end{theorem}

\begin{proof}[Proof sketch]
By Corollary~\ref{cor:acc-exp}, the fixed-depth error is the sum of a mass term,
a noise-transport term, and a leaf-resolution term. For fixed \(r\), optimizing
the noise scales under the privacy constraint
\(2\sum_i1/\sigma_i\le\varepsilon_{\rm main}\) gives the profile-dependent bound
of Proposition~\ref{prop:opts}. The packing condition gives
\(O_j\lesssim C_{\rm pack}2^{jk/d}\), reducing the bound to a tradeoff between
\((\varepsilon_{\rm main}n)^{-1}(\sum_{j<r}2^{j(k-1)/(2d)})^2\) and
\(2^{-r/d}\). Choosing \(r\asymp d\log(\varepsilon_{\rm main}n)\) for \(k=1\),
and \(r\asymp \frac{d}{k}\log(1+\varepsilon_{\rm main}n((k-1)/d)^2)\) for
\(k>1\), gives the rates. See Corollary~\ref{cor:rates}.
\end{proof}

\subsection{Adaptive Pruned-PMM}
We now state the guarantee for the adaptive selector above. The key point is the best-candidate guarantee in
Theorem~\ref{thm:adaptive-library}: the adaptive mechanism performs as well as the best public
candidate, up to the exponential-mechanism selection penalty.

\begin{theorem}[Adaptive Pruned-PMM under packing-growth]
\label{thm:adaptive-rate-main}
Assume \(d\ge2\), \(\varepsilon_{\rm main}n\ge2\), and
\(\varepsilon=\varepsilon_{\rm sel}+\varepsilon_{\rm main}\). Run Adaptive
Pruned-PMM with the standard candidate set indexed by \(s\in[d]\). Then the
mechanism is \(\varepsilon\)-differentially private. If the empirical support
\(S_X\) has packing-growth dimension \(k\), then
\[
\mathbb E W_1(\mu_X,\mu_Y)
\le
\begin{cases}
C\dfrac{d^2 \log^2(\varepsilon_{\rm main}n)}
{\varepsilon_{\rm main}n}+
C_{\rm sel}
\frac{d^2 \log d}{\varepsilon_{\rm sel}}
(\varepsilon_{\rm main}n)^{-(d+1)/(2d)}, & k=1,\\[3mm]
C\left(\dfrac{d}{k-1}\right)^{2/k}
(\varepsilon_{\rm main}n)^{-1/k}+
C_{\rm sel}
\frac{d^2 \log d}{\varepsilon_{\rm sel}}
(\varepsilon_{\rm main}n)^{-(d+1)/(2d)},
& 2\le k\le d.
\end{cases}
\]
Here, \(C\) depends on $C_{\rm pack}$ and \(C_{\rm sel}\).
\end{theorem}

\begin{proof}[Proof sketch]
The score \(B_a(X)\) has controlled sensitivity over the standard set
\(\mathcal A\), and Lemma~\ref{lem:selection-sensitivity-simple} gives the
uniform bound on \(\Delta_{\mathcal A}\). Thus the exponential-mechanism
selector, followed by the fixed-schedule Pruned-PMM run, is private by
Theorem~\ref{thm:adaptive-privacy}. For accuracy, the best-candidate guarantee
in Theorem~\ref{thm:adaptive-library} shows that the adaptive error is bounded
by the best public candidate plus the selection penalty. Under packing-growth
dimension \(k\), the standard set contains the candidate \(s=k\), whose rate is
bounded in Corollary~\ref{cor:adaptive-explicit-rate}. Combining these facts
gives the displayed bound; see Corollary~\ref{cor:adaptive-final-rate}.
\end{proof}

For every \(2\le k\le d\), the selection term is of lower order than the
packing-growth term for fixed \(d\), fixed packing-growth constants, and fixed positive privacy budgets
\(\varepsilon_{\rm sel},\varepsilon_{\rm main}\). Hence, Adaptive Pruned-PMM achieves the intrinsic exponent
\(1/k\) without knowing \(k\) in advance. For \(k=1\), the fixed-depth candidate
has the logarithmic one-dimensional rate, while the adaptive bound retains the selection penalty. When the privacy budgets vary with
\(n\), the full bound above retains the selection penalty.

\subsection{Runtime}
Consistency is applied before expansion, so at most \(m\) nodes
have positive corrected counts at each depth, where \(m\) is the
output size. The synthesis stage therefore visits at most \(1+2mr\)
nodes. With \(\sigma_0=\varepsilon_{\rm main}^{-1}\) and
\(\varepsilon_{\rm main}n\ge2\), its expected runtime is
\(O(nr+dn)\), including explicit output generation.

For Adaptive Pruned-PMM, let \(L:=\max_{a\in\mathcal A}r_a\).
Including occupancy computation and candidate selection, the total
expected runtime is \(O((n+|\mathcal A|)L+dn)\).
For the standard candidate set, \(|\mathcal A|=d\) and
\(L=\lfloor d\log_2(\varepsilon_{\rm main}n)\rfloor\).
See Theorem~\ref{thm:time} and Corollary~\ref{cor:time-intrinsic}
in Appendix~\ref{subsec:runtime-pruning}. The synthesis runtime bound grows linearly with the depth,
allowing deeper, geometry-adapted hierarchies without the
$2^r$ cost of materializing the complete tree.
For fixed $d$ and fixed positive privacy budgets, the standard
candidate set therefore gives $O(n\log n)$ expected runtime
for the complete adaptive mechanism.

\section{Lower Bound and Sharpness}
\label{sec:lower}

We complement the upper bound with a packing lower bound. The upper analysis
uses an external packing upper bound to control the occupied cells. The lower bound uses a corresponding internal lower-growth condition: the
support contains many separated points at every small scale. It shows that the
leading exponent \(1/k\) is forced by the geometry, not by the upper-bound proof
technique.

\begin{theorem}[Packing lower bound]
\label{thm:packing-lower-bound}
Let \(T=[0,1]^d\) with \(\rho(x,y)=\|x-y\|_\infty\). Suppose that
\(S\subset T\) satisfies
$
N_{\mathrm{pack}}(S,t)\ge c_0t^{-k}
\,\,\text{for all }0<t\le t_0,
$
for some \(c_0>0\), \(t_0\in(0,1]\), and \(k\in\{1,\dots,d\}\). Then there exist
constants \(c>0\) and \(n_0\in\mathbb N\), depending only on \(c_0,t_0,k\), such
that for every \(n\ge n_0\), every \(\varepsilon\in(0,1]\) with
\(\varepsilon n\ge2\), and every pure \(\varepsilon\)-DP synthetic-data mechanism
\(\mathcal M:S^n\to\bigsqcup_{m\ge1}T^m\),
\[
\sup_{X\in S^n}
\mathbb E W_1\bigl(\mu_X,\mu_{\mathcal M(X)}\bigr)
\ge
c\,\min\{t_0,(\varepsilon n)^{-1/k}\}.
\]
\end{theorem}

\begin{proof}[Proof sketch]
Set \(M\asymp \varepsilon n\) and
\(t\asymp\min\{t_0,M^{-1/k}\}\). The packing lower-growth assumption gives
\(M+1\) points in \(S\) separated at scale \(t\), including the filler point. A constant-weight code provides
more than \(2e^{\varepsilon n}\) subsets that differ on a constant fraction of
their coordinates. Each codeword defines a dataset by placing equal mass on its
selected points and using one common filler point for the remaining samples.
The resulting empirical measures form a \(W_1\)-packing at scale \(\Omega(t)\).
Restricting any \(\varepsilon\)-DP mechanism to this finite family and
post-processing to empirical measures gives, by group privacy, an
\((\varepsilon n)\)-metrically private mechanism. The master packing lower bound
for metric privacy \citep{BoedihardjoStrohmerVershynin2024PrivateMeasures}
therefore forces expected error \(\Omega(t)\) for some dataset. Since
\(M\asymp\varepsilon n\), this gives
\(\Omega(\min\{t_0,(\varepsilon n)^{-1/k}\})\). Details are in
Appendix~\ref{subsec:lb-main}.
\end{proof}

We next identify a class of supports on which the Pruned-PMM
upper bound matches the lower bound of
Theorem~\ref{thm:packing-lower-bound} in its dependence on
$\varepsilon n$.
The upper analysis controls the noise-transport error through
the number of occupied cells at each level; external packing
provides a geometric bound on these counts.
For finite unions of coordinate cubes, these counts can be
bounded directly, even when the components have different
active coordinates or intersect.
The resulting rate is tight for fixed ambient dimension,
intrinsic dimension, and number of components.

\begin{proposition}[Matching rates on finite unions of coordinate cubes]
\label{prop:matching-coordinate-supports}
Fix $d\ge2$ and $2\le k\le d$. Let
\[
S=\bigcup_{\ell=1}^{J}S_\ell,
\qquad
S_\ell=
\{x\in[0,1]^d:
x_{A_\ell}\in[0,1]^k,\ x_{A_\ell^c}=z_\ell\},
\]
where $J\ge1$, $A_\ell\subseteq[d]$ with $|A_\ell|=k$, and
$z_\ell\in[0,1]^{A_\ell^c}$.
The active coordinate sets may differ, and the components
may intersect. There exist constants $c_k>0$ and $n_0\in\mathbb N$ depending
only on $k$, and a universal constant $C$, such that, for every
$0<\varepsilon\le1$ and $n\ge n_0$ with $\varepsilon n\ge4$,
Pruned-PMM admits a public choice of depth and noise scales
for which it is $\varepsilon$-differentially private and its
output $Y$ satisfies
\[
\sup_{X\in S^n}
\mathbb E W_1(\mu_X,\mu_Y)
\le
C J^{1/k}
\left(\frac{d}{k-1}\right)^{2/k}
(\varepsilon n)^{-1/k}.
\]
The depth and noise scales depend only on $d,k,J,n,\varepsilon$,
not on the active coordinate sets, component locations, or
allocation of samples among components. Conversely, every pure \(\varepsilon\)-DP synthetic-data mechanism has worst-case expected $W_1$ error over $S^n$ at least
$c_k(\varepsilon n)^{-1/k}$. Thus, for fixed $d,k,J$, Pruned-PMM
attains the minimax-optimal rate
$\Theta_{d,k,J}((\varepsilon n)^{-1/k})$.
\end{proposition}

The proof is given in Appendix~\ref{subsec:matching-growth}.

\section{Experiments}
\label{sec:experiments}

The experiments validate two algorithmic messages. First, private depth
selection should improve empirical \(W_1\) over the ambient-depth PMM baseline.
Second, pruning should reduce computation without changing the PMM accuracy
behavior. Unless stated otherwise, we use \(n=30000\), \(\varepsilon=1\), the
\(\ell_\infty\) ground metric, and estimate \(W_1\) on random subsamples.

\paragraph{Adaptive depth selection and pruning.}
Figure~\ref{fig:experiment-summary} summarizes the main findings. In the
high-dimensional experiment, the data lie on a \(k\)-dimensional coordinate
subspace of \([0,1]^d\), while \(d\) varies. Full PMM uses the ambient-depth
rule, whereas Adaptive Pruned-PMM privately selects a depth--schedule pair from
the public candidate set. The left panel shows \(k=1\): Adaptive Pruned-PMM
achieves lower empirical \(W_1\) than PMM, with a gap that increases with \(d\).
Additional plots for \(k=3,5\) appear in Appendix~\ref{app:experiments}.

The right panel isolates pruning. We compare full PMM and Pruned-PMM on nine
two-dimensional datasets: two moons, spiral, annulus, S-curve, pinwheel, figure
eight, four blobs, checkerboard, and cross. We use
\(r=\lceil\log_2(\varepsilon n)\rceil=15\). The \(W_1\) ratio remains close to
one, while runtime and visited-node ratios favor Pruned-PMM. Thus pruning
preserves the PMM accuracy behavior while reducing the cost of traversing the
hierarchy.

\begin{figure}[t]
\centering

\begin{subfigure}[t]{0.49\linewidth}
    \centering
    \includegraphics[width=\linewidth]{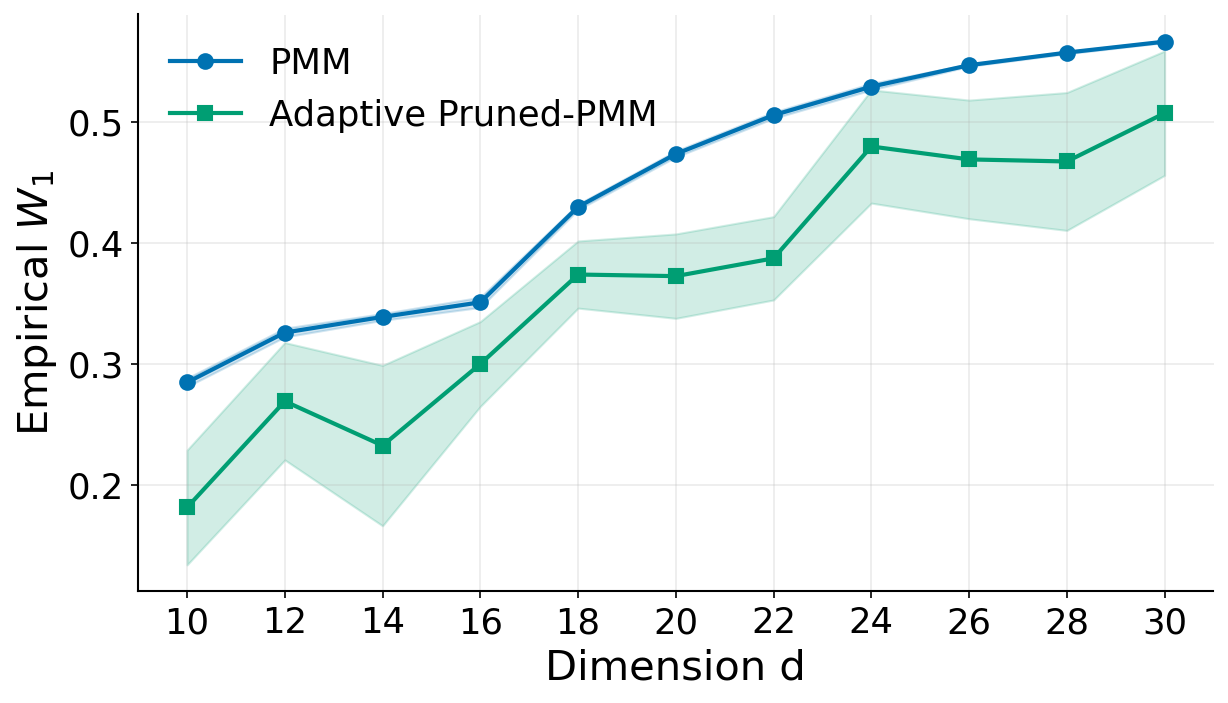}
    \caption{Adaptive depth selection, \(k=1\).}
    \label{fig:exp-summary-adaptive}
\end{subfigure}
\hfill
\begin{subfigure}[t]{0.49\linewidth}
    \centering
    \includegraphics[width=\linewidth]{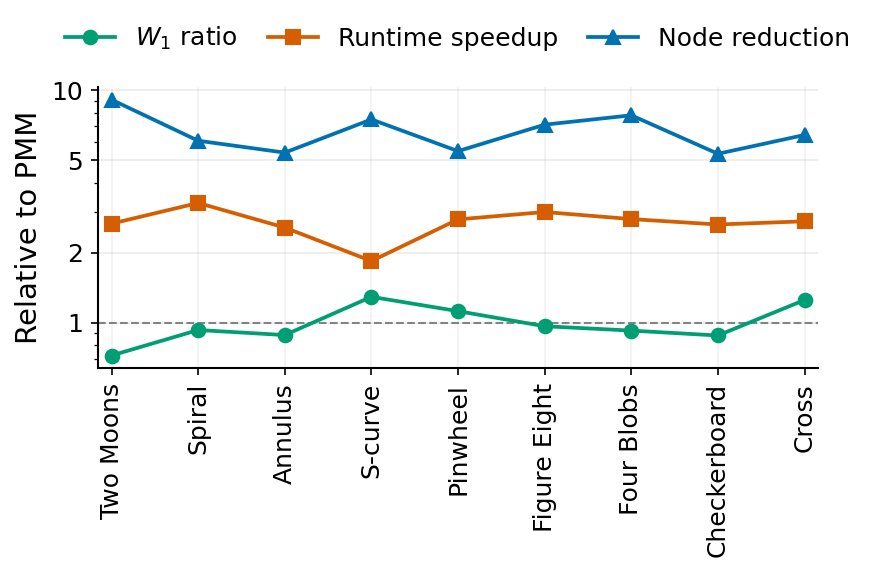}
    \caption{PMM versus Pruned-PMM on 2D shapes.}
    \label{fig:exp-summary-pruning}
\end{subfigure}

\caption{
Left: Adaptive Pruned-PMM lowers empirical \(W_1\) relative to the ambient-depth PMM baseline as \(d\) grows. Right: Pruned-PMM preserves PMM-level \(W_1\) while reducing runtime and visited nodes. Ratios are Pruned-PMM/PMM for \(W_1\), and PMM/Pruned-PMM for runtime and nodes.
}
\label{fig:experiment-summary}
\end{figure}

\paragraph{Depth scaling.}
Figure~\ref{fig:depth-sweep-main} shows the effect of increasing \(r\). Full PMM
materializes the complete binary tree, while Pruned-PMM expands only branches
that remain active under the pruning rule. The \(W_1\) curves stay comparable,
but runtime and visited-node counts separate rapidly on a log scale. This
supports the computational claim: \textit{pruning keeps the same accuracy
mechanism as PMM but avoids expanding branches with zero corrected mass.}

\begin{figure}[t]
\centering
\includegraphics[width=\linewidth]{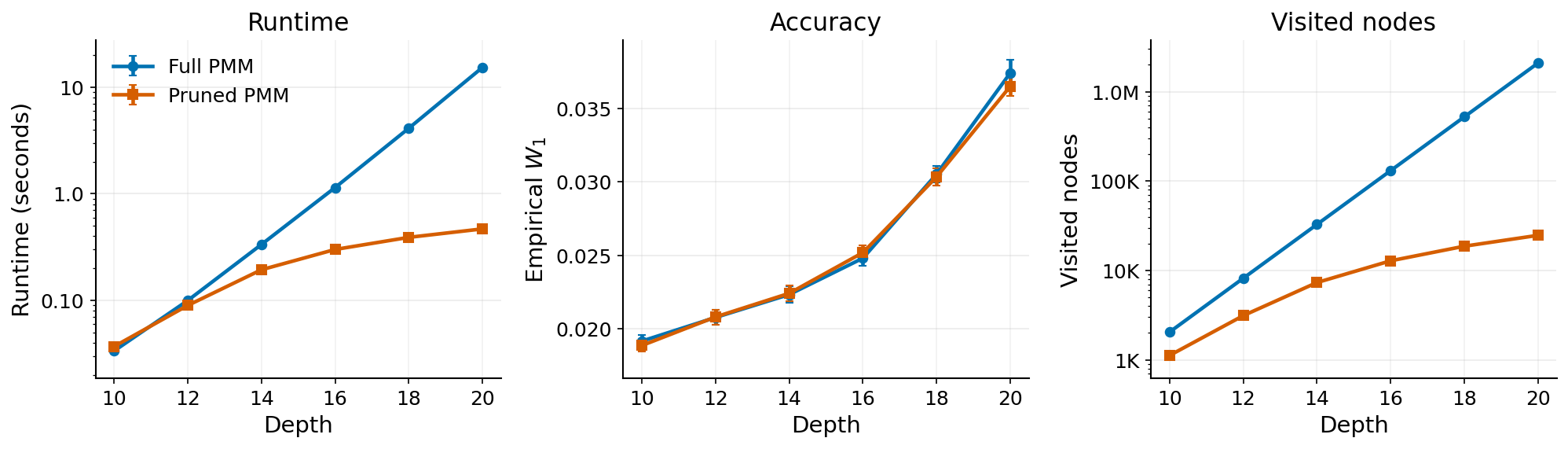}
\caption{
Depth sweep on the 2D shape suite. For each depth, points and error bars average
over the nine shapes and 10 independent runs per shape. Pruned-PMM has comparable
empirical \(W_1\) to full PMM, while using substantially less runtime and visiting
far fewer tree nodes.
}
\label{fig:depth-sweep-main}
\end{figure}

\section{Conclusion}
\label{sec:conclusion}

We studied differentially private synthetic data under Wasserstein utility when
the empirical support occupies a geometrically simpler subset of the ambient cube.
The key structural assumption is packing growth: across scales, the neighborhood
of the empirical support can be packed like a \(k\)-dimensional set rather than a
\(d\)-dimensional one. This condition is natural for regular low-dimensional supports, such as
compact curves, surfaces, Lipschitz manifolds, finite unions of such pieces, and, more generally, sets with controlled upper Minkowski growth.
Building on PMM, we introduced \emph{Pruned-PMM} and \emph{Adaptive Pruned-PMM}. Pruned-PMM
keeps the same multiscale accuracy mechanism as PMM but avoids expanding branches with zero corrected mass. Adaptive Pruned-PMM adds a private depth-selection step, allowing the mechanism to choose among public depth--noise schedules
without knowing the packing-growth dimension in advance. Under packing-growth dimension \(k\), for fixed positive privacy budgets and fixed geometry as \(n\) grows, the resulting Wasserstein error has leading rate \((\varepsilon n)^{-1/k}\) for \(k>1\), replacing the ambient-dimensional PMM exponent \(1/d\) by the packing-growth exponent \(1/k\). The lower bound shows that
this exponent is unavoidable under a matching internal packing-growth condition.

\paragraph{Limitations.}
Our results are based on the realized empirical support of the data. If a population distribution is supported on a fixed set satisfying the same external packing upper bound, every empirical support inherits that bound by containment. Another natural direction is to replace the
integer-indexed candidate set with a finer public grid, allowing adaptation to real-valued growth dimensions at the cost of a larger private-selection penalty.
Finally, the upper and lower bounds use closely related but not identical packing conditions: the upper bound uses external packing growth, while the lower bound uses internal packing growth. Proposition~\ref{prop:matching-coordinate-supports} identifies a class satisfying both conditions and matches the dependence on \(\varepsilon n\) for fixed \(d,k,J\); matching the dependence on ambient dimension and geometric constants more generally remains open. 

\section*{Acknowledgments}

Amir R. Asadi acknowledges support from the Leverhulme Trust
(grant ECF-2023-189) and the Isaac Newton Trust (grant 23.08(b)).

\bibliographystyle{unsrtnat}
\bibliography{bib}

\newpage
\appendix
\phantomsection
\addcontentsline{toc}{chapter}{Appendix Contents}
\localtableofcontents
\include{upper}
\include{exponential}
\include{lower}
\include{app_experiment}

\end{document}

%% file: upper.tex
\section{Pruned PMM under packing-growth}

\subsection{Setup and notation}

Let $\Omega=[0,1]^d$ equipped with $\rho(x,y)=\|x-y\|_\infty$. Then $\operatorname{diam}(\Omega)=1$.
For a finite multiset $\mathcal{X}=\{x_1,\dots,x_n\}\subset\Omega$ with $n\ge 1$ define the empirical
probability measure
\[
\mu_\mathcal{X} := \frac{1}{n}\sum_{i=1}^n \delta_{x_i}.
\]
Let $W_1(\cdot,\cdot)$ denote 1-Wasserstein distance w.r.t.\ $\rho$.
\paragraph{Hierarchical Partition}
Fix a depth $r\ge 0$ and construct a binary partition tree
$\{\Omega_\theta\}_{\theta\in\{0,1\}^{\le r}}$ as follows: $\Omega_\varnothing=\Omega$ and
each node $\theta$ at depth $|\theta|=j<r$ is split into children $\theta0,\theta1$ by
bisecting $\Omega_\theta$ at the midpoint along coordinate $((j\bmod d)+1)$.\\\\
Given data $X=(x_1,\dots,x_n)\in\Omega^n$ ($n\ge 1$), define true counts
\[
n_\theta := \#\{i\in[n]: x_i\in\Omega_\theta\}\qquad(|\theta|\le r).
\]

\paragraph{Noisy counts.}
For each node $\theta$ at level $|\theta|=j$, sample an independent integer-valued noise
$\lambda_\theta\sim\mathrm{LapZ}(\sigma_j)$, and define truncated noisy counts
\[
n'_\theta := (n_\theta+\lambda_\theta)_+ \ :=\ \max\{n_\theta+\lambda_\theta,\,0\}.
\]

\subsection{Algorithm}
Pruned-PMM samples the noisy counts of the two children of each
active node and immediately applies
Algorithm~\ref{alg:pruned_consistency}. Only children with positive
corrected counts are retained. The root is initialized as
\(m_\varnothing=\max\{1,n'_\varnothing\}\).

For each level \(j\), define
\[
A_j:=\{\theta\in\{0,1\}^j:m_\theta>0\},
\qquad
\mathrm{Child}(\theta)
:=\{\theta b:b\in\{0,1\},\ m_{\theta b}>0\}.
\]
Every inactive node and all its descendants have corrected count
zero. Thus \((m_\theta)\) is understood as a consistent family on
the full tree. The zero-respecting local correction retains the
single-zero routing rule and the uniformly random double-zero
fallback. Algorithm~\ref{alg:pruned-pmm} gives the full procedure.

\begin{remark}[Relation to standard PMM]
Standard PMM enforces consistency by replacing each noisy sibling pair by
\emph{any} comparable vector on the line of the parent mass. In contrast,
Pruned-PMM imposes a stricter zero-respecting rule in the zero-noisy cases:
if exactly one noisy child count is zero, that child is assigned zero consistent
mass and the full parent mass is routed to the other child; if both noisy child
counts are zero, one child is selected and kept active while the other is assigned
zero mass. Outside the zero cases, any fixed local comparable rule is admissible.
Children with zero corrected mass are not expanded.
Thus Pruned-PMM is best viewed as a structured PMM-type variant with a
zero-respecting consistency map. 
\end{remark}

\subsection{Privacy}

Throughout this subsection, adjacency means \emph{replacement adjacency} on $\Omega^n$:
two data sets $X,X'\in\Omega^n$ are adjacent if they differ in exactly one coordinate.

The privacy proof is most naturally stated through a conceptual full-tree experiment.
Although the implementation of Pruned-PMM samples noises only when active nodes
are visited, we may equivalently imagine that an independent noise variable is sampled
in advance for \emph{every} node of the complete binary tree of depth $r$.
The pruned algorithm then reveals and uses only those coordinates of the full noisy tree
that are needed by the pruning rule. This is exactly the right viewpoint for privacy.

\begin{remark}[Lazy noise sampling as post-processing]\label{rem:lazy}
For every node $\theta\in\{0,1\}^{\le r}$, let
\[
\lambda_\theta\sim \mathrm{LapZ}(\sigma_{|\theta|})
\]
be sampled independently. The actual implementation of Pruned-PMM samples
$\lambda_\theta$ only when the node $\theta$ is visited, but this is equivalent in distribution
to first sampling the full family
$
\lambda=(\lambda_\theta)_{\theta\in\{0,1\}^{\le r}}
$
and then running the pruning, consistency, and output steps using only the visited coordinates.
Hence the lazy implementation is a measurable function of the full noisy tree and differs from
the full-tree experiment only by post-processing.
\end{remark}

\begin{lemma}[discrete-Laplace mechanism]\label{lem:inhom}
Let $F:\Omega^n\to\mathbb Z^N$ be any map, let $s=(s_i)_{i=1}^N\in(0,\infty)^N$,
and let $\lambda=(\lambda_i)_{i=1}^N$ have independent coordinates
$\lambda_i\sim \mathrm{LapZ}(s_i)$. Then the mechanism
$
X\longmapsto F(X)+\lambda
$
is $\varepsilon$-differentially private with
$
\varepsilon
=
\sup_{X\sim X'} \|F(X)-F(X')\|_{\ell_1(s)},
\,\,
\|z\|_{\ell_1(s)}:=\sum_{i=1}^N \frac{|z_i|}{s_i},
$
where the supremum is over all adjacent data sets $X,X'$.
\end{lemma}

\begin{proof}
Fix adjacent data sets $X,X'$ and an output vector $y\in\mathbb Z^N$. Then
\[
\frac{\Pr(F(X)+\lambda=y)}{\Pr(F(X')+\lambda=y)}
=
\prod_{i=1}^N
\frac{\Pr(\lambda_i=y_i-F(X)_i)}{\Pr(\lambda_i=y_i-F(X')_i)}.
\]
Using the pmf of $\mathrm{LapZ}(s_i)$, the normalizing constants cancel, giving
\[
\frac{\Pr(F(X)+\lambda=y)}{\Pr(F(X')+\lambda=y)}
=
\exp\!\left(
\sum_{i=1}^N
\frac{|y_i-F(X')_i|-|y_i-F(X)_i|}{s_i}
\right).
\]
By the triangle inequality,
$
|y_i-F(X')_i|-|y_i-F(X)_i|
\le |F(X)_i-F(X')_i|,
$
hence
\[
\frac{\Pr(F(X)+\lambda=y)}{\Pr(F(X')+\lambda=y)}
\le
\exp\!\left(
\sum_{i=1}^N \frac{|F(X)_i-F(X')_i|}{s_i}
\right)
=
\exp\!\big(\|F(X)-F(X')\|_{\ell_1(s)}\big).
\]
Taking the supremum over adjacent $X,X'$ yields the claim.
\end{proof}

\begin{theorem}[Privacy of Pruned-PMM]\label{thm:privacy}
Let
$
\varepsilon := 2\sum_{j=1}^r \frac{1}{\sigma_j}.
$
Then Pruned-PMM is $\varepsilon$-differentially private.
\end{theorem}

\begin{proof}
Consider the complete binary tree $\{0,1\}^{\le r}$ and define the count map
\[
F(X):=(n_\theta(X))_{\theta\in\{0,1\}^{\le r}}.
\]
For each node $\theta$, set
$
\sigma_\theta:=\sigma_{|\theta|},
\,\,
\lambda_\theta\sim \mathrm{LapZ}(\sigma_\theta),
$
with all $\lambda_\theta$ independent.
Fix adjacent data sets $X,X'\in\Omega^n$, and suppose they differ only in one coordinate:
for some $i\in[n]$, the point $x_i$ in $X$ is replaced by $x_i'$ in $X'$.
At the root level, both data sets contain exactly $n$ points, so
$
n_\varnothing(X)=n_\varnothing(X')=n.
$
Thus the root contributes zero to the weighted $\ell_1$ sensitivity.

Now fix a level $j\in\{1,\dots,r\}$. Let $\theta_j\in\{0,1\}^j$ be the unique cell at level $j$
containing $x_i$, and let $\theta_j'\in\{0,1\}^j$ be the unique cell at level $j$ containing $x_i'$.

If $\theta_j=\theta_j'$, then no count changes at level $j$, and therefore
\[
\sum_{\theta\in\{0,1\}^j}|n_\theta(X)-n_\theta(X')|=0.
\]
If $\theta_j\neq\theta_j'$, then exactly two counts change at level $j$:
the count of $\Omega_{\theta_j}$ changes by $1$ in magnitude, the count of $\Omega_{\theta_j'}$
also changes by $1$ in magnitude, and all other level-$j$ counts remain unchanged. Hence
\[
\sum_{\theta\in\{0,1\}^j}|n_\theta(X)-n_\theta(X')|=2.
\]
It follows that
\[
\|F(X)-F(X')\|_{\ell_1(\sigma)}
=
\sum_{j=0}^r \sum_{\theta\in\{0,1\}^j}
\frac{|n_\theta(X)-n_\theta(X')|}{\sigma_j}
\le
2\sum_{j=1}^r \frac{1}{\sigma_j}
=
\varepsilon.
\]
By Lemma~\ref{lem:inhom}, the full noisy count vector
$
\big(n_\theta(X)+\lambda_\theta\big)_{\theta\in\{0,1\}^{\le r}}
$
is $\varepsilon$-differentially private.

Now define the truncated noisy counts
$
n'_\theta := (n_\theta+\lambda_\theta)_+.
$
The family $(n'_\theta)$ is a deterministic function of the full noisy count vector.
The active sets \(A_j\), child map \(\mathrm{Child}(\theta)\), and
corrected counts \((m_\theta)\) produced by
Algorithm~\ref{alg:pruned-pmm}, using the local correction in
Algorithm~\ref{alg:pruned_consistency}, are measurable functions
of the full noisy count vector and auxiliary randomness independent
of the data. Stopping at zero corrected counts is therefore
post-processing.
Finally, the synthetic data $Y$ are generated from the counts $(m_\theta)$ using randomness
independent of $X$.

Hence the entire output of Pruned-PMM is obtained from the full noisy count vector by
post-processing and independent randomization. Therefore Pruned-PMM is
$\varepsilon$-differentially private.
\end{proof}
\subsection{Noise model and structural assumptions}

\begin{assumption}[Discrete Laplace noise]\label{ass:noise}
For $\sigma>0$, let $\lambda\sim\mathrm{LapZ}(\sigma)$ denote the integer-valued
discrete Laplace distribution with probability mass function
\[
\Pr(\lambda=z)
=
\frac{1-e^{-1/\sigma}}{1+e^{-1/\sigma}}\,e^{-|z|/\sigma},
\qquad z\in\mathbb Z.
\]
Then $\lambda$ is symmetric, i.e.
$
\Pr(\lambda=z)=\Pr(\lambda=-z)\,\,(z\in\mathbb Z),
$
and the following bounds hold for all $\sigma>0$ and all $t\ge 0$:
\[
\mathbb E|\lambda|
=
\frac{1}{\sinh(1/\sigma)}
\le \sigma,
\qquad
\Pr(\lambda\le -t)
=
\frac{e^{-\lceil t\rceil/\sigma}}{1+e^{-1/\sigma}}
\le e^{-t/\sigma}.
\]
Equivalently, in the bounds
$
\mathbb E|\lambda|\le C_{\mathrm{abs}}\sigma,
\,\,
\Pr(\lambda\le -t)\le C_{\mathrm{tail}}e^{-t/\sigma},
$
one may take
$
C_{\mathrm{abs}}=1,
\,\,
C_{\mathrm{tail}}=1.
$
\end{assumption}

\begin{assumption}[External packing-growth]\label{ass:pack}
For $\gamma,\xi>0$, define the external packing number of the data set
$S:=\{x_1,\dots,x_n\}\subset\Omega$ by
\[
N_{\mathrm{pack}}^{\mathrm{ext}}(S;\gamma,\xi)
:=
\max\Big\{
|T|:\ T\subset\Omega,\ \operatorname{dist}_\infty(t,S)\le \gamma\ \ \forall t\in T,
\ \|t-t'\|_\infty\ge \xi\ \ \forall t\neq t'\in T
\Big\}.
\]
All packing-growth assumptions in the upper bound are imposed on the realized
empirical support \(S_X\), not on an underlying population distribution.
The empirical support \(S_X\) has packing-growth dimension \(k\) if
\[
N_{\rm pack}^{\rm ext}(S_X;\eta,\eta/2)\lesssim \eta^{-k},
\,\, 0<\eta\le1.
\]
The implicit constant is a geometric regularity constant, independent of
\(n\), \(\varepsilon\), and the algorithmic depth. In the main text, we suppress
this constant and refer to \(k\) as the packing-growth dimension. The appendix
states the corresponding bounds with the explicit constant \(C_{\rm pack}\).
\end{assumption}

\begin{remark}[Neighborhood-packing interpretation]\label{rem:pack-neighborhood}
The external packing number is exactly the packing number of the $\eta$-neighborhood of $S$:
\[
N_{\mathrm{pack}}^{\mathrm{ext}}(S;\eta,\xi)
=
N_{\mathrm{pack}}(S^{(+\eta)},\xi),
\qquad
S^{(+\eta)}:=\{z\in\Omega:\operatorname{dist}_\infty(z,S)\le \eta\}.
\]
Thus Assumption~\ref{ass:pack} may equivalently be viewed as a packing bound on the neighborhoods
$S^{(+\eta_j)}$.
\end{remark}

\subsection{Geometry and occupancy consequences of packing}

At depth $j$, the cyclic binary partition produces cells whose diameters are of order $2^{-j/d}$.
The packing assumption is useful because the centers of occupied depth-$j$ cells form a separated
set inside the $\eta_j$-neighborhood of the data. This allows us to convert geometric control of
$S^{(+\eta_j)}$ into a bound on the number of occupied cells.

\begin{proposition}[Cell diameters]\label{prop:diam}
For every node $\theta$ at depth $|\theta|=j$,
$
\operatorname{diam}(\Omega_\theta)\le 2\cdot 2^{-j/d}.
$

In particular, the leaf resolution
$
\delta:=\max_{|\theta|=r}\operatorname{diam}(\Omega_\theta)
$
satisfies
$
\delta\le 2\cdot 2^{-r/d}.
$
\end{proposition}

\begin{proof}
After $j$ bisections, each coordinate has been split at least $\lfloor j/d\rfloor$ times.
Therefore, every side length of $\Omega_\theta$ is at most
$
2^{-\lfloor j/d\rfloor}.
$
Since $\rho=\|\cdot\|_\infty$, the diameter of an axis-aligned rectangle is its largest side length, so
$
\operatorname{diam}(\Omega_\theta)\le 2^{-\lfloor j/d\rfloor}.
$
Using
$
\lfloor j/d\rfloor \ge j/d-1,
$
we obtain
\[
2^{-\lfloor j/d\rfloor}\le 2^{-(j/d-1)}=2\cdot 2^{-j/d},
\]
which proves the claim. The bound for $\delta$ follows by taking the maximum over all leaves.
\end{proof}

\begin{lemma}[Packing-growth controls occupancy]\label{lem:pack2occ}
Let
\[
\tau_j:=\{\theta\in\{0,1\}^j:\ n_\theta>0\}
\]
be the set of occupied depth-$j$ cells, and let $\eta_j:=2^{-j/d}$. Under
Assumption~\ref{ass:pack},
\[
|\tau_j|
\le
N_{\mathrm{pack}}^{\mathrm{ext}}(S;\eta_j,\eta_j/2)
\le
C_{\mathrm{pack}}\,2^{jk/d}
\qquad\text{for all }j\le r.
\]
\end{lemma}

\begin{proof}
Fix $j\le r$. For each $\theta\in\tau_j$, let $c_\theta$ denote the center of the cell
$\Omega_\theta$. Since $\theta$ is occupied, there exists some $x\in S\cap\Omega_\theta$.
Therefore
$
\|c_\theta-x\|_\infty\le \operatorname{diam}(\Omega_\theta)/2.
$
By Proposition~\ref{prop:diam},
$
\operatorname{diam}(\Omega_\theta)\le 2\eta_j,
$
hence
$
\|c_\theta-x\|_\infty\le \eta_j.
$
Thus
$
\operatorname{dist}_\infty(c_\theta,S)\le \eta_j.
$

It remains to show that distinct occupied-cell centers are $\eta_j/2$-separated.
At depth $j$, each coordinate has been bisected either $\lfloor j/d\rfloor$ or $\lceil j/d\rceil$ times,
so the grid spacing in each coordinate is either
$
2^{-\lfloor j/d\rfloor}
\,\,\text{or}\,\,
2^{-\lceil j/d\rceil}.
$
In particular, every grid step is at least
$
2^{-\lceil j/d\rceil}\ge \eta_j/2.
$
Since all depth-$j$ cells have the same side-length vector, their centers lie on a common rectangular
grid at that depth. If $\theta\neq\theta'$, then the centers $c_\theta$ and $c_{\theta'}$ differ in at least
one coordinate by at least one grid step. Therefore
\[
\|c_\theta-c_{\theta'}\|_\infty \ge 2^{-\lceil j/d\rceil}\ge \eta_j/2.
\]
So the set
$
\{c_\theta:\theta\in\tau_j\}
$
is an admissible external packing for $S$ at scale $\eta_j$ and separation $\eta_j/2$. Hence
\[
|\tau_j|\le N_{\mathrm{pack}}^{\mathrm{ext}}(S;\eta_j,\eta_j/2).
\]
The second inequality follows immediately from Assumption~\ref{ass:pack}.
\end{proof}

\begin{remark}[Why this assumption matches occupancy]\label{rem:pack-natural}
A direct packing assumption on the data set $S$ is not well aligned with the quantity $|\tau_j|$:
two distinct occupied cells may contain points that are arbitrarily close if the points lie near a
partition boundary. The external packing formulation avoids this pathology by working with the
centers of occupied cells, which are automatically separated at the cell scale and remain close to
the data set.
\end{remark}

\subsection{Tools}

We collect here the elementary transport lemmas and the flux-based combinatorial tools
used in the proof of the main accuracy theorem.

\begin{lemma}[Changing cardinality by adding points]\label{lem:card}
Let $U\subseteq V$ be finite multisets in $\Omega$ with
$
1\le |U|\le |V|=:N.
$

Then
\[
W_1(\mu_U,\mu_V)\le \frac{|V\setminus U|}{N}\,\operatorname{diam}(\Omega).
\]
\end{lemma}

\begin{proof}
Let $M:=|U|$ and write the multiset decomposition
\[
V = U \uplus R,
\qquad
R:=V\setminus U,
\qquad
|R|=N-M.
\]
Recall that
$
\mu_V=\frac{1}{N}\sum_{v\in V}\delta_v,
\,\,
\mu_U=\frac{1}{M}\sum_{u\in U}\delta_u.
$
Construct a coupling $\pi\in\Pi(\mu_V,\mu_U)$ as follows:
\begin{itemize}
\item for each copy of $u\in U$ inside $V$, keep its mass $1/N$ at $u$;
\item for each $v\in R$, split its mass $1/N$ uniformly over the $M$ points of $U$,
sending mass $1/(NM)$ from $v$ to each $u\in U$.
\end{itemize}
Then the first marginal of $\pi$ is $\mu_V$ by construction. Moreover, each $u\in U$
receives total mass
\[
\frac{1}{N}+|R|\cdot \frac{1}{NM}
=
\frac{1}{N}+\frac{N-M}{NM}
=
\frac{1}{M},
\]
so the second marginal is exactly $\mu_U$.

Only the mass originating from $R$ moves nontrivially. Each $v\in R$ transports total mass $1/N$,
and every transported piece travels distance at most $\operatorname{diam}(\Omega)$. Therefore
\[
\int_{\Omega\times\Omega}\rho(x,y)\,d\pi(x,y)
\le
|R|\cdot \frac{1}{N}\operatorname{diam}(\Omega)
=
\frac{|V\setminus U|}{N}\operatorname{diam}(\Omega).
\]
Taking the infimum over couplings proves the claim.
\end{proof}

\begin{lemma}[Relocation within leaves]\label{lem:leaf}
Let
$
\delta:=\max_{|\theta|=r}\operatorname{diam}(\Omega_\theta).
$
Let $A,B\subset\Omega$ be finite multisets with $|A|=|B|\ge 1,$
and suppose that for every leaf cell $\Omega_\theta$ with $|\theta|=r$,
$
|A\cap\Omega_\theta|=|B\cap\Omega_\theta|.
$

Then
\[
W_1(\mu_A,\mu_B)\le \delta.
\]
\end{lemma}

\begin{proof}
Inside each leaf $\Omega_\theta$, match the points of $A\cap\Omega_\theta$ to the points of
$B\cap\Omega_\theta$. This is possible because the leaf counts agree. The resulting matching
defines a coupling between $\mu_A$ and $\mu_B$ in which every matched pair lies in the same
leaf cell, hence moves by distance at most
$
\operatorname{diam}(\Omega_\theta)\le \delta.
$
Therefore, the transport cost of this coupling is at most $\delta$, and so
$
W_1(\mu_A,\mu_B)\le \delta.
$
\end{proof}

\paragraph{Flux.}
For $a,b\in\mathbb Z_+^2$, define
\[
\mathrm{flux}(a,b):=
\begin{cases}
0, & \text{if $a$ and $b$ are comparable},\\[2mm]
\min\{|a_0-b_0|,\ |a_1-b_1|\}, & \text{otherwise}.
\end{cases}
\]

\begin{lemma}[Flux as incomparability]\label{lem:fluxdist}
For all $a,b\in\mathbb Z_+^2$, the quantity $\mathrm{flux}(a,b)$ equals the
$\ell_\infty$-distance from $a$ to the set of vectors in $\mathbb Z_+^2$ that are comparable to $b$.
\end{lemma}

\begin{proof}
If $a$ and $b$ are comparable, then both quantities are zero.

Otherwise, without loss of generality, assume
$
a_0>b_0,
\,\,
a_1<b_1.
$
The set of vectors comparable to $b$ is
\[
\{x\in\mathbb Z_+^2:\ x_0\le b_0,\ x_1\le b_1\}
\;\cup\;
\{x\in\mathbb Z_+^2:\ x_0\ge b_0,\ x_1\ge b_1\}.
\]
The $\ell_\infty$-distance from $a$ to the first set is $a_0-b_0$, while the
$\ell_\infty$-distance from $a$ to the second set is $b_1-a_1$. Therefore the distance from
$a$ to the union is
\[
\min\{a_0-b_0,\ b_1-a_1\}
=
\min\{|a_0-b_0|,\ |a_1-b_1|\}
=
\mathrm{flux}(a,b).
\]
\end{proof}

\begin{lemma}[Flux is controlled by noise]\label{lem:fluxnoise}
For every internal node $\theta$,
\[
\mathrm{flux}\big((n_{\theta0},n_{\theta1}),(m_{\theta0},m_{\theta1})\big)
\le
\max\{|\lambda_{\theta0}|,\ |\lambda_{\theta1}|\}.
\]
\end{lemma}

\begin{proof}
If \(m_\theta=0\), consistency gives
\(m_{\theta0}=m_{\theta1}=0\), so the flux is zero.
Otherwise, \(\theta\) is expanded and its corrected child pair
is obtained from Algorithm~\ref{alg:pruned_consistency}.
Set
\[
a:=(n_{\theta0},n_{\theta1}),
\qquad
b':=(n'_{\theta0},n'_{\theta1})
=
\big((n_{\theta0}+\lambda_{\theta0})_+,\ (n_{\theta1}+\lambda_{\theta1})_+\big),
\]
and
$
b:=(m_{\theta0},m_{\theta1}).
$
Since the map $x\mapsto x_+$ is $1$-Lipschitz,
$
\|a-b'\|_\infty\le \max\{|\lambda_{\theta0}|,\ |\lambda_{\theta1}|\}.
$
If $a$ and $b$ are comparable, then
$
\mathrm{flux}(a,b)=0
$
and the claim is immediate.

Otherwise, $a$ and $b$ are incomparable. By the consistency rule in
Algorithm~\ref{alg:pruned_consistency}, the chosen vector $b$ is comparable to the noisy vector $b'$.
Hence, $b'$ belongs to the set of vectors comparable to $b$. By Lemma~\ref{lem:fluxdist},
\[
\mathrm{flux}(a,b)
\le
\|a-b'\|_\infty
\le
\max\{|\lambda_{\theta0}|,\ |\lambda_{\theta1}|\},
\]
as claimed.
\end{proof}

\begin{lemma}[Empty nodes contribute zero flux\citep{He2023Algorithmically}]\label{lem:empty}
If $n_\theta=0$, then
\[
\mathrm{flux}\big((n_{\theta0},n_{\theta1}),(m_{\theta0},m_{\theta1})\big)=0.
\]
\end{lemma}

\begin{proof}
If $n_\theta=0$, then necessarily
$
(n_{\theta0},n_{\theta1})=(0,0).
$
Since $(0,0)$ is comparable to every vector in $\mathbb Z_+^2$, the definition of flux gives
$
\mathrm{flux}\big((0,0),(m_{\theta0},m_{\theta1})\big)=0.
$
\end{proof}

\begin{lemma}[Two-bin flux transfer\citep{He2023Algorithmically}]\label{lem:twobin}
Let $a=(a_0,a_1)$ and $b=(b_0,b_1)$ be in $\mathbb Z_+^2$. Then one can transform the
occupancy vector $a$ into $b$ by:
\begin{enumerate}
\item adding $(b_0+b_1)-(a_0+a_1)$ balls in total to the two bins
(removing if this number is negative), and
\item transferring exactly $\mathrm{flux}(a,b)$ balls from one bin to the other.
\end{enumerate}
\end{lemma}

\begin{proof}
If $a$ and $b$ are comparable, then either $a\preceq b$ or $b\preceq a$.
If $a\preceq b$, add $b_i-a_i$ balls to bin $i$ for $i=0,1$.
If $b\preceq a$, remove $a_i-b_i$ balls from bin $i$ for $i=0,1$.
In both cases, no transfer is needed and $\mathrm{flux}(a,b)=0$.

Now suppose $a$ and $b$ are incomparable. Without loss of generality, assume
$
a_0\ge b_0,
\,\,
a_1\le b_1.
$
Set
$
M:=\min\{a_0-b_0,\ b_1-a_1\}=\mathrm{flux}(a,b).
$

\emph{Case 1:} $a_0-b_0\ge b_1-a_1$. Then $M=b_1-a_1$ and
\[
\Delta:=(a_0+a_1)-(b_0+b_1)=(a_0-b_0)-(b_1-a_1)\ge 0.
\]
First remove $\Delta$ balls from bin $0$, obtaining $(a_0-\Delta,a_1)$.
Then transfer $M$ balls from bin $0$ to bin $1$. The final state is
\[
(a_0-\Delta-M,\ a_1+M)=(b_0,b_1),
\]
because $\Delta+M=a_0-b_0$ and $M=b_1-a_1$.

\emph{Case 2:} $a_0-b_0\le b_1-a_1$. Then $M=a_0-b_0$ and
\[
\Delta:=(b_0+b_1)-(a_0+a_1)=(b_1-a_1)-(a_0-b_0)\ge 0.
\]
First add $\Delta$ balls to bin $1$, obtaining $(a_0,a_1+\Delta)$.
Then transfer $M$ balls from bin $0$ to bin $1$. The final state is
\[
(a_0-M,\ a_1+\Delta+M)=(b_0,b_1),
\]
because $M=a_0-b_0$ and $\Delta+M=b_1-a_1$.

In both cases, the construction uses the correct total add/remove amount and then transfers
exactly $M=\mathrm{flux}(a,b)$ balls.
\end{proof}

\begin{lemma}[Hierarchical transformation by flux \citep{He2023Algorithmically}]\label{lem:hier}
Let $X\in\Omega^n$ and let $(n_\theta)_{|\theta|\le r}$ be its counts in the partition
$\{\Omega_\theta\}_{\theta\in\{0,1\}^{\le r}}$. Let $(m_\theta)_{|\theta|\le r}$ be any
consistent collection of nonnegative integers, i.e.
\[
m_\theta=m_{\theta0}+m_{\theta1}
\qquad\text{for all }|\theta|\le r-1,
\]
with root mass $m:=m_\varnothing$.

For each internal node $\theta$, define
\[
t_\theta:=\mathrm{flux}\big((n_{\theta0},n_{\theta1}),(m_{\theta0},m_{\theta1})\big).
\]
Then one can transform $X$ into a multiset $Z\subset\Omega$ of size $m$ whose leaf counts
are $(m_\theta)_{|\theta|=r}$ by:
\begin{enumerate}
\item adding or removing a total of $|m-n|$ points in $\Omega$, and
\item for every internal node $\theta$ with $|\theta|\le r-1$, transferring exactly $t_\theta$
points between $\Omega_{\theta0}$ and $\Omega_{\theta1}$.
\end{enumerate}
Moreover, every transfer at node $\theta$ takes place entirely inside $\Omega_\theta$.
\end{lemma}

\begin{proof}
We first construct a matching bottom-up, without moving any points.

\emph{Step 1: matching within leaves.}
In each leaf $\Omega_\theta$, create $m_\theta$ target slots and match
$\min\{n_\theta,m_\theta\}$ original points to distinct slots in that leaf.
These points will remain in place. The unmatched objects in the leaf are either
$n_\theta-m_\theta$ original points or $m_\theta-n_\theta$ target slots.

\emph{Step 2: matching across sibling subtrees.}
Proceed upward through the internal nodes. Inductively, after processing a child
$\theta i$, its unmatched objects are all of one kind, and their signed number
(original points minus target slots) is $n_{\theta i}-m_{\theta i}$.
If the two children have unmatched objects of opposite kinds, pair original points
from one child with target slots from the other until one side is exhausted.
The number of pairs formed at $\theta$ is then
\[
\min\{|n_{\theta0}-m_{\theta0}|,|n_{\theta1}-m_{\theta1}|\}
=t_\theta.
\]
If the children do not have unmatched objects of opposite kinds, no pair is formed,
and $t_\theta=0$. In either case, the remaining unmatched objects at $\theta$
are all of one kind, with signed number $n_\theta-m_\theta$ by consistency.
This proves the induction. At the root, there remain exactly $|m-n|$ unmatched
objects: original points if $n>m$, or target slots if $m>n$.

\emph{Step 3: realizing the matching.}
First remove all unmatched original points, or fill every unmatched target slot
by adding a point in its leaf. This makes exactly $|m-n|$ additions or removals.
For each pair formed at an internal node $\theta$, move its original point directly
to any point in the leaf of its matched target slot. Each original point is used
in at most one pair, so these transfers can be performed in any order after the
cardinality correction. Exactly $t_\theta$ such transfers join opposite children
of $\theta$, and both endpoints lie in $\Omega_\theta$.
Every target slot is now filled, so the resulting multiset $Z$ has the prescribed
leaf counts and size $m$.
\end{proof}

\subsection{Accuracy}

The transport error of Pruned-PMM is driven by two effects:
the mismatch between the total output size and the input size, and the mass that must be
transferred across sibling cells in order to realize the consistent counts.
The second contribution is localized to occupied nodes and can therefore be controlled by the
packing-based occupancy bound.

For each depth $j$, define the occupied set and occupancy count
\[
\tau_j:=\{\theta\in\{0,1\}^j:\ n_\theta>0\},
\qquad
O_j:=|\tau_j|.
\]
Under Assumption~\ref{ass:pack}, Lemma~\ref{lem:pack2occ} yields
\begin{equation}\label{eq:occ}
O_j=|\tau_j|\le C_{\mathrm{pack}}\,2^{jk/d}
\qquad (0\le j\le r).
\end{equation}

\paragraph{Transport budget.}
For each internal node $\theta$, let
$
t_\theta:=\mathrm{flux}\big((n_{\theta0},n_{\theta1}),(m_{\theta0},m_{\theta1})\big).
$
Since a transfer across the split of $\Omega_\theta$ can be carried out entirely within the parent cell
$\Omega_\theta$, each such moved point travels distance at most $\operatorname{diam}(\Omega_\theta)$.
We therefore define the total transport budget
\begin{equation}\label{eq:Ddef}
D:=\sum_{|\theta|\le r-1} t_\theta\,\operatorname{diam}(\Omega_\theta).
\end{equation}

\begin{theorem}[Deterministic transport bound]\label{thm:acc-det}
Let $\delta:=\max_{|\theta|=r}\operatorname{diam}(\Omega_\theta)$ be the leaf resolution.
Then for the Pruned-PMM output $Y$,
\[
W_1(\mu_X,\mu_Y)
\le
\frac{|m-n|}{n}\operatorname{diam}(\Omega)+\frac{D}{n}+\delta.
\]
In particular, since $\operatorname{diam}(\Omega)=1$ for $\Omega=[0,1]^d$ equipped with
$\|\cdot\|_\infty$,
\[
W_1(\mu_X,\mu_Y)\le \frac{|m-n|}{n}+\frac{D}{n}+\delta.
\]
\end{theorem}

\begin{proof}
Let $Z$ be the multiset provided by Lemma~\ref{lem:hier}, so that $Z$ has the same leaf counts as $Y$.
By Lemma~\ref{lem:hier}, the transformation from $X$ to $Z$ consists of two parts:
\begin{enumerate}
\item adding or removing a total of $|m-n|$ points;
\item for every internal node $\theta$, transferring exactly $t_\theta$ points between
$\Omega_{\theta0}$ and $\Omega_{\theta1}$, each transfer taking place inside $\Omega_\theta$.
\end{enumerate}

Let $X_1$ denote the intermediate multiset after the cardinality correction step, so $|X_1|=m$.

\emph{Case 1: $m\ge n$.}
Then $X\subseteq X_1$ as multisets, and Lemma~\ref{lem:card} gives
\[
W_1(\mu_X,\mu_{X_1})
\le
\frac{m-n}{m}\operatorname{diam}(\Omega)
\le
\frac{m-n}{n}\operatorname{diam}(\Omega).
\]
The transfer step moves points with total distance at most $D$ among $m$ points, so
$
W_1(\mu_{X_1},\mu_Z)\le \frac{D}{m}\le \frac{D}{n}.
$
Hence
\[
W_1(\mu_X,\mu_Z)
\le
\frac{|m-n|}{n}\operatorname{diam}(\Omega)+\frac{D}{n}.
\]
\emph{Case 2: $m<n$.}
Let $X_0$ be the multiset of removed points, so that
$
|X_0|=n-m,
\qquad
X=X_1\uplus X_0.
$
The transfer step moves points of $X_1$ to obtain $Z$ with total distance at most $D$,
while leaving $X_0$ unchanged. Therefore
$
W_1(\mu_X,\mu_{Z\uplus X_0})\le \frac{D}{n}.
$
Since $Z\subseteq Z\uplus X_0$ and $|Z\uplus X_0|=n$, Lemma~\ref{lem:card} gives
\[
W_1(\mu_Z,\mu_{Z\uplus X_0})
\le
\frac{|X_0|}{n}\operatorname{diam}(\Omega)
=
\frac{n-m}{n}\operatorname{diam}(\Omega).
\]

In both cases,
\[
W_1(\mu_X,\mu_Z)
\le
\frac{|m-n|}{n}\operatorname{diam}(\Omega)+\frac{D}{n}.
\]
Finally, $Z$ and $Y$ have the same leaf counts, so Lemma~\ref{lem:leaf} implies
$
W_1(\mu_Z,\mu_Y)\le \delta.
$
By the triangle inequality,
\[
W_1(\mu_X,\mu_Y)
\le
\frac{|m-n|}{n}\operatorname{diam}(\Omega)+\frac{D}{n}+\delta.
\]
\end{proof}

\begin{corollary}[Expected transport bound under occupancy growth]\label{cor:acc-exp}
With \(O_j:=|\tau_j|\), one has
$
\mathbb E D
\le
4C_{\mathrm{abs}}\sum_{j=0}^{r-1}
\sigma_{j+1}O_j2^{-j/d},
$
and therefore
\[
\mathbb E W_1(\mu_X,\mu_Y)
\le
\frac{\mathbb E|m-n|}{n}
+
\frac{4C_{\mathrm{abs}}}{n}\sum_{j=0}^{r-1}
\sigma_{j+1}O_j2^{-j/d}
+
\delta.
\]
\end{corollary}

\begin{proof}
By Lemma~\ref{lem:empty}, only occupied nodes contribute to $D$. Therefore
$
\mathbb E D
=
\sum_{j=0}^{r-1}\sum_{\theta\in\tau_j}\mathbb E[t_\theta]\operatorname{diam}(\Omega_\theta).
$
Fix $j$ and $\theta\in\tau_j$. By Lemma~\ref{lem:fluxnoise},
$
t_\theta
\le
\max\{|\lambda_{\theta0}|,|\lambda_{\theta1}|\}
\le
|\lambda_{\theta0}|+|\lambda_{\theta1}|.
$
Taking expectation and using Assumption~\ref{ass:noise},
$
\mathbb E t_\theta
\le
2\,\mathbb E|\lambda|
\le
2C_{\mathrm{abs}}\sigma_{j+1}.
$
Hence
\[
\mathbb E D
\le
\sum_{j=0}^{r-1}
(2C_{\mathrm{abs}}\sigma_{j+1})
\sum_{\theta\in\tau_j}\operatorname{diam}(\Omega_\theta).
\]
Using Proposition~\ref{prop:diam} and \eqref{eq:occ},
\[
\sum_{\theta\in\tau_j}\operatorname{diam}(\Omega_\theta)
\le
O_j\cdot \max_{|\theta|=j}\operatorname{diam}(\Omega_\theta)
\le
2O_j\,2^{-j/d}
\le
(C_{\mathrm{pack}}2^{jk/d})(2\cdot 2^{-j/d})
=
2C_{\mathrm{pack}}2^{j(k-1)/d}.
\]
This proves
\[
\mathbb E D
\le
4C_{\mathrm{abs}}
\sum_{j=0}^{r-1}\sigma_{j+1}O_j2^{-j/d}.
\]
\end{proof}
Under Assumption~\ref{ass:pack}, using
\(O_j\le C_{\mathrm{pack}}2^{jk/d}\) gives the packing-growth bound
\[
\mathbb E D
\le
4C_{\mathrm{abs}}C_{\mathrm{pack}}
\sum_{j=0}^{r-1}\sigma_{j+1}2^{j(k-1)/d},
\]
and hence
\[
\mathbb E W_1(\mu_X,\mu_Y)
\le
\frac{\mathbb E|m-n|}{n}
+
\frac{4C_{\mathrm{abs}}C_{\mathrm{pack}}}{n}
\sum_{j=0}^{r-1}\sigma_{j+1}2^{j(k-1)/d}
+
\delta,
\qquad
\delta\le 2\cdot 2^{-r/d}.
\]

\subsection{\texorpdfstring{Optimizing the upper bound over $(\sigma_j)$ and the depth $r$}{Optimizing the upper bound over (sigma j) and the depth r}}

We now optimize the expected upper bound from Corollary~\ref{cor:acc-exp} over the
noise magnitudes $(\sigma_j)_{j=1}^r$ and the depth $r$, while treating the root-noise
magnitude $\sigma_0$ as a fixed design parameter.

\paragraph{Tight use of the privacy budget.}
Under replacement adjacency, Theorem~\ref{thm:privacy} gives the privacy constraint
\begin{equation}\label{eq:DPtight}
2\sum_{j=1}^r \frac{1}{\sigma_j}=\varepsilon.
\end{equation}
This may be assumed without loss of generality when optimizing over
$\sigma_1,\dots,\sigma_r$, since the upper bound is increasing in each of these variables.

\paragraph{Upper bound to optimize.}
For the root-mass contribution, since $n\ge1$ and
$m=\max\{1,n+\lambda_\varnothing\}$,
$
|m-n|
=
\big|\max\{1,n+\lambda_\varnothing\}-n\big|
\le
|\lambda_\varnothing|.
$
Therefore, by Assumption~\ref{ass:noise},
$
\mathbb E|m-n|
\le
\mathbb E|\lambda_\varnothing|
\le
C_{\mathrm{abs}}\sigma_0.
$
Combining this with Corollary~\ref{cor:acc-exp} gives
\begin{equation}\label{eq:acc-linear}
\mathbb E W_1(\mu_X,\mu_Y)
\le
\frac{C_{\mathrm{abs}}}{n}\sigma_0
+
\frac{4C_{\mathrm{abs}}C_{\mathrm{pack}}}{n}
\sum_{j=1}^{r}\sigma_j\,2^{(j-1)(k-1)/d}
+
\delta,
\qquad
\delta\le 2\cdot 2^{-r/d}.
\end{equation}
Before specializing to the packing-growth envelope, we first optimize the sharper
occupancy-profile bound
\[
\mathbb E W_1(\mu_X,\mu_Y)
\le
\frac{C_{\mathrm{abs}}}{n}\sigma_0
+
\frac{4C_{\mathrm{abs}}}{n}
\sum_{j=1}^{r}\sigma_j O_{j-1}2^{-(j-1)/d}
+
\delta .
\]
Thus, for fixed $r$ and fixed $\sigma_0$, it remains to minimize a weighted linear function
of $(\sigma_j)_{j=1}^r$ subject to the privacy constraint \eqref{eq:DPtight}.

\begin{proposition}[Optimized occupancy-profile bound]\label{prop:opts}
Fix \(r\ge 1\), define
$
O^{(r)}:=(O_0,\dots,O_{r-1}),
$
and set
$
A_r(O^{(r)})
:=
2\sum_{j=0}^{r-1}\sqrt{O_j}\,2^{-j/(2d)}.
$
Choose \(\sigma_0=\varepsilon^{-1}\). Among all choices
\(\sigma_1,\dots,\sigma_r>0\) satisfying
$
2\sum_{i=1}^r\frac1{\sigma_i}=\varepsilon,
$
the occupancy-profile upper bound is minimized uniquely at
\[
\sigma_i^\star
=
\frac{2A_r(O^{(r)})}
{\varepsilon\sqrt{4O_{i-1}2^{-(i-1)/d}}},
\qquad i=1,\dots,r.
\]
Moreover, with this choice,
\begin{equation*}\label{eq:opt-bound-r}
\mathbb E W_1(\mu_X,\mu_Y)
\le
B_r(O^{(r)}),
\tag{Opt}
\end{equation*}
where
\[
B_r(O^{(r)})
:=
\frac{C_{\mathrm{abs}}}{\varepsilon n}
+
\frac{2C_{\mathrm{abs}}}{\varepsilon n}
\left(
2\sum_{j=0}^{r-1}\sqrt{O_j}\,2^{-j/(2d)}
\right)^2
+
2\cdot 2^{-r/d}.
\]
\end{proposition}

\begin{proof}
Using Corollary~\ref{cor:acc-exp} and the choice \(\sigma_0=\varepsilon^{-1}\), it remains to minimize
$
\sum_{i=1}^r b_i\sigma_i,
\,
b_i:=4O_{i-1}2^{-(i-1)/d},
$
subject to
$
2\sum_{i=1}^r\frac1{\sigma_i}=\varepsilon.
$
Since \(O_{i-1}\ge 1\), all weights \(b_i\) are strictly positive. By Cauchy--Schwarz,
\[
\left(\sum_{i=1}^r b_i\sigma_i\right)
\left(\sum_{i=1}^r\frac1{\sigma_i}\right)
\ge
\left(\sum_{i=1}^r\sqrt{b_i}\right)^2.
\]
Under the constraint \(2\sum_{i=1}^r\sigma_i^{-1}=\varepsilon\), this gives
\[
\sum_{i=1}^r b_i\sigma_i
\ge
\frac{2}{\varepsilon}
\left(\sum_{i=1}^r\sqrt{b_i}\right)^2.
\]
Equality holds if and only if
$
\sigma_i=\frac{c}{\sqrt{b_i}}
\, (i=1,\dots,r)
$
for some constant \(c>0\). Imposing the privacy constraint gives
\[
2\sum_{i=1}^r\frac{\sqrt{b_i}}{c}=\varepsilon,
\qquad\text{hence}\qquad
c=\frac{2\sum_{i=1}^r\sqrt{b_i}}{\varepsilon}.
\]
Since
\[
\sum_{i=1}^r\sqrt{b_i}
=
2\sum_{j=0}^{r-1}\sqrt{O_j}\,2^{-j/(2d)}
=
A_r(O^{(r)}),
\]
the optimizer is
\[
\sigma_i^\star
=
\frac{2A_r(O^{(r)})}
{\varepsilon\sqrt{4O_{i-1}2^{-(i-1)/d}}},
\qquad i=1,\dots,r,
\]
and the minimum value is
\[
\frac{2A_r(O^{(r)})^2}{\varepsilon}.
\]
Multiplying by \(C_{\mathrm{abs}}/n\), adding the root contribution
\(C_{\mathrm{abs}}/(\varepsilon n)\), and using
\(\delta\le 2\cdot 2^{-r/d}\) gives \eqref{eq:opt-bound-r}.
\end{proof}

Proposition~\ref{prop:opts} gives the formally optimal schedule for a known occupancy
profile. Since this profile is data-dependent, we next use the packing envelope to define
a public schedule.

\paragraph{Specialization under packing-growth.}
Under Assumption~\ref{ass:pack}, we have
$
O_j\le C_{\mathrm{pack}}2^{jk/d}
\,\, (0\le j\le r-1).
$
Therefore
\[
A_r(O^{(r)})
=
2\sum_{j=0}^{r-1}\sqrt{O_j}\,2^{-j/(2d)}
\le
2\sqrt{C_{\mathrm{pack}}}
\sum_{j=0}^{r-1}2^{j(k-1)/(2d)}.
\]
Let
$
q:=2^{(k-1)/(2d)}\ge 1,
$
\[
S_r:=\sum_{j=0}^{r-1} q^j
=
\begin{cases}
r, & k=1,\\[2mm]
\dfrac{q^r-1}{q-1}, & k>1.
\end{cases}
\]
Then
$
A_r(O^{(r)})\le 2\sqrt{C_{\mathrm{pack}}}\,S_r.
$
Consequently,
\begin{equation}\label{eq:Bclosed}
B_r(O^{(r)})
\le
\frac{C_{\mathrm{abs}}}{\varepsilon n}
+
\frac{8C_{\mathrm{abs}}C_{\mathrm{pack}}}{\varepsilon n}S_r^2
+
2\cdot 2^{-r/d}.
\end{equation}
Note that, by using \(\sigma_i:=\frac{2S_r}{\varepsilon q^{i-1}}\), we can achieve the same upper bound. For this public schedule, applying the packing bound directly gives
\[
\frac{4C_{\mathrm{abs}}C_{\mathrm{pack}}}{n}
\sum_{i=1}^r \sigma_i 2^{(i-1)(k-1)/d}
=
\frac{4C_{\mathrm{abs}}C_{\mathrm{pack}}}{n}
\sum_{i=1}^r
\frac{2S_r}{\varepsilon q^{i-1}}q^{2(i-1)}
=
\frac{8C_{\mathrm{abs}}C_{\mathrm{pack}}}{\varepsilon n}S_r^2.
\]
Hence,
\[
\mathbb E W_1(\mu_X,\mu_Y)
\le
\frac{C_{\mathrm{abs}}}{\varepsilon n}
+
\frac{8C_{\mathrm{abs}}C_{\mathrm{pack}}}{\varepsilon n}S_r^2
+
2\cdot 2^{-r/d}.
\]

\begin{corollary}[Fixed-depth Pruned-PMM under packing-growth]\label{cor:rates}
Assume \(\varepsilon n\ge 2\) and suppose Assumption~\ref{ass:pack} holds with
packing-growth dimension \(k\in\{1,\dots,d\}\) and constant \(C_{\mathrm{pack}}\). Set
\(\sigma_0=\varepsilon^{-1}\). Let \(q:=2^{(k-1)/(2d)}\) and
\(S_r:=\sum_{j=0}^{r-1}q^j\). Choose the public schedule
\(\sigma_i:=\frac{2S_r}{\varepsilon q^{i-1}}\), \(i=1,\dots,r\). Then
\(2\sum_{i=1}^r\sigma_i^{-1}=\varepsilon\), so Pruned-PMM is
\(\varepsilon\)-differentially private.

\smallskip
\noindent
\textbf{Case \(k=1\).}
Take \(r:=\left\lceil d\log_2(\varepsilon n)\right\rceil\). Then
\[
\mathbb E W_1(\mu_X,\mu_Y)
=
O\!\left(
\frac{d^2\log^2(\varepsilon n)}
{\varepsilon n}
\right).
\]

\smallskip
\noindent
\textbf{Case \(2\le k\le d\).}
Put \(\alpha:=2^{(k-1)/(2d)}-1\), set \(M:=\varepsilon n\alpha^2\), and take
$
r:=
\max\left\{
1,
\left\lfloor
\frac{d}{k}\log_2(1+M)
\right\rfloor
\right\}.
$
Then
\[
\mathbb E W_1(\mu_X,\mu_Y)
=
O\!\left(
(\varepsilon n\alpha^2)^{-1/k}
\right).
\]
Equivalently, since \(\alpha\asymp (k-1)/d\),
\[
\mathbb E W_1(\mu_X,\mu_Y)
=
O\!\left(
\left(\frac{d}{k-1}\right)^{2/k}
(\varepsilon n)^{-1/k}
\right).
\]
The constants depend only on \(C_{\mathrm{pack}}\).
\end{corollary}

\begin{proof}
First,
$
2\sum_{i=1}^r\frac1{\sigma_i}
=
2\sum_{i=1}^r
\frac{\varepsilon q^{i-1}}{2S_r}
=
\frac{\varepsilon}{S_r}\sum_{i=1}^r q^{i-1}
=
\frac{\varepsilon}{S_r}S_r
=
\varepsilon.
$
Thus, the schedule uses the privacy budget tightly.

\smallskip
\noindent
\textbf{Case $k=1$.}
Here $q=1$ and hence $S_r=r$. With
$r=\lceil d\log_2(\varepsilon n)\rceil$, we have
\[
2^{-r/d}
\le
2^{-\log_2(\varepsilon n)}
=
(\varepsilon n)^{-1}.
\]
Therefore the bound for the public schedule gives
\[
\mathbb E W_1(\mu_X,\mu_Y)
\le
\frac{C_{\mathrm{abs}}}{\varepsilon n}
+
\frac{8C_{\mathrm{abs}}C_{\mathrm{pack}}}{\varepsilon n}r^2
+
\frac{2}{\varepsilon n}.
\]
Moreover,
$
r
=
\lceil d\log_2(\varepsilon n)\rceil
\le
d\log_2(\varepsilon n)+1
=
O\!\left(d\log(\varepsilon n)\right),
$
since $\varepsilon n\ge2$. Hence
\[
\mathbb E W_1(\mu_X,\mu_Y)
=
O\!\left(
\frac{d^2\log^2(\varepsilon n)}
{\varepsilon n}
\right).
\]
This proves the claimed bound for $k=1$.

\smallskip
\noindent
\smallskip
\noindent
\textbf{Case \(2\le k\le d\).}
Let $\alpha:=q-1=2^{(k-1)/(2d)}-1$ and set
$
M:=\varepsilon n\alpha^2.
$
It is enough to prove the claim when $M\ge1$. Indeed, if $M<1$, then
$M^{-1/k}>1$, while $W_1(\mu_X,\mu_Y)\le\operatorname{diam}(\Omega)=1$.

Assume therefore that $M\ge1$. Since $d/k\ge1$ and $\log_2(1+M)\ge1$,
the lower truncation in the definition of $r$ is harmless.
With
$
r=\left\lfloor \frac{d}{k}\log_2 (M+1)\right\rfloor ,
$
we have
$
r\le \frac{d}{k}\log_2 (M+1),
\,\,
r\ge \frac{d}{k}\log_2 (M+1)-1.
$
Since $k>1$, $q>1$, and
\[
S_r=\sum_{j=0}^{r-1}q^j
=
\frac{q^r-1}{q-1}
\le
\frac{q^r}{\alpha}.
\]
Therefore
$
S_r^2
\le
\frac{q^{2r}}{\alpha^2}
=
\frac{2^{r(k-1)/d}}{\alpha^2}.
$
Using the upper bound on $r$,
$
2^{r(k-1)/d}
\le
2^{\frac{k-1}{k}\log_2(M+1)}
=
(M+1)^{(k-1)/k}
\le
2M^{(k-1)/k}.
$
Hence
$
\frac{8C_{\mathrm{abs}}C_{\mathrm{pack}}}{\varepsilon n}S_r^2
\le
16C_{\mathrm{abs}}C_{\mathrm{pack}}
M^{-1/k}.
$
For the resolution term, the lower bound on $r$ gives
$
2^{-r/d}
\le
2^{1/d}(M+1)^{-1/k}
\le
2^{1/d}M^{-1/k}.
$
Thus
$
2\cdot 2^{-r/d}
\le
2^{1+1/d}M^{-1/k}.
$

Finally, since \(2\le k\le d\),
$
0<\alpha=2^{(k-1)/(2d)}-1\le 2^{1/2}-1<1.
$
Thus $\alpha^2\le1$, and since $M\ge1$,
\[
\frac{C_{\mathrm{abs}}}{\varepsilon n}
=
C_{\mathrm{abs}}\frac{\alpha^2}{M}
\le
C_{\mathrm{abs}}M^{-1}
\le
C_{\mathrm{abs}}M^{-1/k}.
\]
Combining these estimates in the bound for the public schedule,
\[
\mathbb E W_1(\mu_X,\mu_Y)
=
O\!\left(
M^{-1/k}
\right)
=
O\!\left(
(\varepsilon n\alpha^2)^{-1/k}
\right).
\]

Finally, since $0<(k-1)/(2d)\le1/2$, the estimate $2^x-1\asymp x$ on
$[0,1/2]$ gives
$
\alpha
=
2^{(k-1)/(2d)}-1
\asymp
\frac{k-1}{d}.
$
Therefore
\[
(\varepsilon n\alpha^2)^{-1/k}
\asymp
\left(\frac{d}{k-1}\right)^{2/k}
(\varepsilon n)^{-1/k}.
\]
This proves the claimed $k>1$ bound.
\end{proof}

\subsection{Expected runtime with pruning}
\label{subsec:runtime-pruning}

A node is visited when its noisy count is materialized.
Only nodes with positive corrected counts are expanded.
We use sparse data-point lists, partitioned when a node is
expanded. Under the unit-cost arithmetic and noise-sampling model,
local consistency and bookkeeping take constant time per node.
Cell geometry is maintained by single-coordinate updates.
We include the cost of writing the synthetic points explicitly.

\begin{theorem}[Runtime under pruning]
\label{thm:time}
Let \(n\ge1\), \(r\ge1\), and
\(m=\max\{1,n+\lambda_\varnothing\}\), where
\(\lambda_\varnothing\sim\mathrm{LapZ}(\sigma_0)\).
The number of visited nodes satisfies
\[
N_{\mathrm{vis}}\le1+2mr.
\]
Under the implementation above,
\[
T=O(nr+mr+dm),
\qquad
\mathbb E T=O\!\left(nr+(n+\sigma_0)(r+d)\right).
\]
No packing-growth assumption is required.
\end{theorem}

\begin{proof}
At every level, the positive integer corrected counts sum to \(m\).
Thus \(|A_j|\le m\), and
\[
N_{\mathrm{vis}}
=1+2\sum_{j=0}^{r-1}|A_j|
\le1+2mr.
\]
The expanded cells at each depth are disjoint, so routing their
data-point lists costs at most \(O(n)\) per level. Tree operations
cost \(O(N_{\mathrm{vis}})\), and writing the output costs \(O(dm)\).
This proves the pathwise bound.

Since \(n\ge1\), we have
\(m\le n+(\lambda_\varnothing)_+\). Symmetry and
Assumption~\ref{ass:noise} give
\[
\mathbb E m
\le n+\tfrac12\mathbb E|\lambda_\varnothing|
\le n+\tfrac12\sigma_0.
\]
Taking expectations proves the result.
\end{proof}

\begin{corollary}[Runtime with calibrated root noise]
\label{cor:time-intrinsic}
If \(\sigma_0=\varepsilon^{-1}\) and \(\varepsilon n\ge2\),
then fixed-depth Pruned-PMM satisfies
\[
\mathbb E T=O(nr+dn).
\]
For Adaptive Pruned-PMM, let \(L:=\max_{a\in\mathcal A}r_a\).
With \(\sigma_0=\varepsilon_{\rm main}^{-1}\) and
\(\varepsilon_{\rm main}n\ge2\), including occupancy computation
and private selection gives
\[
\mathbb E T_{\rm adaptive}
=O\!\left((n+|\mathcal A|)L+dn\right).
\]
For the standard candidate set, \(|\mathcal A|=d\) and
\(L=\lfloor d\log_2(\varepsilon_{\rm main}n)\rfloor\).
\end{corollary}

\begin{proof}
The calibrated root noise gives
\(\mathbb E m\le n+(2\varepsilon)^{-1}\le5n/4\).
The fixed-depth claim follows from Theorem~\ref{thm:time}.

For the adaptive mechanism, sparse routing to depth \(L\)
computes all required occupied-cell counts in \(O(nL)\) time.
Given the public candidate schedules, evaluating their scores
and sensitivities and selecting a candidate cost \(O(|\mathcal A|L)\).
For the standard candidate set, constructing the schedules also
takes \(O(|\mathcal A|L)\) time.
Since the selected depth is at most \(L\), the expected synthesis
cost is \(O(nL+dn)\).
Adding these costs proves the adaptive bound.
\end{proof}

%% file: exponential.tex
\section{Private Adaptive Depth Selection via the Exponential Mechanism}
\label{sec:depth-selection}

The fixed-depth analysis shows that, for a given occupancy profile, the leading transport
upper bound is minimized by a depth and noise schedule calibrated to that profile. If the noise scales were chosen as a deterministic function of \(O_0(X),\dots,O_{r-1}(X)\) without privacy accounting, then even before releasing synthetic data, the selected scales themselves could leak information about the data. The exponential mechanism is used only to privatize this adaptive choice; after selection, the synthesis stage is an ordinary  Pruned-PMM run. The
occupancy profile is private, so directly optimizing the schedule from the realized profile would
make the final noise distribution data-dependent. We avoid this by separating adaptation from
synthesis: first, we privately select one candidate from a finite public candidate set of
depth--schedule pairs using the exponential mechanism; then, conditional on the selected candidate,
we run Pruned-PMM once with a fixed schedule.

Throughout this section, \(\Omega=[0,1]^d\) is equipped with
\(\rho(x,y)=\|x-y\|_\infty\). Split the privacy budget as
\(\varepsilon=\varepsilon_{\mathrm{sel}}+\varepsilon_{\mathrm{main}}\), with
\(\varepsilon_{\mathrm{sel}},\varepsilon_{\mathrm{main}}>0\). Assume
\(\varepsilon_{\mathrm{main}}n\ge 2\), and write
$
N_{\mathrm{main}}:=\varepsilon_{\mathrm{main}}n.
$
The root perturbation is not required for replacement privacy, since \(n\) is fixed.
We include it only to match the previous fixed-depth statement, taking
\(\lambda_\varnothing\sim\mathrm{LapZ}(\varepsilon_{\mathrm{main}}^{-1})\)
independently of the adaptive selection. It does not enter the privacy constraint. The term
\(C_{\rm abs}/(\varepsilon_{\rm main}n)\) may be dropped if one sets \(m_\varnothing=n\).

\paragraph{Fixed-depth upper-bound notation.}
For a depth \(r\ge1\), write \(O^{(r)}=(O_0,\dots,O_{r-1})\). Recall the fixed-depth quantity
$
A_r(O^{(r)})
:=
2\sum_{j=0}^{r-1}\sqrt{O_j}\,2^{-j/(2d)}.
$
With privacy budget \(\varepsilon_{\mathrm{main}}\), the corresponding optimized fixed-depth upper
bound is
\[
B_r(O^{(r)})
:=
\frac{C_{\mathrm{abs}}}{\varepsilon_{\mathrm{main}}n}
+
\frac{2C_{\mathrm{abs}}}{\varepsilon_{\mathrm{main}}n}
A_r(O^{(r)})^2
+
2\cdot2^{-r/d}.
\]
This is exactly the fixed-depth bound from the previous section with \(\varepsilon\) replaced by
\(\varepsilon_{\mathrm{main}}\). The following notation will be used throughout, and the distinctions among the three quantities are summarized below

\[
\begin{array}{c|c|c}
\text{symbol} & \text{depends on data?} & \text{meaning}\\
\hline
B_r(O^{(r)}) & \text{yes} & \text{fixed-depth bound using actual occupancy profile}\\
B_{s,r} & \text{no} & \text{public model bound for model dimension}\,\, s \,\,\text{and depth}\,\, r\\
B_a(X) & \text{yes} & \text{exponential-mechanism score for public candidate}\,\, a
\end{array}
\]

\subsection{A standard candidate set of depth--schedule pairs}

We adapt over the integer grid \(s\in[d]\), where each \(s\) represents the model that the
occupied-cell growth behaves like \(2^{js/d}\). Here \(s\) should be read as a candidate packing-growth dimension. Define
$
L:=\left\lfloor d\log_2(\varepsilon_{\mathrm{main}}n)\right\rfloor.
$
For each \(s\in[d]\), define the model occupancy profile
$
P_j^{(s)}:=\min\{n,2^{js/d}\},\,\, j\ge0.
$
For \(1\le r\le L\), write
$
P^{(s,r)}:=(P_0^{(s)},\dots,P_{r-1}^{(s)})
$
and define
$
A_{s,r}:=A_r(P^{(s,r)})
=
2\sum_{j=0}^{r-1}\sqrt{P_j^{(s)}}\,2^{-j/(2d)}.
$
The schedule associated with the pair \((s,r)\) is
\[
\sigma_{j+1}^{(s,r)}
:=
\frac{A_{s,r}}
{\varepsilon_{\mathrm{main}}\sqrt{P_j^{(s)}}\,2^{-j/(2d)}},
\qquad j=0,\dots,r-1.
\]
Then
$
2\sum_{i=1}^{r}\frac{1}{\sigma_i^{(s,r)}}
=
\varepsilon_{\mathrm{main}}.
$
Thus Pruned-PMM run at depth \(r\) with schedule \(\sigma^{(s,r)}\) uses exactly
privacy budget \(\varepsilon_{\mathrm{main}}\) for the counts.

The model bound for \((s,r)\) is
\[
B_{s,r}:=B_r(P^{(s,r)})
=
\frac{C_{\mathrm{abs}}}{\varepsilon_{\mathrm{main}}n}
+
\frac{2C_{\mathrm{abs}}}{\varepsilon_{\mathrm{main}}n}
A_{s,r}^2
+
2\cdot2^{-r/d}.
\]
As in the main text, set \(r_1^\circ:=L\) and, for \(s\ge2\),
\[
r_s^\circ:=\max\left\{1,
\left\lfloor\frac{d}{s}\log_2\!\left(
1+\varepsilon_{\mathrm{main}}n\left(\frac{s-1}{d}\right)^2
\right)\right\rfloor\right\}.
\]
These depths lie in \(\{1,\dots,L\}\), since \(N_{\mathrm{main}}\ge2\) and
\(1+N_{\mathrm{main}}((s-1)/d)^2\le N_{\mathrm{main}}^s\) for \(s\ge2\).
The \textit{standard} candidate set is
\begin{equation}
\mathcal A^\circ
:=
\{(s,r_s^\circ,\sigma^{(s,r_s^\circ)}):s\in[d]\}.
\end{equation}
Here \(r_s^\circ\) is the depth denoted by \(r_s\) in the main text.

Each candidate contains both a depth and a full noise schedule. This is
essential: the schedule that is optimal for a fixed depth depends on the private
occupancy profile \(O_0,\dots,O_{r-1}\). If we selected only the depth and then
computed the optimal schedule from the realized profile, the final noise scales
would be data-dependent outside the accounted privacy mechanism. Instead, every
candidate schedule is fixed in advance from a public model profile
\(P^{(s,r)}\), and the data-dependent part is confined to the exponential-mechanism
score.

The results below hold for any finite public candidate set \(\mathcal A\) of
candidates \(a=(s,r,\sigma^{(s,r)})\) satisfying \(s\in[d]\), \(1\le r\le L\),
and \(2\sum_{i=1}^r1/\sigma_i^{(s,r)}=\varepsilon_{\mathrm{main}}\). In particular,
one may optionally enlarge the candidate set by adding the public optimizing
depth \(r_s^\star\) for each \(s\), defined as the smallest depth in
\(\{1,\dots,L\}\) minimizing \(B_{s,r}\).
Adding such candidates can only improve the best-candidate term, but may increase
both \(\log|\mathcal A|\) and the maximum score sensitivity
\(\Delta_{\mathcal A}\) defined below. For the standard candidate
set, \(|\mathcal A^\circ|=d\); for the enlarged candidate set containing both
\(r_s^\circ\) and \(r_s^\star\), one has \(|\mathcal A|\le2d\).

\subsection{Private score and exponential-mechanism selection}

For a candidate \(a=(s,r,\sigma^{(s,r)})\in\mathcal A\), define its data-dependent fixed-schedule
upper-bound score by
\[
B_a(X)
:=
\frac{C_{\mathrm{abs}}}{\varepsilon_{\mathrm{main}}n}
+
\frac{4C_{\mathrm{abs}}}{n}
\sum_{j=0}^{r-1}
\sigma_{j+1}^{(s,r)}O_j(X)2^{-j/d}
+
2\cdot2^{-r/d}.
\]
This is the fixed-depth Pruned-PMM bound evaluated at the private occupancy profile, but with the schedule associated with candidate \(a\). The score \(B_a(X)\) is not the realized Wasserstein error. It is the deterministic
upper bound from the fixed-schedule analysis, evaluated on the occupancy profile of \(X\).
It is computable from the private data and therefore must itself be selected privately.

For every \(j\ge1\), replacing one
point changes \(O_j(X)\) by at most one. Indeed, one occupied cell may disappear and one occupied
cell may appear, but the net change in the number of occupied cells has absolute value at most one.
Therefore the sensitivity of \(B_a\) is bounded by
$
\Delta_a
:=
\frac{4C_{\mathrm{abs}}}{n}
\sum_{j=1}^{r-1}
\sigma_{j+1}^{(s,r)}2^{-j/d}.
$
Let
$
\Delta_{\mathcal A}:=\max_{a\in\mathcal A}\Delta_a.
$
If \(\Delta_{\mathcal A}=0\), select a minimizer of \(B_a(X)\) using a fixed
public tie-breaking rule. Otherwise, the exponential mechanism selects
\(\widehat a\in\mathcal A\) with probability
\[
\Pr(\widehat a=a)
=
\frac{
\exp\!\left(
-\frac{\varepsilon_{\mathrm{sel}}B_a(X)}{2\Delta_{\mathcal A}}
\right)
}
{
\sum_{b\in\mathcal A}
\exp\!\left(
-\frac{\varepsilon_{\mathrm{sel}}B_b(X)}{2\Delta_{\mathcal A}}
\right)
}.
\]
If \(\widehat a=(\widehat s,\widehat r,\sigma^{(\widehat s,\widehat r)})\), run Pruned-PMM once at
depth \(\widehat r\), using root mass \(m_\varnothing=\max\{n+\lambda_\varnothing,1\}\) and the
selected schedule \(\sigma^{(\widehat s,\widehat r)}\). Use fresh synthesis randomness
independent of the selection stage. Denote the output by \(Y\).

\subsection{A Gibbs selection inequality}

\begin{lemma}[Gibbs selection inequality]
\label{lem:gibbs-selection}
Let \(\mathcal A\) be finite and let \(G_a\in\mathbb R\) be arbitrary scores. For \(\alpha>0\),
sample \(\widehat a\in\mathcal A\) with probability proportional to \(\exp(-\alpha G_a)\). Then
\[
\mathbb E G_{\widehat a}
\le
\min_{a\in\mathcal A}G_a
+
\frac{\log|\mathcal A|}{\alpha}.
\]
\end{lemma}

\begin{proof}
Let \(G_\star:=\min_{a\in\mathcal A}G_a\), \(Z:=\sum_{a\in\mathcal A}\exp(-\alpha G_a)\), and
\(p_a:=\exp(-\alpha G_a)/Z\). Since \(\log p_a=-\alpha G_a-\log Z\), the entropy
\(H(p):=-\sum_a p_a\log p_a\) satisfies
$
H(p)=\alpha\sum_a p_aG_a+\log Z.
$
Therefore
$
\mathbb E G_{\widehat a}
=
\frac{H(p)-\log Z}{\alpha}.
$
Since \(H(p)\le\log|\mathcal A|\) and \(Z\ge\exp(-\alpha G_\star)\), we get
$
\mathbb E G_{\widehat a}
\le
G_\star+\frac{\log|\mathcal A|}{\alpha}.
$
\end{proof}

\subsection{Privacy and finite-candidate guarantee}
\begin{theorem}[Privacy of Adaptive Pruned-PMM]
\label{thm:adaptive-privacy}
The Adaptive Pruned-PMM mechanism described above is
\((\varepsilon_{\mathrm{sel}}+\varepsilon_{\mathrm{main}})\)-differentially private under replacement
adjacency.
\end{theorem}

\begin{proof}
The score sensitivity is bounded by \(\Delta_{\mathcal A}\), so the exponential mechanism is
\(\varepsilon_{\mathrm{sel}}\)-differentially private when \(\Delta_{\mathcal A}>0\).
If \(\Delta_{\mathcal A}=0\), adjacent datasets have the same score vector, so
the fixed tie-breaking rule has zero privacy cost. Conditional on the selected candidate
\(\widehat a=(\widehat s,\widehat r,\sigma^{(\widehat s,\widehat r)})\), the depth and all noise
scales are fixed and satisfy
$
2\sum_{i=1}^{\widehat r}
\frac{1}{\sigma_i^{(\widehat s,\widehat r)}}
=
\varepsilon_{\mathrm{main}}.
$
Therefore, by the fixed-depth privacy theorem for Pruned-PMM, the final synthesis stage is
\(\varepsilon_{\mathrm{main}}\)-differentially private. Sequential composition gives total privacy
\(\varepsilon_{\mathrm{sel}}+\varepsilon_{\mathrm{main}}\).
\end{proof}

\begin{theorem}[Best-candidate guarantee]
\label{thm:adaptive-library}
For every fixed dataset \(X\),
\[
\mathbb E W_1(\mu_X,\mu_Y)
\le
\min_{a\in\mathcal A}B_a(X)
+
\frac{2\Delta_{\mathcal A}}{\varepsilon_{\mathrm{sel}}}\log|\mathcal A|.
\]
The expectation is over both the exponential-mechanism randomness and the final Pruned-PMM
randomness.
\end{theorem}

\begin{proof}
For accuracy, condition on \(\widehat a=a\). Since the selected schedule is fixed after selection
and synthesis uses fresh independent randomness,
the fixed-depth Pruned-PMM upper bound gives
\[
\mathbb E_{\mathrm{main}}
\!\left[
W_1(\mu_X,\mu_Y)\mid \widehat a=a
\right]
\le
B_a(X).
\]
Taking expectation over the selection stage gives
\[
\mathbb E W_1(\mu_X,\mu_Y)
\le
\mathbb E_{\mathrm{sel}}B_{\widehat a}(X).
\]
If \(\Delta_{\mathcal A}=0\), the score vector is unchanged under adjacent modifications, and the
deterministic selection rule has zero privacy cost and zero selection penalty. Otherwise,
\(\Delta_{\mathcal A}>0\), and applying Lemma~\ref{lem:gibbs-selection} with \(G_a=B_a(X)\) and
$
\alpha=\frac{\varepsilon_{\mathrm{sel}}}{2\Delta_{\mathcal A}}
$
gives the claim.
\end{proof}

\subsection{Explicit adaptive rates under packing-growth}

Assume the empirical support satisfies Assumption~\ref{ass:pack} with
packing-growth dimension \(k\in[d]\).
For each candidate \(s\in[d]\), define
$
q_s:=2^{(s-1)/(2d)}
$
and
\[
S_{s,r}:=\sum_{j=0}^{r-1}q_s^j
=
\begin{cases}
r, & s=1,\\[2mm]
\dfrac{q_s^r-1}{q_s-1}, & s>1.
\end{cases}
\]
This is the same \(S_r\) from the fixed-depth upper bound, written as \(S_{s,r}\) only because the
candidate growth exponent is \(s\).

\begin{corollary}[Explicit adaptive rate under packing-growth]
\label{cor:adaptive-explicit-rate}
Suppose the candidate set contains the candidate \(s=k\), where \(k\in[d]\), with either
the optimizer \(r_k^\star\) or the explicit optimized depth
$
r_1^\circ:=L
$
when \(k=1\), and
$
r_k^\circ
:=
\max\left\{
1,
\left\lfloor
\frac{d}{k}
\log_2\!\left(
1+\varepsilon_{\mathrm{main}}n
\left(\frac{k-1}{d}\right)^2
\right)
\right\rfloor
\right\}
$
when \(2\le k\le d\). Assume Assumption~\ref{ass:pack} holds with
packing-growth dimension \(k\). Then there exists a constant \(C>0\), depending only on
\(C_{\mathrm{abs}}\) and \(C_{\mathrm{pack}}\), such that
the following bounds hold.

If \(k=1\), then
\[
\mathbb E W_1(\mu_X,\mu_Y)
\le
C\frac{d^2\log_2^2(\varepsilon_{\mathrm{main}}n)}
{\varepsilon_{\mathrm{main}}n}
+
\frac{2\Delta_{\mathcal A}}{\varepsilon_{\mathrm{sel}}}
\log|\mathcal A|.
\]

If \(2\le k\le d\), then
\[
\mathbb E W_1(\mu_X,\mu_Y)
\le
C\left(\frac{d}{k-1}\right)^{2/k}
(\varepsilon_{\mathrm{main}}n)^{-1/k}
+
\frac{2\Delta_{\mathcal A}}{\varepsilon_{\mathrm{sel}}}
\log|\mathcal A|.
\]
\end{corollary}

\begin{proof}
By Assumption~\ref{ass:pack} and Lemma~\ref{lem:pack2occ}, and since
\(O_j(X)\le n\), the packing-growth condition implies
$
O_j(X)
\le
C_{\mathrm{pack}}\min\{n,2^{jk/d}\}
=
C_{\mathrm{pack}}P_j^{(k)}.
$
Fix a depth \(r\), and let \(a=(k,r,\sigma^{(k,r)})\). Using the definition of
\(\sigma^{(k,r)}\), we have
\[
\frac{4C_{\mathrm{abs}}}{n}
\sum_{j=0}^{r-1}
\sigma_{j+1}^{(k,r)}O_j(X)2^{-j/d}
\le
\frac{4C_{\mathrm{abs}}C_{\mathrm{pack}}}{n}
\sum_{j=0}^{r-1}
\sigma_{j+1}^{(k,r)}P_j^{(k)}2^{-j/d}
=
\frac{2C_{\mathrm{abs}}C_{\mathrm{pack}}}{\varepsilon_{\mathrm{main}}n}
A_{k,r}^2.
\]
Since \(C_{\mathrm{pack}}\ge1\), the root and resolution terms are also bounded by
\(C_{\mathrm{pack}}\) times the corresponding terms in \(B_{k,r}\). Hence
$
B_a(X)\le C_{\mathrm{pack}}B_{k,r}.
$
If the candidate set contains \(r_k^\circ\), evaluate this bound at \(r_k^\circ\). If the candidate set contains
\(r_k^\star\), then \(B_{k,r_k^\star}\le B_{k,r_k^\circ}\), so the same upper bound follows.

We now bound \(B_{k,r_k^\circ}\). If \(k=1\), then \(q_1=1\), so \(S_{1,r}=r\), and
$
A_{1,r_1^\circ}
\le
2S_{1,r_1^\circ}
=
2r_1^\circ
\le
2d\log_2(\varepsilon_{\mathrm{main}}n).
$
Also,
$
2^{-r_1^\circ/d}
\le
2^{1/d}(\varepsilon_{\mathrm{main}}n)^{-1}.
$
Since \(\varepsilon_{\mathrm{main}}n\ge2\), we may upper bound the root and resolution terms by the
same logarithmic scale. Hence
\[
B_{1,r_1^\circ}
=
O\!\left(
\frac{d^2\log_2^2(\varepsilon_{\mathrm{main}}n)}
{\varepsilon_{\mathrm{main}}n}
\right).
\]

Now suppose \(2\le k\le d\). Write \(N:=\varepsilon_{\mathrm{main}}n\),
\(a:=(k-1)/d\), \(M:=Na^2\), and \(q:=2^{(k-1)/(2d)}\).
Then \(q-1\asymp a\), with universal comparison constants, and
\[
\frac{d}{k}\log_2(1+M)-1
\le r_k^\circ
\le 1+\frac{d}{k}\log_2(1+M).
\]
If \(M\le1\), then \(q^{r_k^\circ}\le2\) and
\(r_k^\circ\le1+dM/(k\ln2)\). Since
\(A_{k,r}\le2\sum_{j<r}q^j\le2rq^r\),
\[
\frac{A_{k,r_k^\circ}^2}{N}
\le C\frac{(r_k^\circ)^2}{N}
\le C\left(\frac1N+\frac{d^2M^2}{k^2N}\right)
=C\left(\frac1N+\frac{(k-1)^2}{k^2}M\right)
\le C.
\]
The root and resolution terms are also bounded, so
\(B_{k,r_k^\circ}\le C\le CM^{-1/k}\).

If \(M>1\), then
\[
A_{k,r_k^\circ}
\le\frac{2q^{r_k^\circ}}{q-1}
\le Ca^{-1}M^{(k-1)/(2k)},
\qquad
2^{-r_k^\circ/d}
\le2^{1/d}(1+M)^{-1/k}\le CM^{-1/k}.
\]
Thus \(A_{k,r_k^\circ}^2/N\le CM^{-1/k}\); also
\(N^{-1}\le M^{-1/k}\), since \(a\le1\) and \(N\ge2\).
In both cases we get
\[
B_{k,r_k^\circ}
=
O\!\left(
\left(\frac{d}{k-1}\right)^{2/k}
(\varepsilon_{\mathrm{main}}n)^{-1/k}
\right).
\]
Combining these estimates with Theorem~\ref{thm:adaptive-library} proves the result.
\end{proof}

\subsection{Making the selection penalty explicit}

The finite-candidate guarantee depends on the score sensitivity
\(\Delta_{\mathcal A}\). We now give a uniform bound for the standard candidate
set \(\mathcal A^\circ\), and also for the enlarged candidate set obtained by
adding the optional public optimizing depths \(r_s^\star\). Assume \(d\ge2\) and
\(\varepsilon_{\mathrm{main}}n\ge2\). Constants in this subsection may depend on
\(C_{\mathrm{abs}}\), but not on \(n\), \(\varepsilon_{\mathrm{main}}\),
\(\varepsilon_{\mathrm{sel}}\), \(X\), or \(|\mathcal A|\).

For the explicit depths, we use
\[
r_1^\circ:=L,
\qquad
r_s^\circ
:=
\max\left\{
1,
\left\lfloor
\frac{d}{s}
\log_2\!\left(
1+\varepsilon_{\mathrm{main}}n
\left(\frac{s-1}{d}\right)^2
\right)
\right\rfloor
\right\},
\qquad s\ge2.
\]

\begin{lemma}[Uniform selection-sensitivity bound]
\label{lem:selection-sensitivity-simple}
For the standard candidate set \(\mathcal A^\circ\), and also for the enlarged
candidate set obtained by adding the optional public optimizing depths \(r_s^\star\), one has
\[
\Delta_{\mathcal A}
\le
C_{\mathrm{sel}}d^2
(\varepsilon_{\mathrm{main}}n)^{-(d+1)/(2d)},
\]
where \(C_{\mathrm{sel}}\) depends only on the noise convention.
\end{lemma}

\begin{proof}
The reason this sensitivity is small is that the contribution of deeper levels is
discounted by \(2^{-j/d}\). We combine this with the candidate-specific depth choices to
bound the schedule normalizations uniformly.

For any candidate \(a=(s,r,\sigma^{(s,r)})\),
$
\Delta_a
=
\frac{4C_{\mathrm{abs}}}{n}
\sum_{j=1}^{r-1}
\sigma_{j+1}^{(s,r)}2^{-j/d}.
$
Using the definition of \(\sigma^{(s,r)}\) and \(P_j^{(s)}\ge1\),
\[
\sigma_{j+1}^{(s,r)}2^{-j/d}
=
\frac{A_{s,r}}{\varepsilon_{\mathrm{main}}}
\frac{2^{-j/(2d)}}{\sqrt{P_j^{(s)}}}
\le
\frac{A_{s,r}}{\varepsilon_{\mathrm{main}}}2^{-j/(2d)}.
\]
Since
$
\sum_{j=1}^{\infty}2^{-j/(2d)}
=
\frac{1}{2^{1/(2d)}-1}
\le
\frac{2d}{\ln2},
$
we get
$
\Delta_a
\le
C d A_{s,r}(\varepsilon_{\mathrm{main}}n)^{-1}.
$
It remains to bound \(A_{s,r}\) uniformly over the candidates.

First, consider the explicit depths. If \(s=1\), then
\(A_{1,r_1^\circ}\le2r_1^\circ\le C d\log_2(\varepsilon_{\mathrm{main}}n)\). Since \(d\ge2\),
$
\log_2 x\le Cx^{(d-1)/(2d)},\,\, x\ge1,
$
and hence
$
A_{1,r_1^\circ}
\le
C d (\varepsilon_{\mathrm{main}}n)^{(d-1)/(2d)}.
$

Now let \(s\ge2\), and set \(q_s:=2^{(s-1)/(2d)}\) and
\(a_s:=(s-1)/d\). Since \(q_s-1\ge c a_s\) and
$
r_s^\circ
\le
1+\frac{d}{s}\log_2(1+(\varepsilon_{\mathrm{main}}n)a_s^2),
$
we have
\[
A_{s,r_s^\circ}
\le
2\sum_{j=0}^{r_s^\circ-1}q_s^j
\le
\frac{2q_s^{r_s^\circ}}{q_s-1}
\le
C d(1+(\varepsilon_{\mathrm{main}}n)a_s^2)^{(s-1)/(2s)}.
\]
Because \(s\le d\), the exponent satisfies
$
\frac{s-1}{2s}\le\frac{d-1}{2d}.
$
Also \(1+(\varepsilon_{\mathrm{main}}n)a_s^2\le 1+(\varepsilon_{\mathrm{main}}n)\le 2(\varepsilon_{\mathrm{main}}n)\), since
\((\varepsilon_{\mathrm{main}}n)\ge2\). Therefore
\[
A_{s,r_s^\circ}
\le
C d(\varepsilon_{\mathrm{main}}n)^{(d-1)/(2d)}.
\]

For an optional optimizer \(r_s^\star\), use optimality:
\(B_{s,r_s^\star}\le B_{s,r_s^\circ}\).
The model-bound calculation in the proof of
Corollary~\ref{cor:adaptive-explicit-rate} gives, for \(s\ge2\),
\[
B_{s,r_s^\circ}
\le C\left(\frac{d}{s-1}\right)^{2/s}
(\varepsilon_{\mathrm{main}}n)^{-1/s}
\le C d^2(\varepsilon_{\mathrm{main}}n)^{-1/d}.
\]
For \(s=1\), the same conclusion follows from
\(B_{1,r_1^\circ}\le C d^2\log_2^2(N_{\mathrm{main}})/N_{\mathrm{main}}\)
and \(\log_2^2 x\le Cx^{1-1/d}\) for \(x\ge2\), \(d\ge2\).
By optimality,
\(B_{s,r_s^\star}\le B_{s,r_s^\circ}\), while
\(B_{s,r_s^\star}\ge (2C_{\mathrm{abs}}/(\varepsilon_{\mathrm{main}}n))A_{s,r_s^\star}^2\).
Therefore \(A_{s,r_s^\star}\le C d(\varepsilon_{\mathrm{main}}n)^{(d-1)/(2d)}\). Combining
this with the same bound for the explicit-depth candidates, every candidate in the standard or enlarged candidate set satisfies
\(A_{s,r}\le C d(\varepsilon_{\mathrm{main}}n)^{(d-1)/(2d)}\). Since
\(\Delta_a\le CdA_{s,r}(\varepsilon_{\mathrm{main}}n)^{-1}\), we get
\(\Delta_a\le C_{\mathrm{sel}}d^2(\varepsilon_{\mathrm{main}}n)^{-(d+1)/(2d)}\) for every
candidate \(a\), and hence
\(\Delta_{\mathcal A}\le C_{\mathrm{sel}}d^2(\varepsilon_{\mathrm{main}}n)^{-(d+1)/(2d)}\).
\end{proof}

\begin{corollary}[Adaptive Pruned-PMM rate for the standard candidate set]
\label{cor:adaptive-final-rate}
Assume \(d\ge2\), \(\varepsilon_{\mathrm{main}}n\ge2\), and run Adaptive
Pruned-PMM with the standard candidate set \(\mathcal A^\circ\) indexed by \(s\in[d]\).
Suppose the empirical support \(S_X\) satisfies Assumption~\ref{ass:pack} with
packing-growth dimension \(k\in[d]\). Then, Adaptive Pruned-PMM is
\((\varepsilon_{\mathrm{sel}}+\varepsilon_{\mathrm{main}})\)-differentially
private and satisfies the following bounds.

If \(k=1\), then
\[
\mathbb E W_1(\mu_X,\mu_Y)
\le
O\!\left(
\frac{d^2\log_2^2(\varepsilon_{\mathrm{main}}n)}
{\varepsilon_{\mathrm{main}}n}
\right)
+
C_{\mathrm{sel}}
\frac{d^2\log d}{\varepsilon_{\mathrm{sel}}}
(\varepsilon_{\mathrm{main}}n)^{-(d+1)/(2d)}.
\]

If \(2\le k\le d\), then
\[
\mathbb E W_1(\mu_X,\mu_Y)
\le
O\!\left(
\left(\frac{d}{k-1}\right)^{2/k}
(\varepsilon_{\mathrm{main}}n)^{-1/k}
\right)
+
C_{\mathrm{sel}}
\frac{d^2\log d}{\varepsilon_{\mathrm{sel}}}
(\varepsilon_{\mathrm{main}}n)^{-(d+1)/(2d)}.
\]
Here the constants in the \(O(\cdot)\) terms depend only on
\(C_{\mathrm{abs}}\) and \(C_{\mathrm{pack}}\), while \(C_{\mathrm{sel}}\)
depends only on the noise convention.
\end{corollary}

\begin{proof}
Privacy follows from Theorem~\ref{thm:adaptive-privacy}. The finite-candidate set
guarantee, Theorem~\ref{thm:adaptive-library}, gives
\[
\mathbb E W_1(\mu_X,\mu_Y)
\le
\min_{a\in\mathcal A^\circ}B_a(X)
+
\frac{2\Delta_{\mathcal A^\circ}}{\varepsilon_{\mathrm{sel}}}
\log|\mathcal A^\circ|.
\]
Since \(|\mathcal A^\circ|=d\), Lemma~\ref{lem:selection-sensitivity-simple}
bounds the selection term by
\[
C_{\mathrm{sel}}
\frac{d^2\log d}{\varepsilon_{\mathrm{sel}}}
(\varepsilon_{\mathrm{main}}n)^{-(d+1)/(2d)}.
\]
The best-candidate term is bounded by Corollary~\ref{cor:adaptive-explicit-rate},
because the standard candidate set contains the candidate \(s=k\). Combining these
bounds proves the result.
\end{proof}

%% file: lower.tex
\section{Main lower bound}\label{subsec:lb-main}

We now prove a minimax lower bound for private synthetic data under an intrinsic packing-growth
assumption on the support set \(S\). For \(k>1\), the result establishes the optimal exponent
\(1/k\) under matching upper and lower growth assumptions; Proposition~\ref{prop:matching-coordinate-supports}
gives an explicit class on which the rates match.

\begin{theorem}[DP synthetic-data lower bound]\label{thm:lb}
Let \(T=[0,1]^d\) with \(\rho(x,y)=\|x-y\|_\infty\), and let \(S\subset T\) satisfy the packing-growth
assumption: there exist constants \(c_0>0\), \(t_0\in(0,1]\), and an integer \(k\in\{1,\dots,d\}\) such that
\begin{equation}\label{eq:pack-growth}
N_{\mathrm{pack}}(S,\|\cdot\|_\infty,t)\ \ge\ c_0\,t^{-k}
\qquad\text{for all }0<t\le t_0.
\end{equation}
Fix \(\varepsilon\in(0,1]\). Then there exist constants \(c>0\) and \(n_0\in\mathbb N\),
depending only on \(c_0,t_0,k\), such that for all \(n\ge n_0\) and $\varepsilon n \ge 2$, every
\(\varepsilon\)-differentially private mechanism
\[
\mathcal M:S^n\to \bigsqcup_{m\ge 1}T^m
\]
satisfies
\[
\sup_{X\in S^n}\ \mathbb E\,W_1\!\bigl(\mu_X,\mu_{\mathcal M(X)}\bigr)
\ \ge\
c\,\min\bigl\{t_0,(\varepsilon n)^{-1/k}\bigr\}.
\]
\end{theorem}

The proof is based on a packing argument in Wasserstein space. Starting from the lower bound
\eqref{eq:pack-growth}, we first extract many well-separated points in \(S\) at an appropriate
scale \(t\). We then use a large constant-weight code to construct a family of datasets
\(\{X_U\}_{U\in\mathcal U}\subset S^n\) whose empirical measures \(\mu_{X_U}\) form a large
\(W_1\)-packing. Finally, after post-processing the output of the synthetic-data mechanism into an
empirical measure, group privacy turns ordinary \(\varepsilon\)-differential privacy on datasets into
metric privacy on this finite family, and Proposition~\ref{prop:master} then forces the expected
error to be at least a constant multiple of the packing scale.

For clarity, the proof is organized into four steps. In Section~\ref{subsec:lb-tools} we collect the
three ingredients needed later: the abstract master lower bound
(Proposition~\ref{prop:master}), a simple lemma relating \(W_1\) and total variation on a
separated support, and a constant-weight coding lemma. In
Section~\ref{subsec:lb-construction} we construct the hard family of datasets.
In Section~\ref{subsec:lb-separation} we prove that the associated empirical measures are
pairwise separated in \(W_1\). In Section~\ref{subsec:lb-proof} we combine this packing with the
metric-privacy reduction to conclude the theorem.

To streamline the presentation, the main proof below is written in the regime
\(\varepsilon n\ge 2\).

\subsection{Tools for the lower bound}\label{subsec:lb-tools}

We collect here the four ingredients used in the proof of the lower bound.
The first is an abstract packing lower bound for metrically private algorithms, which is the master lower-bound principle of
\citep{BoedihardjoStrohmerVershynin2024PrivateMeasures}.
The second converts separation of support into a Wasserstein lower bound through
total variation. The third is the standard group-privacy reduction, which turns
ordinary \(\varepsilon\)-differential privacy on datasets into metric privacy on a
finite family of hard instances. The fourth is a constant-weight coding lemma
that provides exponentially many well-separated subsets of \([M]\).

\paragraph{Convention on total variation.}
Throughout this section we use
\[
\mathrm{TV}(\mu,\nu):=\frac12\|\mu-\nu\|_1
\]
for probability measures on a finite set.

\begin{proposition}[Master lower bound: metric privacy{\citep{BoedihardjoStrohmerVershynin2024PrivateMeasures}}]\label{prop:master}
Let \((M_1,\rho_1)\) be a metric space and let \(M_0\subset M_1\).
Let \(\rho_0\) be a metric on \(M_0\) with \(\mathrm{diam}(M_0,\rho_0)\le 1\).
Assume that for some \(t>0\),
\[
N_{\mathrm{pack}}(M_0,\rho_1,t)\ >\ 2e^\alpha.
\]
Then for any randomized algorithm \(A:M_0\to M_1\) that is
\(\alpha\)-metrically private with respect to \(\rho_0\), there exists
\(x\in M_0\) such that
\[
\mathbb{E}\,\rho_1(A(x),x)\ >\ \frac{t}{4}.
\]
\end{proposition}

\begin{lemma}[\(t\)-separated support implies \(W_1\ge t\,\mathrm{TV}\)]
\label{lem:sep-TV}
Let \(\{z_1,\dots,z_M\}\subset T\) be \(t\)-separated in \(\|\cdot\|_\infty\), meaning
$
\|z_i-z_j\|_\infty\ge t
\,\text{for all }i\neq j.
$
If \(\mu,\nu\) are probability measures supported on \(\{z_1,\dots,z_M\}\), then
\[
W_1(\mu,\nu)\ \ge\ t\,\mathrm{TV}(\mu,\nu).
\]
\end{lemma}

\begin{proof}
Let \(\pi\) be any coupling of \((\mu,\nu)\). Any mass matched on the diagonal
\((z_i,z_i)\) has zero transport cost. Any mass transported from \(z_i\) to
\(z_j\) with \(i\neq j\) travels distance at least \(t\), by the
\(t\)-separation assumption. Therefore
\[
\int \|x-y\|_\infty\,d\pi(x,y)
\ \ge\
t\cdot \pi\bigl(\{(x,y):x\neq y\}\bigr).
\]
Now the smallest possible off-diagonal mass over all couplings equals
\(\mathrm{TV}(\mu,\nu)\). Hence
\[
\int \|x-y\|_\infty\,d\pi(x,y)\ \ge\ t\,\mathrm{TV}(\mu,\nu).
\]
Taking the infimum over all couplings \(\pi\) proves the claim.
\end{proof}

\paragraph{A constant-weight code.}
For \(u\in(0,1)\), let
$
H(u):=-u\ln u-(1-u)\ln(1-u)
$
denote the binary entropy function, with natural logarithms, extended continuously by \(H(0)=H(1)=0\).

\begin{lemma}[Large constant-weight family]\label{lem:cw}
Fix
$
\delta:=\frac1{25}.
$
Then there exist constants \(\beta>1/2\) and \(M_\star\in\mathbb N\) such that
for every even integer \(M\ge M_\star\), there exists a family
\[
\mathcal U\subset \binom{[M]}{M/2}
\]
satisfying
\[
|\mathcal U|\ \ge\ e^{\beta M},
\qquad
|U\triangle V|\ \ge\ \delta M
\quad\text{for all }U\neq V\text{ in }\mathcal U.
\]
\end{lemma}

\begin{proof}
Consider the graph whose vertex set is \(\binom{[M]}{M/2}\), and where two
vertices \(U,V\) are adjacent if
\[
|U\triangle V|<\delta M.
\]
A greedy maximal independent-set argument yields a family \(\mathcal U\) of size at least
\[
|\mathcal U|
\ \ge\
\frac{\binom{M}{M/2}}{B(M,\delta)},
\]
where
\[
B(M,\delta)
:=
\max_{U\in\binom{[M]}{M/2}}
\#\Bigl\{
V\in\binom{[M]}{M/2}: |U\triangle V|<\delta M
\Bigr\}.
\]

If \(|U|=|V|=M/2\), then \(|U\triangle V|=2s\), where
$
s:=|U\setminus V|=|V\setminus U|.
$
Hence the condition \(|U\triangle V|<\delta M\) is equivalent to \(s<\delta M/2\), and therefore
\[
B(M,\delta)\le \sum_{s<\delta M/2}\binom{M/2}{s}^2.
\]
Writing \(s=\alpha M/2\), the standard entropy bound
\[
\binom{m}{\alpha m}\le e^{mH(\alpha)}
\]
gives
\[
\binom{M/2}{s}^2 \le e^{M H(\alpha)}
\qquad\text{for }\alpha\in[0,\delta).
\]
Since \(H(\alpha)\) is increasing on \([0,1/2]\), it follows that
\[
B(M,\delta)
\le
\left(\frac{\delta M}{2}+1\right)e^{M H(\delta)}.
\]
On the other hand,
\[
\binom{M}{M/2}\ge e^{M\ln 2-o(M)}.
\]
Therefore
$
\ln |\mathcal U|
\ge
M\bigl(\ln 2 - H(\delta)\bigr)-o(M).
$
For \(\delta=1/25\), one checks that
$
\ln 2 - H(1/25) > \frac12.
$
Choose any
$
\beta\in\Bigl(\frac12,\ \ln 2 - H(1/25)\Bigr).
$
Then for all sufficiently large even \(M\), namely for all \(M\ge M_\star\),
the \(o(M)\)-term is dominated and we obtain
$
\ln |\mathcal U|\ge \beta M,
$
equivalently,
$
|\mathcal U|\ge e^{\beta M}.
$
By construction, distinct members of \(\mathcal U\) satisfy
$
|U\triangle V|\ge \delta M.
$
This completes the proof.
\end{proof}

\subsection{Construction of the hard family}\label{subsec:lb-construction}

We now construct a finite family of datasets in \(S^n\) whose empirical measures will later be shown
to form a large packing in the \(W_1\)-metric. The construction has three ingredients:

\begin{enumerate}
\item a combinatorial parameter \(M\), chosen so that the constant-weight family from
Lemma~\ref{lem:cw} has cardinality larger than \(2e^{\varepsilon n}\);
\item a set of \(M+1\) well-separated points in \(S\), extracted from the packing-growth assumption;
\item a dataset \(X_U\in S^n\) associated to each codeword \(U\in\mathcal U\).
\end{enumerate}

\paragraph{Choice of the combinatorial size \(M\).}
Fix the constants \(\beta>1/2\) and \(M_\star\in\mathbb N\) from Lemma~\ref{lem:cw}. Define
\begin{equation}\label{eq:M-def}
M
:=
2\left\lceil
\frac12
\max\left\{
M_\star,\,
\frac{\varepsilon n+\ln 2+1}{\beta}
\right\}
\right\rceil .
\end{equation}
Then \(M\) is even, and the following two properties hold:
\begin{equation}\label{eq:M-basic}
M\ge M_\star,
\qquad
\beta M\ge \varepsilon n+\ln 2+1.
\end{equation}
Moreover, by the elementary bound \(2\lceil a/2\rceil\le a+2\),
\begin{equation}\label{eq:M-upper}
M
\le
\max\left\{
M_\star,\,
\frac{\varepsilon n+\ln 2+1}{\beta}
\right\}+2.
\end{equation}

The next observation is the only place where we need \(n\) to be sufficiently large.

\begin{lemma}[A convenient upper bound on \(M\)]\label{lem:Mle2n}
There exists \(n_0\in\mathbb N\), depending only on \(\beta\) and \(M_\star\), such that for all
\(n\ge n_0\),
$\,
M\le 2n.
$
\end{lemma}

\begin{proof}
By \eqref{eq:M-upper}, it is enough to control the two terms inside the maximum.

First, if
$
n\ge \frac{M_\star+2}{2},
$
then \(M_\star+2\le 2n\).

Second, since \(\varepsilon\le 1\), we have
$
\frac{\varepsilon n+\ln 2+1}{\beta}+2
\le
\frac{n+\ln 2+1}{\beta}+2.
$
Thus it is enough that
$
\frac{n+\ln 2+1}{\beta}+2\le 2n,
$
which is equivalent to
\[
(2\beta-1)n\ge \ln 2+1+2\beta.
\]
Because \(\beta>1/2\), this holds for all sufficiently large \(n\).

Therefore the conclusion follows for
\[
n_0:=
\left\lceil
\max\left\{
\frac{M_\star+2}{2},\,
\frac{\ln 2+1+2\beta}{2\beta-1}
\right\}
\right\rceil .
\]
\end{proof}

Henceforth we assume \(n\ge n_0\), so that \(M\le 2n\). In particular,
$
\left\lfloor \frac{2n}{M}\right\rfloor \ge 1.
$

\paragraph{The code family.}
Since \(M\) is even and \(M\ge M_\star\), Lemma~\ref{lem:cw} yields a family
\[
\mathcal U\subset \binom{[M]}{M/2}
\]
such that
\begin{equation}\label{eq:U-family}
|\mathcal U|\ge e^{\beta M},
\qquad
|U\triangle V|\ge \delta M
\quad\text{for all }U\neq V\text{ in }\mathcal U,
\end{equation}
where \(\delta=1/25\).
By \eqref{eq:M-basic},
$
|\mathcal U|
\ge
e^{\beta M}
\ge
e^{\varepsilon n+\ln 2+1}
>
2e^{\varepsilon n}.
$
Thus the code family is already exponentially larger than the privacy threshold appearing in
Proposition~\ref{prop:master}.

\paragraph{Separated points in \(S\).}
Next define the geometric scale
\begin{equation}\label{eq:t-def}
t
:=
\min\left\{
t_0,\,
\left(\frac{c_0}{2(M+1)}\right)^{1/k}
\right\}.
\end{equation}
Since \(t\le t_0\), the packing-growth assumption \eqref{eq:pack-growth} implies
\[
N_{\mathrm{pack}}(S,\|\cdot\|_\infty,t)
\ge
c_0 t^{-k}
\ge
2(M+1).
\]
Therefore there exist points
$
z_1,\dots,z_M,z_\star\in S
$
which are pairwise \(t\)-separated:
$
\|z_i-z_j\|_\infty\ge t
\,\,\text{for all }i\neq j.
$

\paragraph{Datasets indexed by codewords.}
For each \(U\in\mathcal U\), write
\[
U=\{i_1<\cdots<i_{M/2}\}.
\]
We define the ordered dataset \(X_U\in S^n\) by
\begin{equation}\label{eq:XU-def}
X_U
:=
\big(
\underbrace{z_{i_1},\dots,z_{i_1}}_{q},
\underbrace{z_{i_2},\dots,z_{i_2}}_{q},
\dots,
\underbrace{z_{i_{M/2}},\dots,z_{i_{M/2}}}_{q},
\underbrace{z_\star,\dots,z_\star}_{\,n-(M/2)q\,}
\big),
\end{equation}
where
$
q:=\left\lfloor \frac{2n}{M}\right\rfloor .
$
The ordering in \eqref{eq:XU-def} is fixed once and for all; this will be important later when we
define Hamming distance between two datasets \(X_U\) and \(X_V\).

Define also
\begin{equation}\label{eq:p-def}
p:=\frac{(M/2)q}{n}.
\end{equation}
Thus \(p\) is the fraction of the empirical mass of \(X_U\) that sits on the selected points
\(\{z_i:i\in U\}\), while the remaining mass \(1-p\) is concentrated at \(z_\star\).

Finally, define the probability measure
$
\mu_U
:=
\frac{2}{M}\sum_{i\in U}\delta_{z_i},
\,\, U\in\mathcal U.
$
This is simply the uniform probability measure on the support points indexed by \(U\).

The next lemma records the basic properties of this construction.

\begin{lemma}[Structure of the hard family]\label{lem:hard-family}
For every \(U\in\mathcal U\), the dataset \(X_U\) and the measure \(\mu_U\) satisfy:

\begin{enumerate}
\item \label{eq:p-half}
$
p\ge \frac12;
$

\item \label{eq:muXU-decomp} the empirical measure of \(X_U\) decomposes as
$
\mu_{X_U}
=
p\,\mu_U+(1-p)\delta_{z_\star};
$

\item the map \(U\mapsto \mu_{X_U}\) is injective.
\end{enumerate}
\end{lemma}

\begin{proof}
We first prove \eqref{eq:p-half}.

If \(M>n\), then
$
1\le \frac{2n}{M}<2,
$
so \(q=1\), and therefore
$
p=\frac{M}{2n}\ge \frac12.
$

If \(M\le n\), then
$
q=\left\lfloor \frac{2n}{M}\right\rfloor
\ge \frac{2n}{M}-1
\ge \frac{n}{M},
$
and hence
$
p=\frac{(M/2)q}{n}
\ge
\frac{(M/2)(n/M)}{n}
=
\frac12.
$
Thus \eqref{eq:p-half} holds in all cases.

Next we prove \eqref{eq:muXU-decomp}. By construction, \(X_U\) contains exactly \(q\) copies
of each \(z_i\) with \(i\in U\), and \(n-(M/2)q\) copies of \(z_\star\). Therefore the empirical
measure of \(X_U\) is
\[
\mu_{X_U}
=
\frac{q}{n}\sum_{i\in U}\delta_{z_i}
+
\frac{n-(M/2)q}{n}\delta_{z_\star}.
\]
Using the definitions of \(p\) and \(\mu_U\), we rewrite the first term as
\[
\frac{q}{n}\sum_{i\in U}\delta_{z_i}
=
\frac{(M/2)q}{n}\cdot \frac{2}{M}\sum_{i\in U}\delta_{z_i}
=
p\,\mu_U,
\]
while the second term is \((1-p)\delta_{z_\star}\). This proves
\eqref{eq:muXU-decomp}.

Finally, to prove injectivity, note that \(q\ge 1\), so for every \(i\in[M]\),
\[
\mu_{X_U}(\{z_i\})
=
\begin{cases}
q/n, & i\in U,\\[1mm]
0,   & i\notin U.
\end{cases}
\]
Thus the set \(U\) can be recovered uniquely from \(\mu_{X_U}\), and the map
\(U\mapsto \mu_{X_U}\) is injective.
\end{proof}

The family \(\{X_U:U\in\mathcal U\}\) constructed above will serve as the hard family in the proof of
the lower bound. The next subsection shows that the empirical measures \(\mu_{X_U}\) are pairwise
well separated in the \(W_1\)-metric.

\subsection{Wasserstein separation of the hard family}\label{subsec:lb-separation}

We now show that the empirical measures constructed in
Section~\ref{subsec:lb-construction} form a large packing in the \(W_1\)-metric.
This is the geometric core of the lower bound: different codewords \(U\in\mathcal U\)
produce empirical measures that are uniformly separated by a distance of order \(t\).

Recall that for each \(U\in\mathcal U\),
$
\mu_U=\frac{2}{M}\sum_{i\in U}\delta_{z_i},
\,\,
\mu_{X_U}=p\,\mu_U+(1-p)\delta_{z_\star},
$
where \(p\ge \frac12\) by Lemma~\ref{lem:hard-family}, and where the points
$
z_1,\dots,z_M,z_\star
$
are pairwise \(t\)-separated in \(\|\cdot\|_\infty\).

\begin{proposition}[Pairwise \(W_1\)-separation]\label{prop:W1-sep}
For all distinct \(U,V\in\mathcal U\),
$
W_1(\mu_{X_U},\mu_{X_V})
\ \ge\
\frac{\delta}{2}\,t.
$
Equivalently, setting
\begin{equation}\label{eq:t-sep-def}
t_{\mathrm{sep}}:=\frac{\delta}{2}\,t,
\end{equation}
the family
$
M_0:=\{\mu_{X_U}:U\in\mathcal U\}\subset \mathcal P(T)
$
is \(t_{\mathrm{sep}}\)-separated in the \(W_1\)-metric.
\end{proposition}

\begin{proof}
Fix distinct \(U,V\in\mathcal U\). By Lemma~\ref{lem:hard-family},
\begin{equation}\label{eq:decomp-UV}
\mu_{X_U}=p\,\mu_U+(1-p)\delta_{z_\star},
\qquad
\mu_{X_V}=p\,\mu_V+(1-p)\delta_{z_\star}.
\end{equation}

\paragraph{Step 1: removing the common filler mass.}
We claim that
\begin{equation}\label{eq:common-mass-cancel}
W_1(\mu_{X_U},\mu_{X_V})=p\,W_1(\mu_U,\mu_V).
\end{equation}
Indeed, both measures in \eqref{eq:decomp-UV} are probability measures on the bounded metric space
\(T=[0,1]^d\), so Kantorovich--Rubinstein duality gives
\[
W_1(\mu,\nu)=\sup_{\operatorname{Lip}(f)\le 1}\left(\int f\,d\mu-\int f\,d\nu\right).
\]
Applying this to \(\mu_{X_U}\) and \(\mu_{X_V}\), and using \eqref{eq:decomp-UV}, we obtain
\[
\int f\,d\mu_{X_U}-\int f\,d\mu_{X_V}
=
p\left(\int f\,d\mu_U-\int f\,d\mu_V\right)
+
(1-p)\bigl(f(z_\star)-f(z_\star)\bigr),
\]
so the second term vanishes and
\[
\int f\,d\mu_{X_U}-\int f\,d\mu_{X_V}
=
p\left(\int f\,d\mu_U-\int f\,d\mu_V\right).
\]
Taking the supremum over all \(1\)-Lipschitz \(f\) gives \eqref{eq:common-mass-cancel}.

\paragraph{Step 2: lower-bounding \(W_1(\mu_U,\mu_V)\) by total variation.}
The measures \(\mu_U\) and \(\mu_V\) are supported on the set
$
\{z_1,\dots,z_M\},
$
which is \(t\)-separated in \(\|\cdot\|_\infty\). Therefore
Lemma~\ref{lem:sep-TV} implies
\begin{equation}\label{eq:sep-TV-UV}
W_1(\mu_U,\mu_V)\ge t\,\mathrm{TV}(\mu_U,\mu_V).
\end{equation}

\paragraph{Step 3: computing the total variation exactly.}
For each \(i\in[M]\),
$
\mu_U(\{z_i\})=\frac{2}{M}\mathbf 1_{\{i\in U\}},
\,\,
\mu_V(\{z_i\})=\frac{2}{M}\mathbf 1_{\{i\in V\}}.
$
Hence \(\mu_U\) and \(\mu_V\) differ only on the symmetric difference \(U\triangle V\), and at
every point \(z_i\) with \(i\in U\triangle V\) the absolute difference of the masses equals \(2/M\).
Therefore
$
\|\mu_U-\mu_V\|_1=\frac{2}{M}|U\triangle V|.
$
Using the convention
$
\mathrm{TV}(\mu,\nu)=\frac12\|\mu-\nu\|_1,
$
we conclude that
\begin{equation}\label{eq:TV-code}
\mathrm{TV}(\mu_U,\mu_V)=\frac{|U\triangle V|}{M}.
\end{equation}

Since \(U\neq V\) and \(U,V\in\mathcal U\), the code separation property
\eqref{eq:U-family} gives
$
|U\triangle V|\ge \delta M.
$
Combining this with \eqref{eq:TV-code}, we get
$
\mathrm{TV}(\mu_U,\mu_V)\ge \delta.
$
Substituting into \eqref{eq:sep-TV-UV} yields
\[
W_1(\mu_U,\mu_V)\ge \delta\,t.
\]

\paragraph{Step 4: concluding the lower bound for \(\mu_{X_U}\) and \(\mu_{X_V}\).}
Combining the previous estimate with \eqref{eq:common-mass-cancel}, we obtain
\[
W_1(\mu_{X_U},\mu_{X_V})
=
p\,W_1(\mu_U,\mu_V)
\ge
p\,\delta\,t.
\]
Finally, Lemma~\ref{lem:hard-family} gives \(p\ge \frac12\), so
\[
W_1(\mu_{X_U},\mu_{X_V})
\ge
\frac{\delta}{2}\,t.
\]
This proves the proposition.
\end{proof}

As an immediate consequence, the hard family has exponentially large packing number in the
\(W_1\)-metric.

\begin{corollary}[Packing number of the hard family]\label{cor:M0-pack}
Let \(t_{\mathrm{sep}}\) be defined by \eqref{eq:t-sep-def}. Then
\[
N_{\mathrm{pack}}(M_0,W_1,t_{\mathrm{sep}})
\ \ge\
|M_0|
=
|\mathcal U|
>
2e^{\varepsilon n}.
\]
\end{corollary}

\begin{proof}
By Proposition~\ref{prop:W1-sep}, distinct points of \(M_0\) are separated by at least
\(t_{\mathrm{sep}}\), so
\[
N_{\mathrm{pack}}(M_0,W_1,t_{\mathrm{sep}})\ge |M_0|.
\]
By Lemma~\ref{lem:hard-family}, the map \(U\mapsto \mu_{X_U}\) is injective, hence
$
|M_0|=|\mathcal U|.
$
Finally, from \eqref{eq:U-family} and the choice of \(M\), we have
\[
|\mathcal U|
\ge
e^{\beta M}
\ge
e^{\varepsilon n+\ln 2+1}
>
2e^{\varepsilon n}.
\]
This completes the proof.
\end{proof}

The role of Proposition~\ref{prop:W1-sep} is simple but fundamental: it converts the combinatorial
separation of the code family \(\mathcal U\) into geometric separation in Wasserstein distance.
The next subsection uses group privacy to turn the original \(\varepsilon\)-differentially private
mechanism into a metrically private algorithm on the finite set \(M_0\), at which point
Proposition~\ref{prop:master} applies directly.

\subsection{From differential privacy to metric privacy on the hard family}
\label{subsec:lb-privacy-reduction}

We now restrict the original mechanism to the finite hard family constructed in
Section~\ref{subsec:lb-construction}. The point of this step is to convert ordinary
\(\varepsilon\)-differential privacy on datasets into metric privacy on the set
\[
M_0=\{\mu_{X_U}:U\in\mathcal U\},
\]
so that Proposition~\ref{prop:master} can be applied.

\paragraph{The ambient output space.}
Let
$
M_1:=\mathcal P(T),
\,\,
\rho_1:=W_1,
$
where \(\mathcal P(T)\) denotes the set of Borel probability measures on \(T\).
By Corollary~\ref{cor:M0-pack}, the subset
\[
M_0:=\{\mu_{X_U}:U\in\mathcal U\}\subset M_1
\]
has packing number larger than \(2e^{\varepsilon n}\) at scale
$
t_{\mathrm{sep}}=\frac{\delta}{2}\,t.
$

\paragraph{A metric on the hard family.}
Recall from Section~\ref{subsec:lb-construction} that each codeword \(U\in\mathcal U\)
produces an \emph{ordered} dataset \(X_U\in S^n\), and that the map
$
U\longmapsto \mu_{X_U}
$
is injective by Lemma~\ref{lem:hard-family}. Therefore every point of \(M_0\) has a unique
representative of the form \(\mu_{X_U}\), and it makes sense to define
\begin{equation}\label{eq:rho0-def}
\rho_0(\mu_{X_U},\mu_{X_V})
:=
\frac1n\,d_H(X_U,X_V),
\qquad U,V\in\mathcal U,
\end{equation}
where \(d_H\) denotes Hamming distance on ordered datasets.

\begin{lemma}[Basic properties of \(\rho_0\)]\label{lem:rho0-metric}
The function \(\rho_0\) defined in \eqref{eq:rho0-def} is a metric on \(M_0\), and
\[
\operatorname{diam}(M_0,\rho_0)\le 1.
\]
\end{lemma}

\begin{proof}
Since \(d_H\) is a metric on the finite set of ordered datasets \(\{X_U:U\in\mathcal U\}\),
and the map \(U\mapsto \mu_{X_U}\) is injective, the formula \eqref{eq:rho0-def}
defines a metric on \(M_0\). Explicitly:

\begin{itemize}
\item nonnegativity and symmetry are immediate from those of \(d_H\);
\item \(\rho_0(\mu_{X_U},\mu_{X_V})=0\) implies \(d_H(X_U,X_V)=0\), hence \(X_U=X_V\),
so \(\mu_{X_U}=\mu_{X_V}\);
\item the triangle inequality follows from
\[
d_H(X_U,X_W)\le d_H(X_U,X_V)+d_H(X_V,X_W).
\]
\end{itemize}

Finally, because \(X_U,X_V\in S^n\) are ordered datasets of length \(n\), one always has
$
d_H(X_U,X_V)\le n.
$
Dividing by \(n\) yields
$
\rho_0(\mu_{X_U},\mu_{X_V})\le 1
\,\,\text{for all }U,V\in\mathcal U,
$
which proves \(\operatorname{diam}(M_0,\rho_0)\le 1\).
\end{proof}

\paragraph{Restricting the mechanism.}
Let
\[
\mathcal M:S^n\to \bigsqcup_{m\ge 1} T^m
\]
be any \(\varepsilon\)-differentially private mechanism. We post-process its output by replacing
the synthetic dataset with its empirical measure:
$
\bar{\mathcal M}(X):=\mu_{\mathcal M(X)}\in \mathcal P(T).
$
Since post-processing preserves differential privacy, the mechanism \(\bar{\mathcal M}\) is still
\(\varepsilon\)-differentially private.

We now restrict \(\bar{\mathcal M}\) to the hard family \(M_0\) by defining
\begin{equation}\label{eq:A0-def}
A_0(\mu_{X_U}) := \bar{\mathcal M}(X_U)=\mu_{\mathcal M(X_U)},
\qquad U\in\mathcal U.
\end{equation}
Because the representation \(\mu_{X_U}\leftrightarrow U\) is unique, this defines a
well-posed randomized map
\[
A_0:M_0\to M_1.
\]

The next lemma is the key privacy reduction.

\begin{lemma}[Metric privacy on the hard family]\label{lem:A0-metric-private}
The restricted mechanism \(A_0:M_0\to M_1\) is \((\varepsilon n)\)-metrically private
with respect to \(\rho_0\). Equivalently, for every \(U,V\in\mathcal U\) and every measurable
set \(B\subset M_1\),
\[
\Pr\bigl(A_0(\mu_{X_U})\in B\bigr)
\le
\exp\!\bigl((\varepsilon n)\rho_0(\mu_{X_U},\mu_{X_V})\bigr)\,
\Pr\bigl(A_0(\mu_{X_V})\in B\bigr).
\]
\end{lemma}

\begin{proof}
We begin with the standard group-privacy consequence of pure differential privacy:
if \(X,X'\in S^n\) differ in exactly \(h\) coordinates, then for every measurable event \(E\),
\begin{equation}\label{eq:group-privacy}
\Pr\bigl(\bar{\mathcal M}(X)\in E\bigr)
\le
e^{\varepsilon h}\,
\Pr\bigl(\bar{\mathcal M}(X')\in E\bigr).
\end{equation}
Indeed, one can connect \(X\) to \(X'\) by a chain of \(h=d_H(X,X')\) ordered datasets
\[
X=X^{(0)},X^{(1)},\dots,X^{(h)}=X'
\]
such that each consecutive pair differs in exactly one coordinate. Applying \(\varepsilon\)-DP
successively along this chain gives
\[
\Pr\bigl(\bar{\mathcal M}(X)\in E\bigr)
\le
e^{\varepsilon h}\Pr\bigl(\bar{\mathcal M}(X')\in E\bigr),
\]
which is \eqref{eq:group-privacy}.

Now fix \(U,V\in\mathcal U\) and a measurable set \(B\subset M_1\). Using the definition
\eqref{eq:A0-def} and then \eqref{eq:group-privacy} with \(X=X_U\) and \(X'=X_V\), we get
\[
\Pr\bigl(A_0(\mu_{X_U})\in B\bigr)
=
\Pr\bigl(\bar{\mathcal M}(X_U)\in B\bigr)
\le
e^{\varepsilon d_H(X_U,X_V)}
\Pr\bigl(\bar{\mathcal M}(X_V)\in B\bigr).
\]
Applying \eqref{eq:A0-def} again,
\[
\Pr\bigl(A_0(\mu_{X_U})\in B\bigr)
\le
e^{\varepsilon d_H(X_U,X_V)}
\Pr\bigl(A_0(\mu_{X_V})\in B\bigr).
\]
Finally, by the definition of \(\rho_0\),
$
d_H(X_U,X_V)=n\,\rho_0(\mu_{X_U},\mu_{X_V}),
$
so the preceding inequality becomes
\[
\Pr\bigl(A_0(\mu_{X_U})\in B\bigr)
\le
\exp\!\bigl((\varepsilon n)\rho_0(\mu_{X_U},\mu_{X_V})\bigr)\,
\Pr\bigl(A_0(\mu_{X_V})\in B\bigr),
\]
which is exactly \((\varepsilon n)\)-metric privacy.
\end{proof}

Combining Lemma~\ref{lem:rho0-metric}, Lemma~\ref{lem:A0-metric-private}, and
Corollary~\ref{cor:M0-pack}, we have now verified all assumptions of the master lower bound
(Proposition~\ref{prop:master}) with
\[
\alpha=\varepsilon n,
\qquad
(M_1,\rho_1)=\bigl(\mathcal P(T),W_1\bigr),
\qquad
M_0=\{\mu_{X_U}:U\in\mathcal U\},
\qquad
t=t_{\mathrm{sep}}.
\]
The proof of Theorem~\ref{thm:lb} is therefore reduced to a single application of
Proposition~\ref{prop:master} followed by the simplification of the scale \(t\), which we carry
out in the next subsection.

\subsection{Proof of the lower bound}\label{subsec:lb-proof}

We now combine the hard-family construction, the Wasserstein separation estimate, and the
metric-privacy reduction to prove Theorem~\ref{thm:lb}.

\begin{proof}[Proof of Theorem~\ref{thm:lb}]
Fix \(n\ge n_0\), where \(n_0\) is as in Lemma~\ref{lem:Mle2n}, and let
\[
\mathcal M:S^n\to \bigsqcup_{m\ge 1}T^m
\]
be any \(\varepsilon\)-differentially private mechanism.

By Corollary~\ref{cor:M0-pack}, the hard family
$
M_0=\{\mu_{X_U}:U\in\mathcal U\}\subset \mathcal P(T)
$
satisfies
\[
N_{\mathrm{pack}}(M_0,W_1,t_{\mathrm{sep}})\ >\ 2e^{\varepsilon n},
\qquad
t_{\mathrm{sep}}=\frac{\delta}{2}\,t.
\]
By Lemma~\ref{lem:rho0-metric}, the metric \(\rho_0\) on \(M_0\) has diameter at most \(1\), and by
Lemma~\ref{lem:A0-metric-private}, the restricted post-processed mechanism
\[
A_0:M_0\to \mathcal P(T),
\qquad
A_0(\mu_{X_U})=\mu_{\mathcal M(X_U)},
\]
is \((\varepsilon n)\)-metrically private with respect to \(\rho_0\).

Therefore Proposition~\ref{prop:master} applies with
$
(M_1,\rho_1)=\bigl(\mathcal P(T),W_1\bigr),
\,
\alpha=\varepsilon n,
\,
t=t_{\mathrm{sep}}.
$
We conclude that there exists some \(U\in\mathcal U\) such that
\[
\mathbb E\,W_1\!\bigl(A_0(\mu_{X_U}),\mu_{X_U}\bigr)
\ >\
\frac{t_{\mathrm{sep}}}{4}.
\]
Recalling the definition of \(A_0\), this becomes
\[
\mathbb E\,W_1\!\bigl(\mu_{\mathcal M(X_U)},\mu_{X_U}\bigr)
\ >\
\frac{t_{\mathrm{sep}}}{4}
=
\frac{\delta}{8}\,t.
\]
Since \(X_U\in S^n\), we obtain the lower bound
\begin{equation}\label{eq:pre-final-lb}
\sup_{X\in S^n}\mathbb E\,W_1\!\bigl(\mu_X,\mu_{\mathcal M(X)}\bigr)
\ \ge\
\frac{\delta}{8}\,t.
\end{equation}

It remains to simplify the scale \(t\). By definition,
\[
t=\min\left\{
t_0,\,
\left(\frac{c_0}{2(M+1)}\right)^{1/k}
\right\}.
\]
From \eqref{eq:M-upper},
$
M+1
\le
\max\left\{
M_\star,\,
\frac{\varepsilon n+\ln 2+1}{\beta}
\right\}+3.
$
Since \(\varepsilon n\ge 2\), both terms inside the maximum are bounded by a constant multiple of
\(\varepsilon n\). More precisely,
$
M_\star+3 \le \frac{M_\star+3}{2}\,\varepsilon n,
$
and
$
\frac{\varepsilon n+\ln 2+1}{\beta}+3
\le
\left(
\frac{1}{\beta}
+\frac{\ln 2+1}{2\beta}
+\frac{3}{2}
\right)\varepsilon n.
$
Hence, there exists a constant
\[
C_0:=
\max\left\{
\frac{M_\star+3}{2},\,
\frac{1}{\beta}+\frac{\ln 2+1}{2\beta}+\frac{3}{2}
\right\}
\]
such that
\begin{equation}\label{eq:Mplus1-bound}
M+1\le C_0\,\varepsilon n.
\end{equation}
Here \(C_0\) is universal, since \(\beta\) and \(M_\star\) come from the constant-weight
Lemma~\ref{lem:cw} with the fixed choice \(\delta=1/25\).

Substituting \eqref{eq:Mplus1-bound} into the definition of \(t\) yields
\[
t
\ge
\min\left\{
t_0,\,
\left(\frac{c_0}{2C_0\,\varepsilon n}\right)^{1/k}
\right\}.
\]
Set
$
a:=\left(\frac{c_0}{2C_0}\right)^{1/k},
\,\,
c_1:=\min\{1,a\}.
$
Then for every \(x\ge 0\),
$
\min\{t_0,ax\}\ge c_1\,\min\{t_0,x\},
$
and therefore
$
t
\ge
c_1\,\min\{t_0,(\varepsilon n)^{-1/k}\}.
$
Combining this with \eqref{eq:pre-final-lb}, we obtain
\[
\sup_{X\in S^n}\mathbb E\,W_1\!\bigl(\mu_X,\mu_{\mathcal M(X)}\bigr)
\ \ge\
\frac{\delta c_1}{8}\,
\min\{t_0,(\varepsilon n)^{-1/k}\}.
\]
Thus the theorem holds with
$
c:=\frac{\delta c_1}{8}>0.
$

Finally, the constants \(c\) and \(n_0\) depend only on \(c_0,t_0,k\): the dependence on
\(\beta\), \(M_\star\), and \(\delta\) is universal, while the only non-universal parameter entering the
last estimate is \(c_0\), together with the explicit cutoff \(t_0\) and exponent \(k\).
This completes the proof.
\end{proof}

\subsection{Matching growth on coordinate-cube unions}
\label{subsec:matching-growth}
We prove Proposition~\ref{prop:matching-coordinate-supports}
by bounding occupied-cell counts for the upper bound and
applying Theorem~\ref{thm:packing-lower-bound} for the lower bound.

\begin{proof}[Proof of Proposition~\ref{prop:matching-coordinate-supports}]
\textbf{Occupied cells.}
Use a fixed boundary convention for the binary partition,
assigning each splitting hyperplane to the child with larger
values in the split coordinate.
Write $j=qd+h$ with $0\le h<d$, and set
$a_{\ell,h}=|A_\ell\cap[h]|$, with $a_{\ell,0}=0$.
Among the first $j$ cyclic splits, exactly $qk+a_{\ell,h}$
occur in coordinates belonging to $A_\ell$.
The component $S_\ell$ therefore intersects at most
$2^{qk+a_{\ell,h}}$ depth-$j$ cells.

Since
\[
a_{\ell,h}-\frac{hk}{d}
\le
\min\{h,k\}-\frac{hk}{d}
\le
\frac{k(d-k)}{d},
\]
we obtain, for every $X\in S^n$,
\[
O_j(X)
\le
\min\left\{n,\sum_{\ell=1}^{J}2^{qk+a_{\ell,h}}\right\}
\le
\min\{n,C_{\rm occ}2^{jk/d}\},
\qquad
C_{\rm occ}:=J2^{k(d-k)/d}.
\]
This estimate holds regardless of intersections between
components or the allocation of samples among them.

\paragraph{External packing of the empirical support.}
For every $X\in S^n$, the $t$-neighborhood of $S_X$ is
contained in the union of the $t$-neighborhoods of the
components.
Within $[0,1]^d$, the neighborhood of $S_\ell$ lies in a
rectangle of side length $1$ in its $k$ active coordinates
and at most $2t$ in the remaining coordinates.
A coordinate-bin count with width $t/2$ gives
\[
N_{\rm pack}^{\rm ext}(S_X;t,t/2)
\le
J(1+2/t)^k5^{d-k}
\le
J3^k5^{d-k}t^{-k},
\qquad 0<t\le1.
\]
One may use half-open bins, with singleton terminal bins
when necessary, so that two points in the same bin cannot
be separated by $t/2$.
Thus the external packing-growth condition holds uniformly
over all samples supported on $S$.
We use the sharper occupied-cell estimate above to obtain
the stated dependence on $J$ and $d$.

\paragraph{Public fixed-depth upper bound.}
Set
\[
N=\varepsilon n,\qquad
b=\frac{d}{k-1},\qquad
x=\frac{N}{C_{\rm occ}b^2}.
\]
Choose the public depth
\[
r=
\max\left\{
1,\left\lfloor\frac{d}{k}\log_2(1+x)\right\rfloor
\right\},
\qquad
Q_r=\sum_{j=0}^{r-1}2^{j(k-1)/(2d)},
\]
and noise scales
\[
\sigma_0=\varepsilon^{-1},
\qquad
\sigma_{j+1}
=
\frac{2Q_r}{\varepsilon\,2^{j(k-1)/(2d)}},
\qquad 0\le j<r.
\]
These parameters depend only on $d,k,J,n,\varepsilon$, and
$
2\sum_{i=1}^{r}\frac1{\sigma_i}=\varepsilon.
$
Hence Theorem~\ref{thm:privacy} gives an $\varepsilon$-DP
mechanism on the entire ambient domain.

With the nonempty root convention
$m_\varnothing=\max\{1,n+\lambda_\varnothing\}$, we have
$
|m_\varnothing-n|\le|\lambda_\varnothing|.
$
Corollary~\ref{cor:acc-exp} and the occupied-cell estimate yield
\[
\mathbb E W_1(\mu_X,\mu_Y)
\le
\frac{C_{\rm abs}}{N}
+
\frac{8C_{\rm abs}C_{\rm occ}}{N}Q_r^2
+
2\cdot2^{-r/d}.
\]
For $x\ge1$, the unclamped depth is at least one, and
\[
Q_r
\le
\frac{2^{r(k-1)/(2d)}}{2^{(k-1)/(2d)}-1}
\le
\frac{2b}{\ln2}\,2^{r(k-1)/(2d)}.
\]
Using
\[
\frac{d}{k}\log_2(1+x)-1
\le r\le
\frac{d}{k}\log_2(1+x)
\]
and $1+x\le2x$, we obtain
\[
\frac{C_{\rm occ}}{N}Q_r^2
\le
\frac{4}{(\ln2)^2x}(1+x)^{(k-1)/k}
\le
\frac{8}{(\ln2)^2}x^{-1/k}.
\]
The resolution term satisfies
\[
2\cdot2^{-r/d}
\le
2^{1+1/d}(1+x)^{-1/k}
\le
4x^{-1/k}.
\]
The root term is also absorbed, since $1\le x\le N$.
For $x<1$, the same bound follows from
$W_1\le1\le x^{-1/k}$.
Consequently,
\[
\sup_{X\in S^n}\mathbb E W_1(\mu_X,\mu_Y)
\le
C J^{1/k}2^{(d-k)/d}
\left(\frac{d}{k-1}\right)^{2/k}
N^{-1/k}.
\]
Here $C$ is universal under the discrete-Laplace convention
used in the paper, for which $C_{\rm abs}=1$.
Since $2^{(d-k)/d}\le2$, absorbing this factor into $C$
proves the stated upper bound.

\paragraph{Lower bound.}
For every $0<t\le1$, the first component contains an
isometric copy of the grid
$\{0,t,\ldots,\lfloor1/t\rfloor t\}^k$.
Consequently,
\[
N_{\rm pack}(S,t)
\ge
(\lfloor1/t\rfloor+1)^k
\ge t^{-k}.
\]
Theorem~\ref{thm:packing-lower-bound}, applied with
$c_0=t_0=1$, therefore gives constants $c_k>0$ and
$n_0\in\mathbb N$ depending only on $k$ such that
\[
\sup_{X\in S^n}
\mathbb E W_1(\mu_X,\mu_{\mathcal M(X)})
\ge
c_k(\varepsilon n)^{-1/k}
\]
for every pure $\varepsilon$-DP synthetic-data mechanism
$\mathcal M:S^n\to\bigsqcup_{m\ge1}([0,1]^d)^m$
whenever $n\ge n_0$, $0<\varepsilon\le1$, and
$\varepsilon n\ge4$.
No separation between components is required.
Applying this lower bound to the public Pruned-PMM construction
above establishes the claimed tight rate for fixed $d,k,J$.
\end{proof}

%% file: app_experiment.tex
\section{Additional Experiments}
\label{app:experiments}

This appendix reports the experimental outputs omitted from the main text for
space. Unless stated otherwise, all experiments use pure-DP budget
\(\varepsilon=1\), the \(\ell_\infty\) ground metric, and the same subsampled
Wasserstein evaluation protocol as in the main text.

\paragraph{Adaptive depth selection.}
We first consider \(n=30000\) samples supported on a linear \(k\)-dimensional
coordinate subspace of \([0,1]^d\), for \(k\in\{1,3,5\}\), while the ambient
dimension \(d\) varies. Full PMM uses the ambient-depth rule, whereas Adaptive
Pruned-PMM privately selects a public depth--schedule pair from the standard
candidate set. Figure~\ref{fig:appendix-adaptive-k} shows that adaptation
improves empirical \(W_1\), especially when \(k\ll d\).

\paragraph{Full two-dimensional shape suite.}
We next compare full PMM and Pruned-PMM on the complete two-dimensional shape
suite, using \(n=30000\) and depth \(r=15\). This experiment isolates the effect
of pruning: Pruned-PMM should preserve PMM-like \(W_1\) accuracy while visiting
far fewer nodes and using less runtime. Figures~\ref{fig:appendix-exp1-gallery-a}
and~\ref{fig:appendix-exp1-gallery-b} show that the two outputs are visually
similar across manifold-like, multimodal, grid-like, and piecewise-structured
supports.

\begin{figure}[!htbp]
\centering
\begin{subfigure}[t]{0.31\linewidth}
    \centering
    \includegraphics[width=\linewidth]{figs/em1.png}
    \caption{\(k=1\)}
\end{subfigure}
\hfill
\begin{subfigure}[t]{0.31\linewidth}
    \centering
    \includegraphics[width=\linewidth]{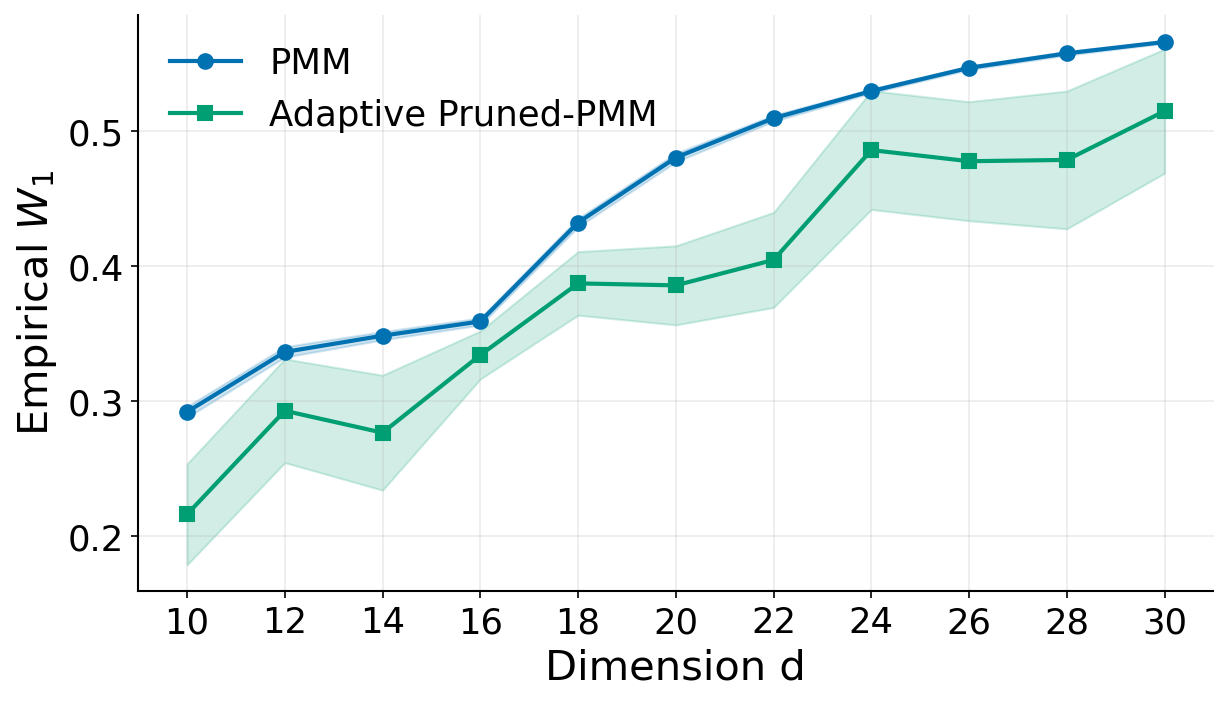}
    \caption{\(k=3\)}
\end{subfigure}
\hfill
\begin{subfigure}[t]{0.31\linewidth}
    \centering
    \includegraphics[width=\linewidth]{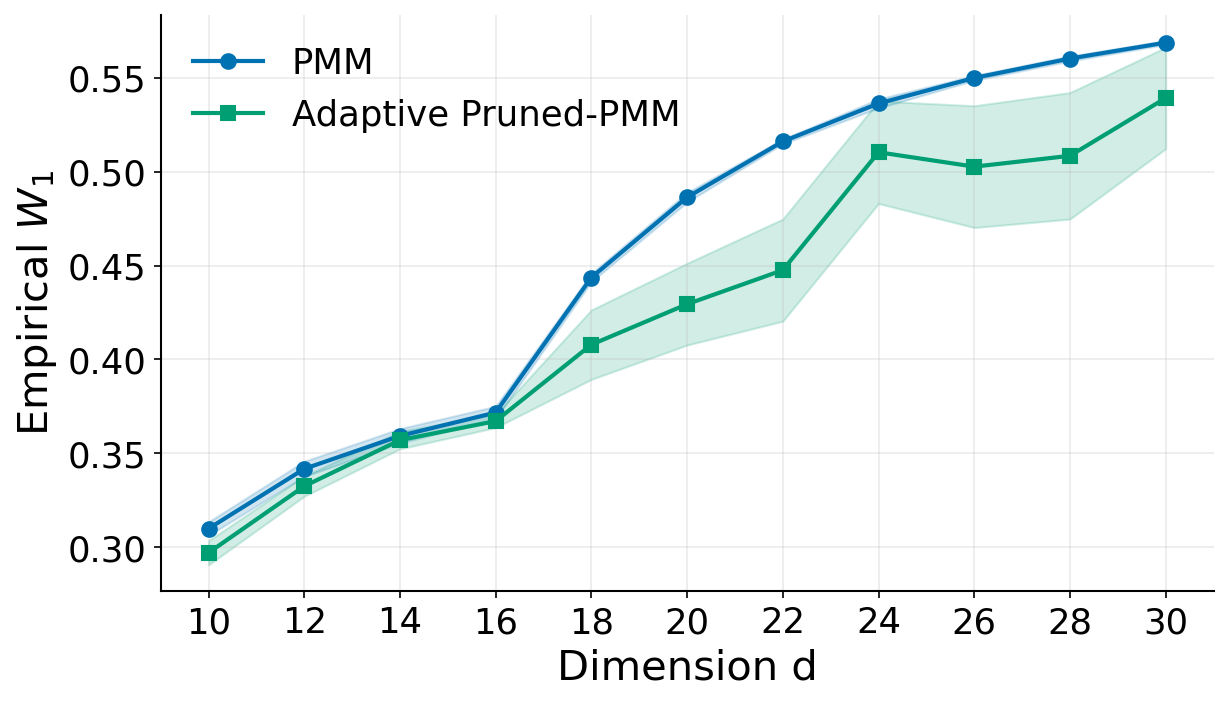}
    \caption{\(k=5\)}
\end{subfigure}
\vspace{-0.4em}
\caption{
Adaptive depth selection on \(n=30000\) samples from linear \(k\)-dimensional
coordinate subspaces of \([0,1]^d\).
}
\label{fig:appendix-adaptive-k}
\end{figure}

\begin{figure}[!htbp]
\centering
\begin{subfigure}[t]{0.31\linewidth}
    \centering
    \includegraphics[width=\linewidth]{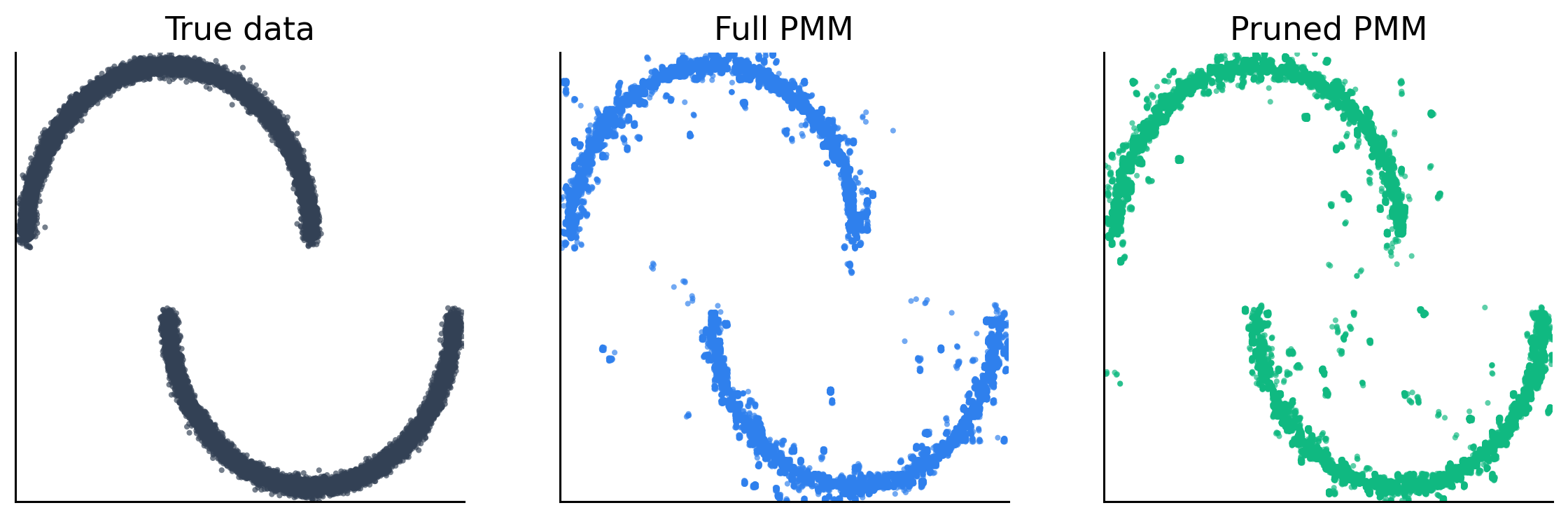}
    \caption{Two moons}
\end{subfigure}
\hfill
\begin{subfigure}[t]{0.31\linewidth}
    \centering
    \includegraphics[width=\linewidth]{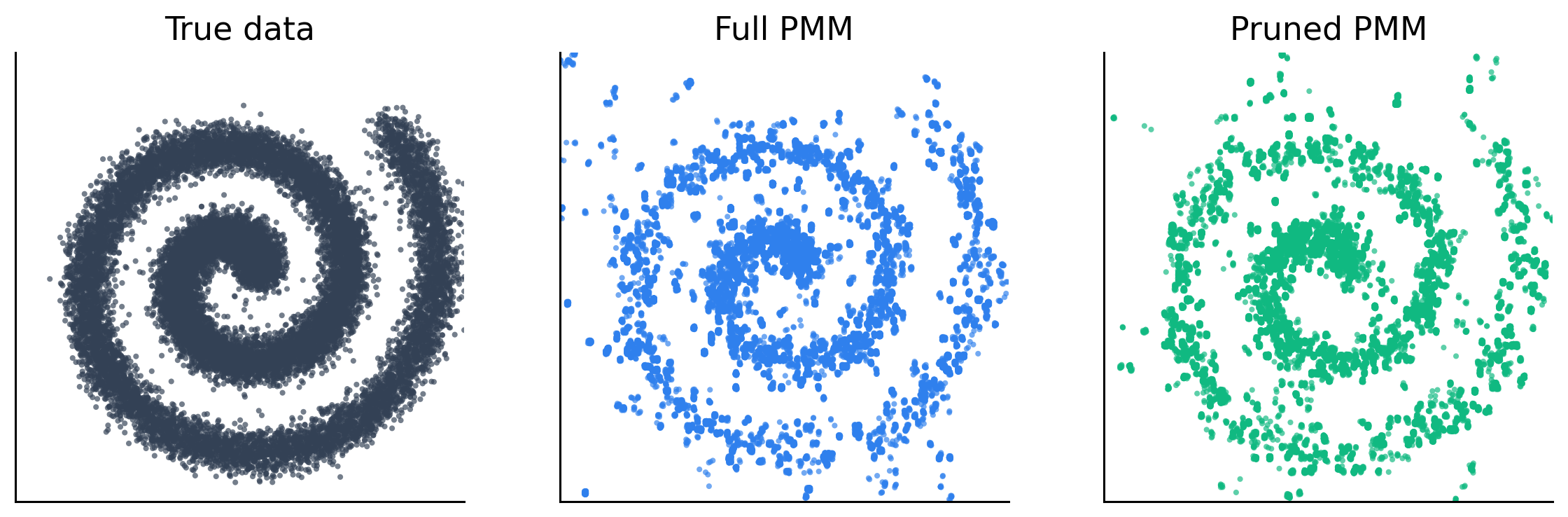}
    \caption{Spiral}
\end{subfigure}
\hfill
\begin{subfigure}[t]{0.31\linewidth}
    \centering
    \includegraphics[width=\linewidth]{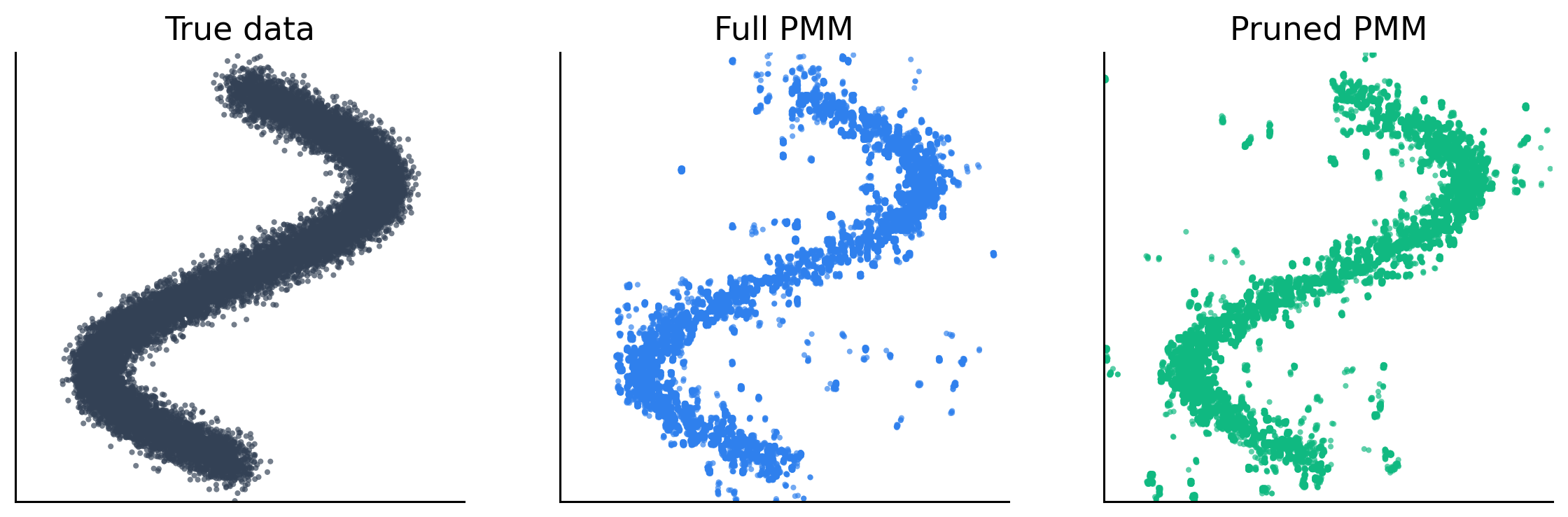}
    \caption{S-curve}
\end{subfigure}

\vspace{0.4em}

\begin{subfigure}[t]{0.31\linewidth}
    \centering
    \includegraphics[width=\linewidth]{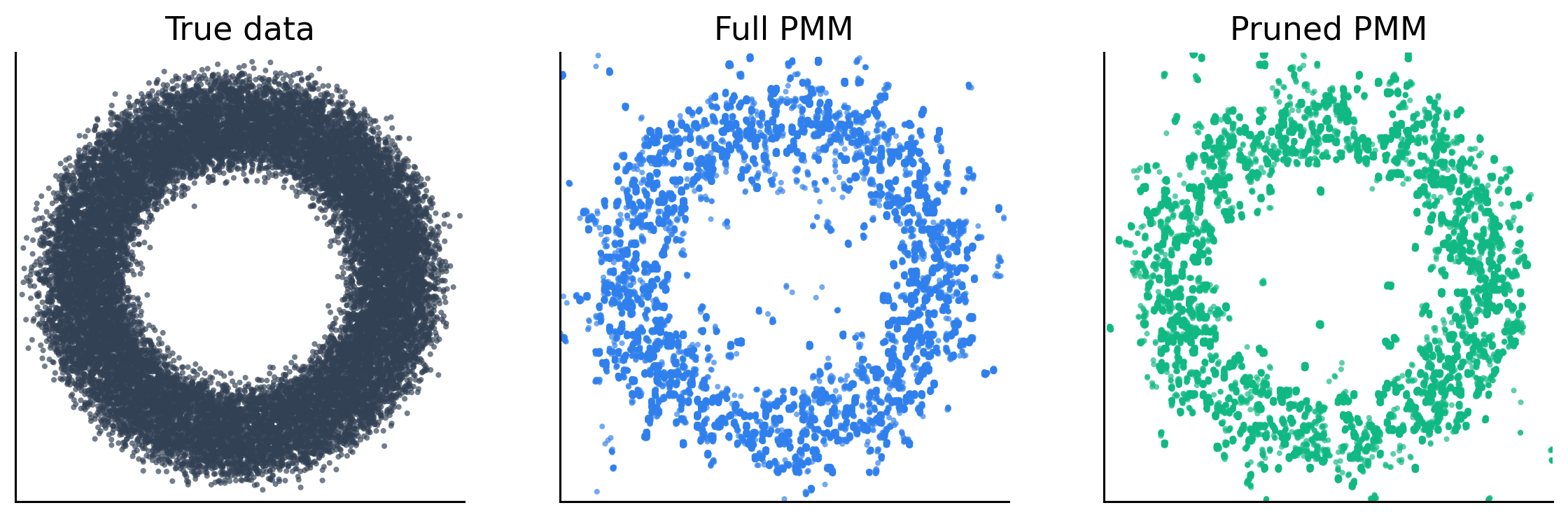}
    \caption{Annulus}
\end{subfigure}
\hfill
\begin{subfigure}[t]{0.31\linewidth}
    \centering
    \includegraphics[width=\linewidth]{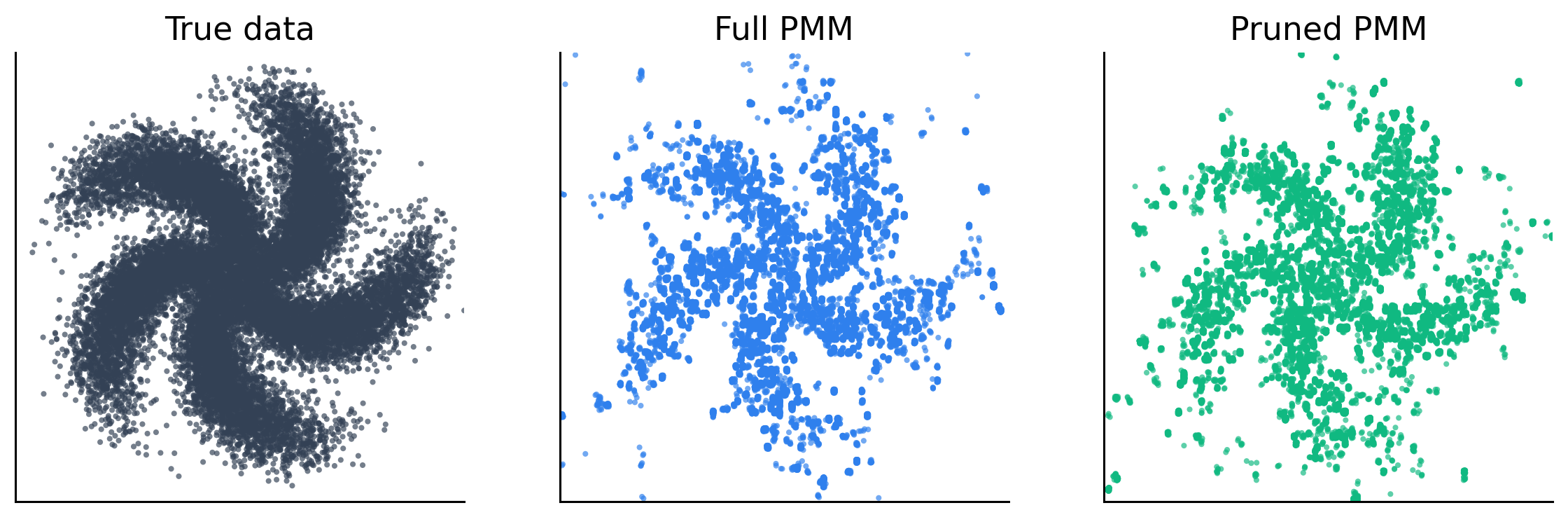}
    \caption{Pinwheel}
\end{subfigure}
\hfill
\begin{subfigure}[t]{0.31\linewidth}
    \centering
    \includegraphics[width=\linewidth]{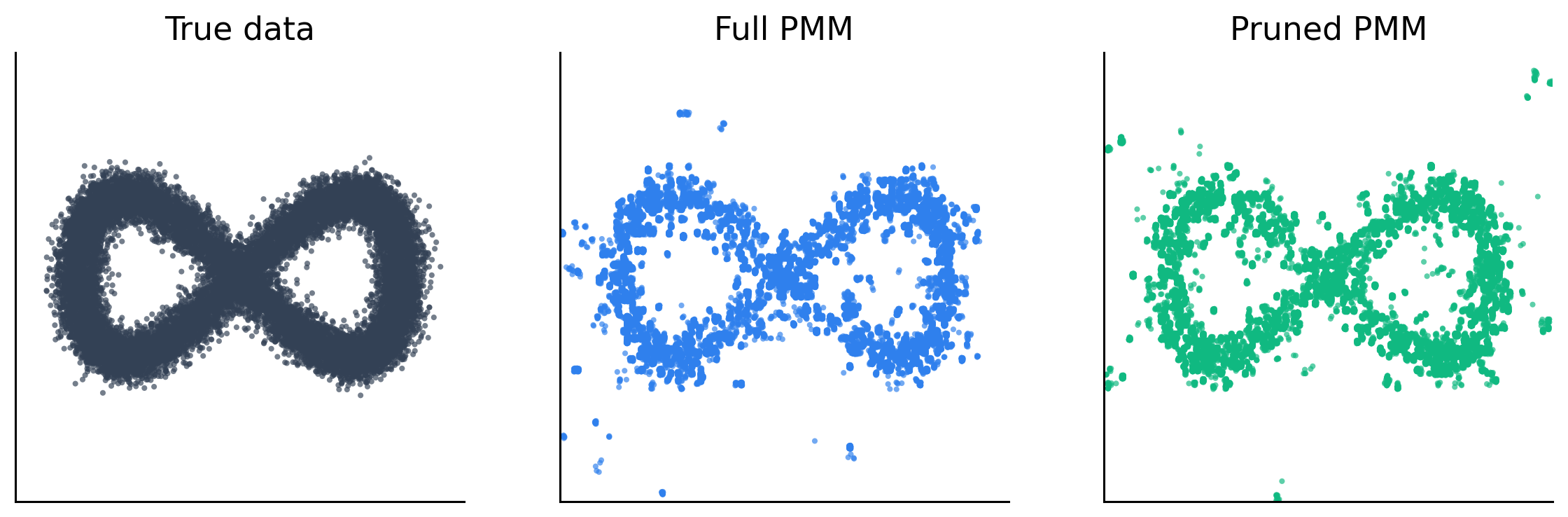}
    \caption{Figure eight}
\end{subfigure}
\vspace{-0.4em}
\caption{
Two-dimensional shape suite, part I. Pruned-PMM remains visually close to full
PMM on curved and manifold-like supports.
}
\label{fig:appendix-exp1-gallery-a}
\end{figure}

\begin{figure}[!htbp]
\centering
\begin{subfigure}[t]{0.31\linewidth}
    \centering
    \includegraphics[width=\linewidth]{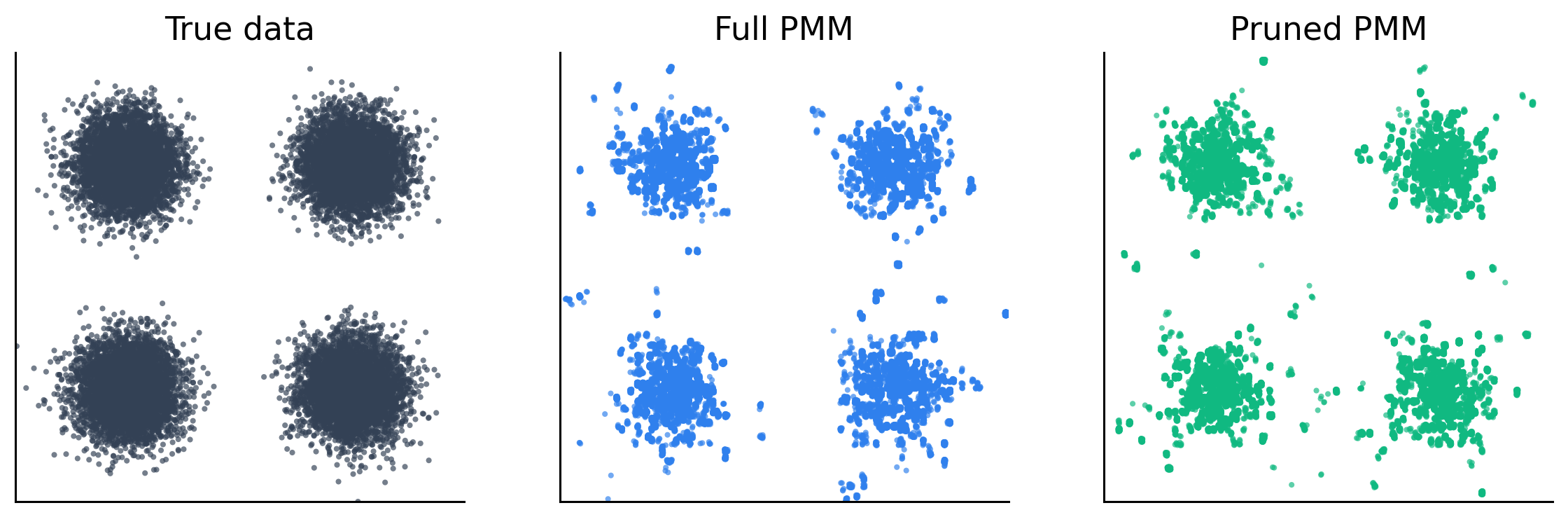}
    \caption{Four blobs}
\end{subfigure}
\hfill
\begin{subfigure}[t]{0.31\linewidth}
    \centering
    \includegraphics[width=\linewidth]{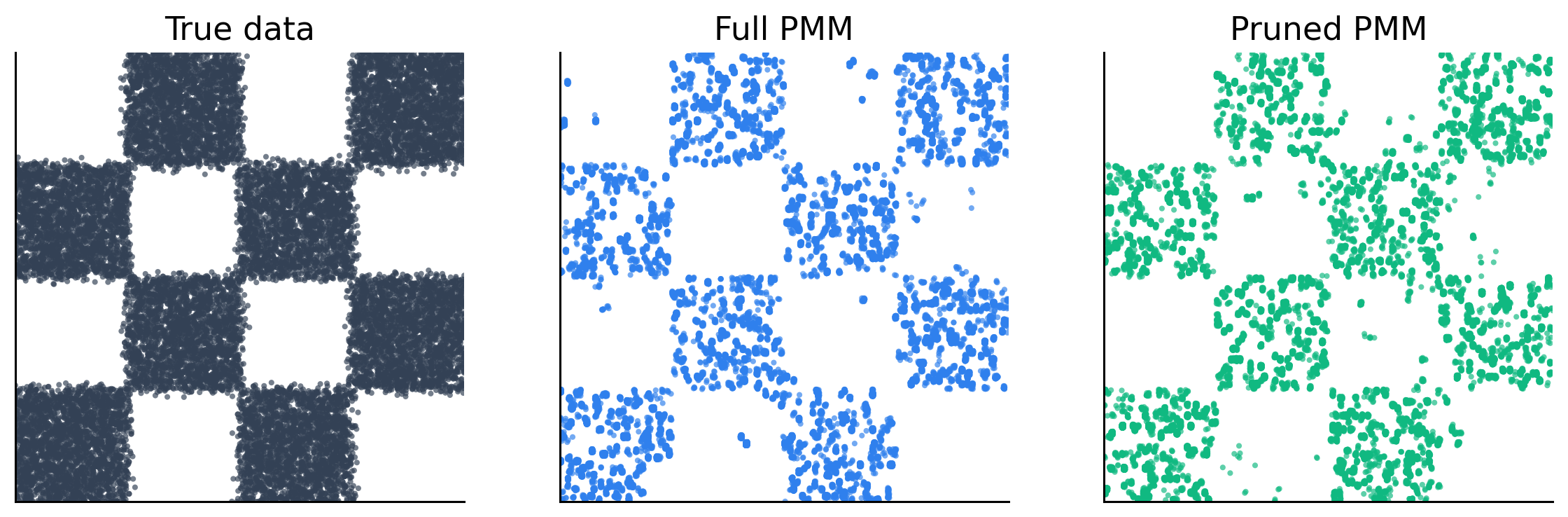}
    \caption{Checkerboard}
\end{subfigure}
\hfill
\begin{subfigure}[t]{0.31\linewidth}
    \centering
    \includegraphics[width=\linewidth]{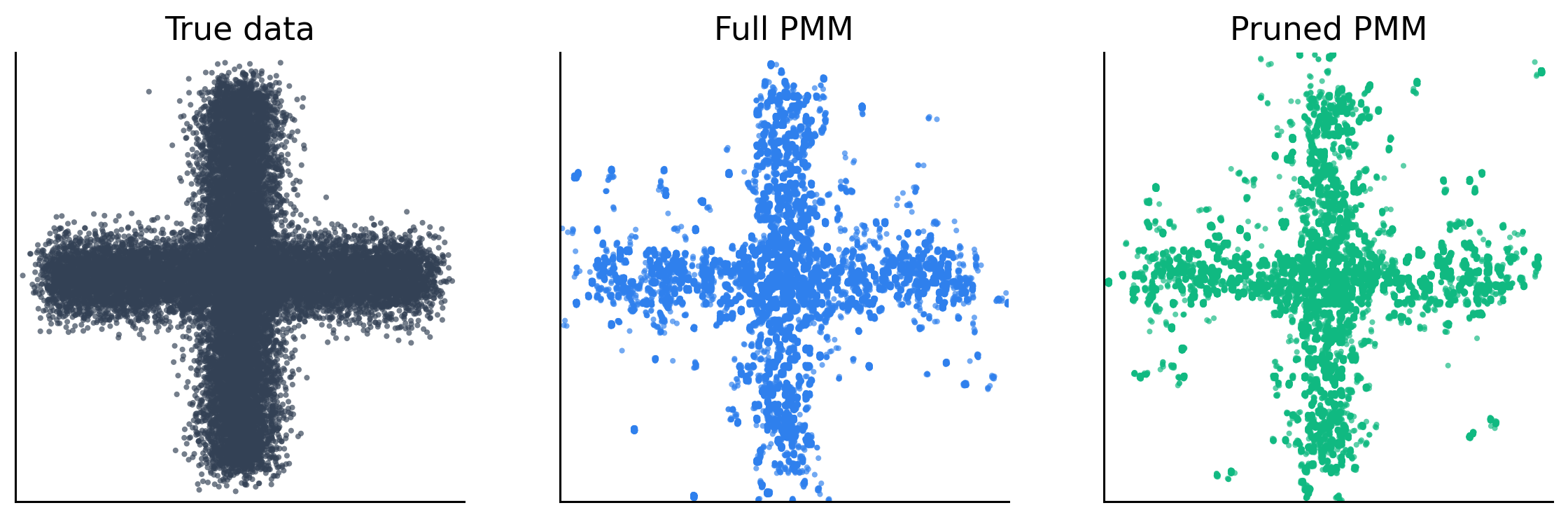}
    \caption{Cross}
\end{subfigure}
\vspace{-0.4em}
\caption{
Two-dimensional shape suite, part II. Pruned-PMM remains close to full PMM on
multimodal, grid-like, and piecewise-structured supports.
}
\label{fig:appendix-exp1-gallery-b}
\end{figure}